%% file: root.tex
\RequirePackage{xcolor}

\documentclass[journal,twoside,web]{ieeecolor}
\usepackage{generic}  %

\usepackage{amsmath,,mathtools}

\usepackage{tikz}
\usetikzlibrary{automata,positioning, shapes,arrows,patterns,graphs, calc}
\usepackage{pgfplots}
\usepackage{xcolor}

\usepackage{amsthm}
\newtheorem{theorem}{Theorem}

\newtheorem{assumption}{Assumption}
\newtheorem{proposition}{Proposition}
\newtheorem{corollary}{Corollary}
\newtheorem{remark}{Remark}
\newtheorem{lemma}{Lemma}
\usepackage{amsmath} 
\usepackage{dsfont}
\usepackage{bm}
\usepackage{subfig}
\usepackage{url}
\usepackage{amssymb}  %
\usepackage{nicefrac}
\usepackage{tabularx}
\usepackage{xcolor}
\definecolor{matlabblue}{HTML}{0072BD}
\definecolor{matlabred}{HTML}{A2142F}
\definecolor{matlabyellow}{HTML}{EDB120}
\definecolor{matlabpurple}{HTML}{7E2F8E}
\definecolor{matlabgreen}{HTML}{77AC30}
\definecolor{matlaborange}{HTML}{D95319}
\definecolor{mycolor}{rgb}{1,0,0}

\usepackage{multirow}
\input{commands}

\title{\LARGE \bf
When Should a Leader Act Suboptimally? \\
The Role of Inferability in Repeated Stackelberg Games
}

\author{Mustafa O. Karabag$^{1}$, Sophia Smith$^{1}$,  Negar Mehr$^{2}$, David Fridovich-Keil$^{1}$, and Ufuk Topcu$^{1}$
\thanks{$^{1}$The University of Texas at Austin.}
\thanks{$^{2}$University of California, Berkeley.}
\thanks{Correspondence to karabag@utexas.edu.}%
}

\begin{document}

\maketitle
\thispagestyle{empty}
\pagestyle{empty}

\begin{abstract}
When interacting with other decision-making agents in non-adversarial scenarios, it is critical for an autonomous agent to have inferable behavior: The agent's actions must convey their intention and strategy. We model the inferability problem using Stackelberg games with observations where a leader and a follower repeatedly interact. During the interactions, the leader uses a fixed mixed strategy. The follower does not know the leader's strategy and dynamically reacts to the statistically inferred strategy based on the leader's previous actions. In the inference setting, the leader may have a lower performance compared to the setting where the follower has full information on the leader's strategy. We refer to the performance gap between these settings as the inferability gap. For a variety of game settings, we show that the inferability gap is upper-bounded by a function of the number of interactions and the stochasticity level of the leader's strategy, encouraging the use of inferable strategies with lower stochasticity levels. We also analyze bimatrix Stackelberg games and identify a set of games where the leader's near-optimal strategy may potentially suffer from a large inferability gap.
\end{abstract}

\input{intro}

\input{prelims}

\input{problem_formulation}

\input{performancebounds}

\input{achievabilityresult}
\input{converseresult}

\input{limitedgap}

\input{numerical_examples}

\input{human_subject_study}

\section{Conclusions}
When interacting with other non-competitive agents, an agent should have an inferable behavior to inform others about its intentions effectively. We model the inferability problem using repeated Stackelberg games where the follower infers the leader's strategy via observation from previous interactions. For a variety of repeated Stackelberg game settings, we show that in the inference setting, the leader may suffer from an inferability gap compared to the setting where the follower has full information of the leader's strategy. However, this gap is upper bounded by a function that depends on the stochasticity level of the leader's strategy. The bounds and the results for the experiments and the human subject study show that to maximize the transient returns, the leader may be better off using a less stochastic strategy compared to the strategy that is optimal in the full information setting.

\section*{Appendix: Proofs of the Technical Results}

\input{appendix}

\section*{Acknowledgement}
This work was supported in part by the National Science Foundation (NSF), under grant numbers 1836900, 2145134, 2211432, 2211548, 2336840, and 2423131.
The authors would like to thank Tichakorn Wongpiromsarn
for insightful discussions that motivated the considered problem. The authors would like to thank Neel Bhatt for his help in creating the simulation environment for the human subject study.

\section*{References}
\bibliographystyle{unsrt}

\bibliography{ref}

\begin{IEEEbiography}[{\includegraphics[width=1in,height=1.25in,clip,keepaspectratio]{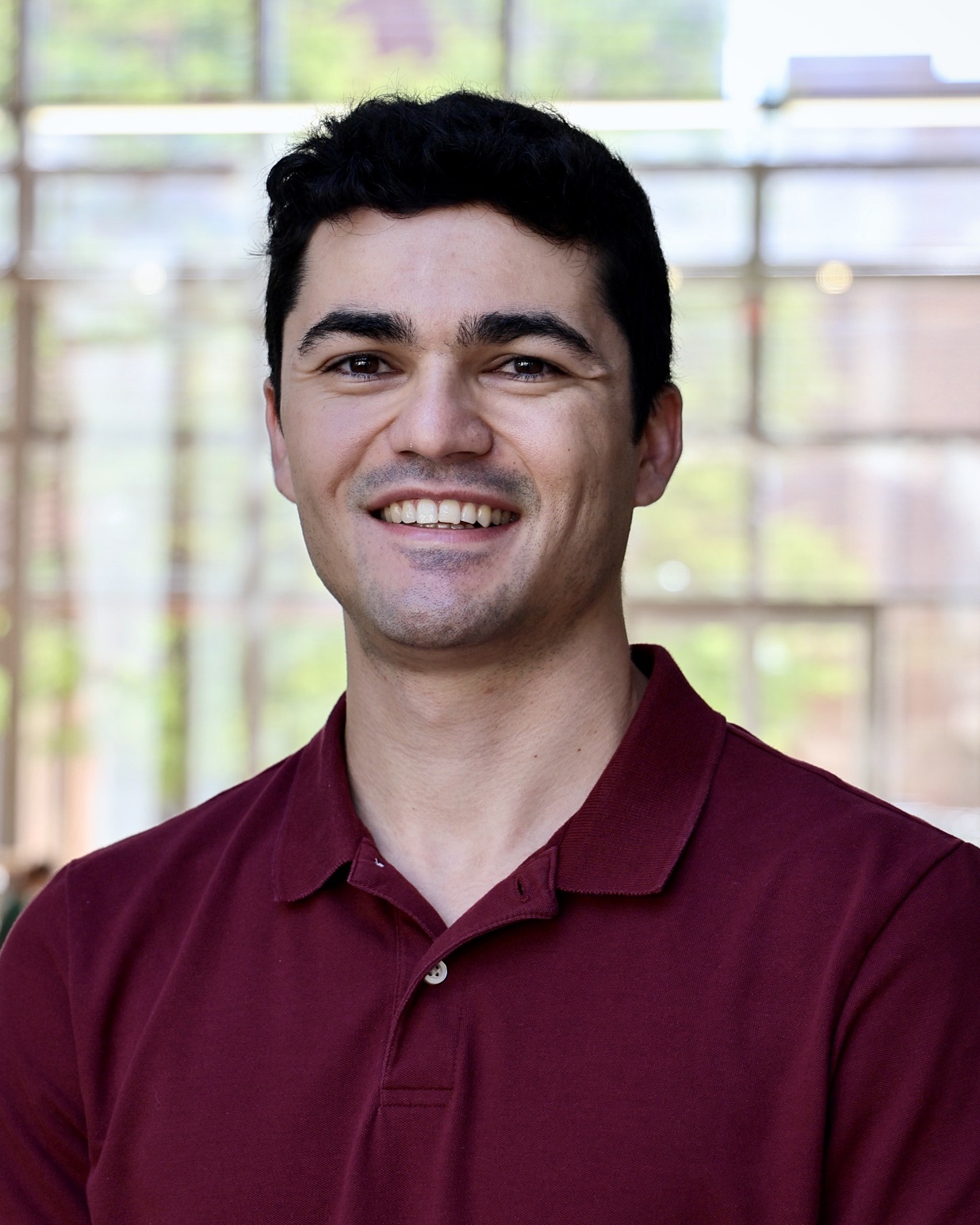}}]{Mustafa O. Karabag} is a postdoctoral fellow in the Oden Institute for Computational Engineering \& Sciences at the University of Texas at Austin. He received his Ph.D. degree from the University of Texas at Austin in 2023. His research focuses on developing theory and algorithms to control the information flow of autonomous systems to succeed in information-scarce or adversarial environments.
\end{IEEEbiography}

\begin{IEEEbiography}[{\includegraphics[width=1in,height=1.25in,clip,keepaspectratio]{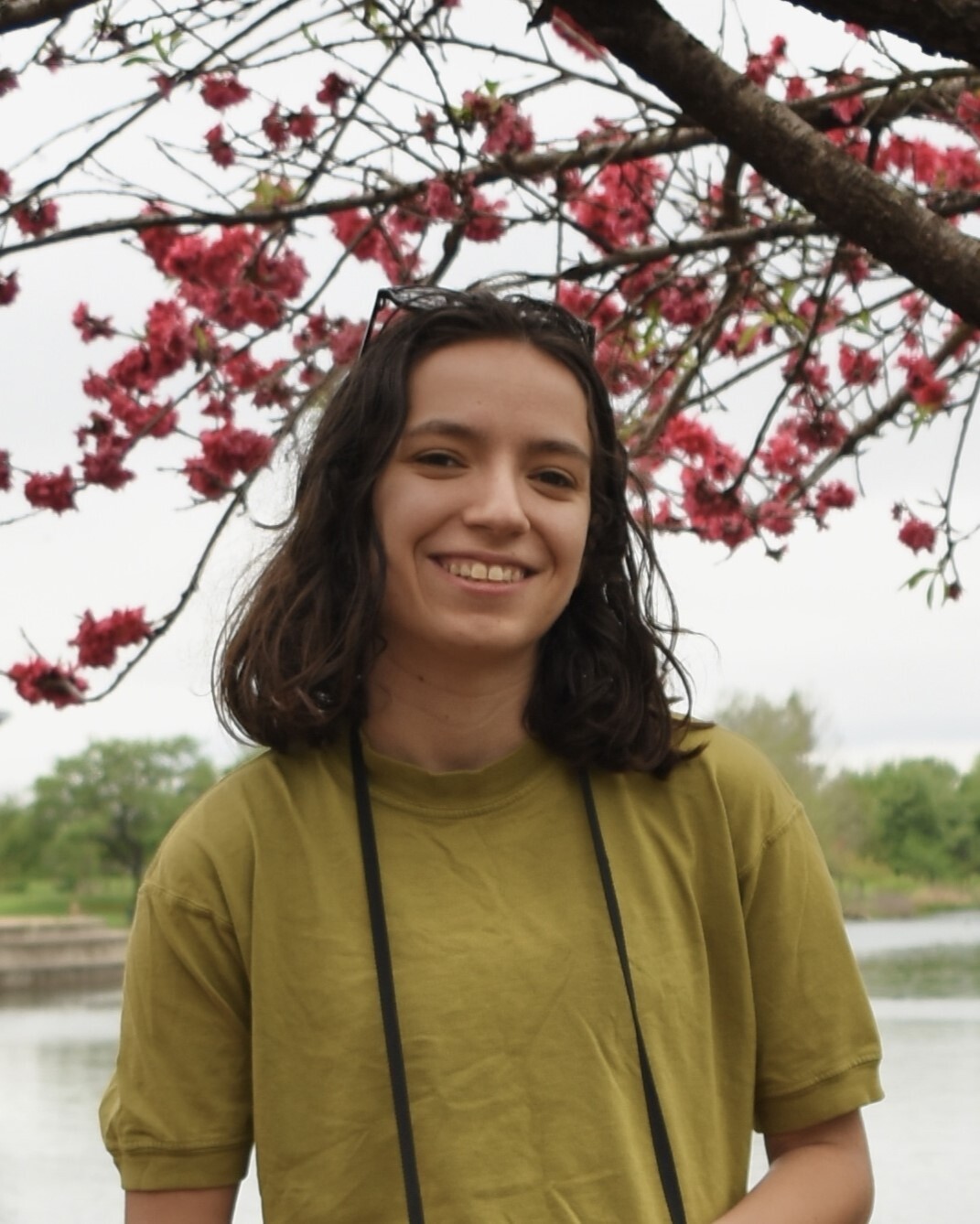}}]{Sophia Smith} received her B.S. in mathematics and her B.A. in Physics at the University of Chicago, Chicago IL, USA in 2021. She is currently working towards the Ph.D. degree in Computational Science Engineering in Math within the Oden Institute at the University of Texas at Austin, Austin, TX, USA. Her research interests include developing theory and algorithms for decomposing teams of autonomous agents.
\end{IEEEbiography}

\begin{IEEEbiography}[{\includegraphics[width=1in,height=1.25in,clip,keepaspectratio]{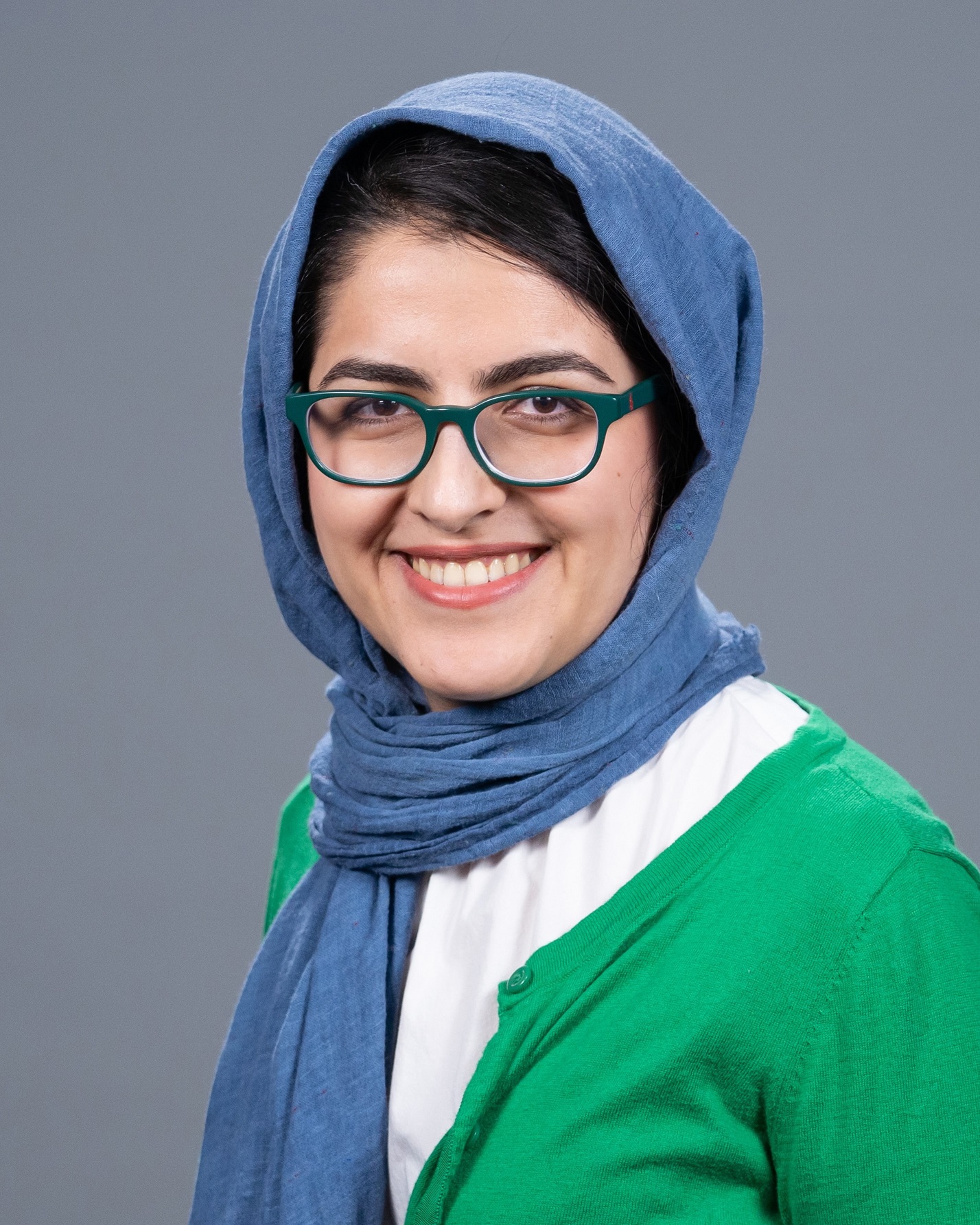}}]{Negar Mehr} is an assistant professor in the Department of Mechanical Engineering at the University of California, Berkeley. Before that, she was an assistant professor of Aerospace Engineering at the University of Illinois Urbana-Champaign. 
She was a postdoctoral scholar at Stanford Aeronautics and Astronautics department from 2019 to 2020. She received her Ph.D. in Mechanical Engineering from UC Berkeley in 2019 and her B.Sc. in Mechanical Engineering from Sharif University of Technology, Tehran, Iran, in 2013. The focus of her research is to develop control algorithms that allow autonomous systems to safely and intelligently interact with each other and with humans. She draws from the fields of control theory, robotics, game theory, and machine learning. She is a recipient of the NSF CAREER Award. She was awarded the IEEE Intelligent Transportation Systems best Ph.D. dissertation award in 2020.

\end{IEEEbiography}

\begin{IEEEbiography}[{\includegraphics[width=1in,height=1.25in,clip,keepaspectratio]{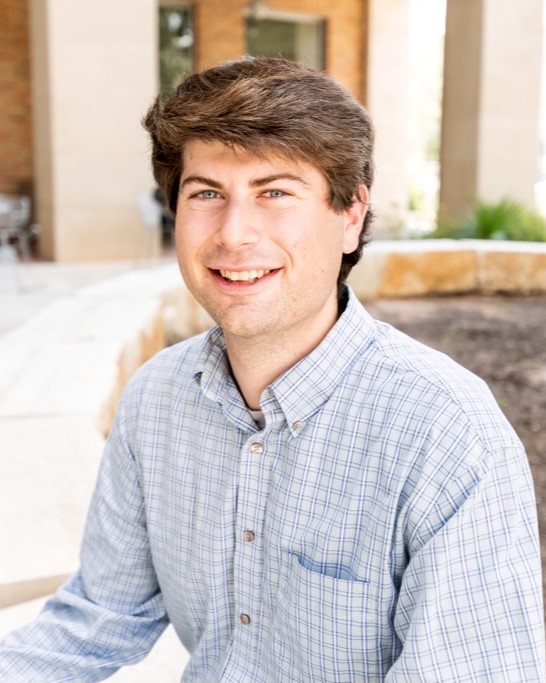}}]{David Fridovich-Keil} received the B.S.E. degree in electrical engineering from Princeton University, and the Ph.D. Degree from the University of California, Berkeley. He is an Assistant Professor in the Department of Aerospace Engineering and Engineering Mechanics at the University of Texas at Austin. Fridovich-Keil is the recipient of an NSF Graduate Research Fellowship and an NSF CAREER Award.
\end{IEEEbiography}

\begin{IEEEbiography}[{\includegraphics[width=1in,height=1.25in,clip,keepaspectratio]{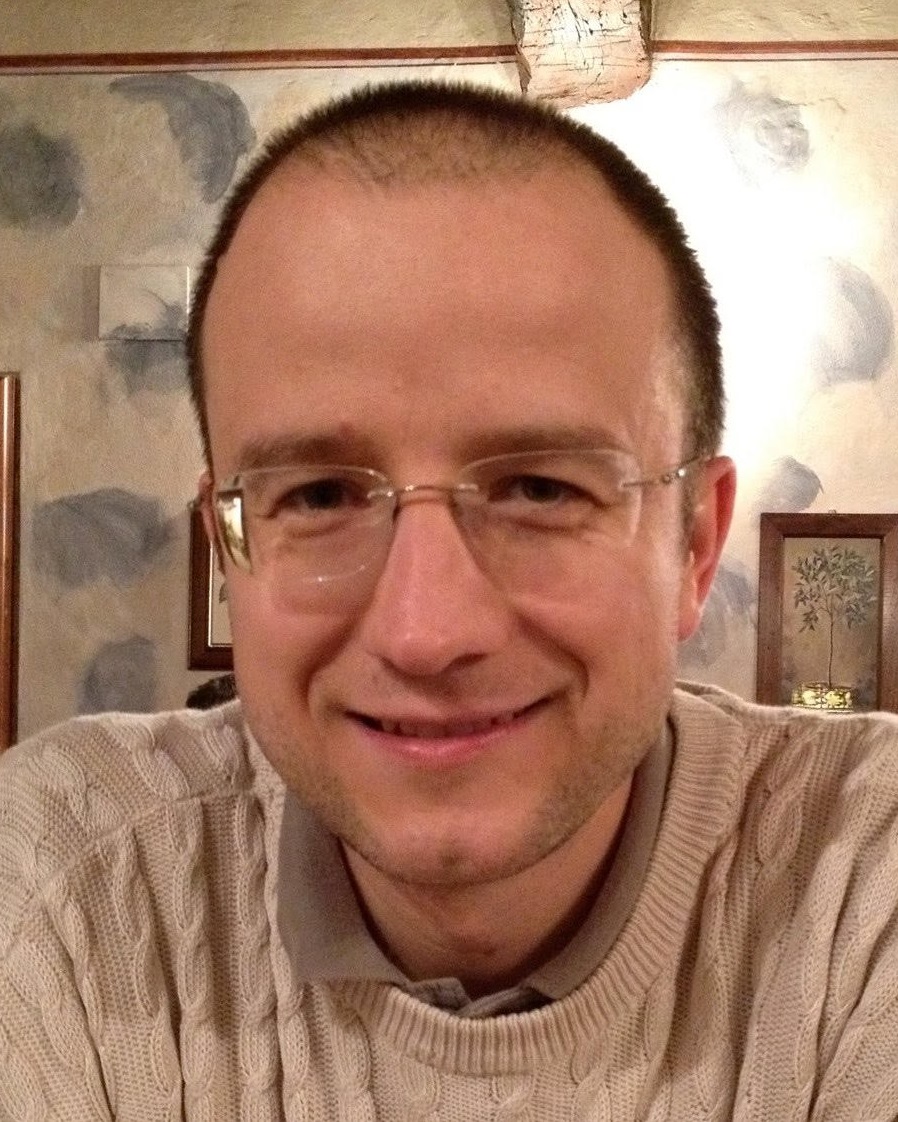}}]{Ufuk Topcu} received the Ph.D. degree from the University of California at Berkeley, Berkeley, CA, USA, in 2008. He joined the Department of Aerospace Engineering, University of Texas at Austin, Austin, TX, USA, in Fall 2015. He held research positions with the University of Pennsylvania, Philadelphia, PA, USA, and California Institute of Technology, Pasadena, CA, USA. His research focuses on the theoretical, algorithmic and computational aspects of design and verification of autonomous systems through novel connections between formal methods, learning theory, and controls.
\end{IEEEbiography}

\end{document}

%% file: commands.tex
\newcommand{\follower}{\textrm{f}}
\newcommand{\leader}{\textrm{l}}

\newcommand{\leadergenericaction}{i}
\newcommand{\followergenericaction}{j}

\newcommand{\leaderutilitymatrix}{A}
\newcommand{\followerutilitymatrix}{B}
\newcommand{\leadernumofactions}{m}
\newcommand{\followernumofactions}{n}
\newcommand{\expectation}[2]{\mathbb{E}_{#2}\left[#1\right]}

\newcommand{\numofinteractions}{K}
\newcommand{\genericinteraction}{k}

\newcommand{\leaderdistribution}{x}
\newcommand{\followerdistribution}{y}
\newcommand{\emprical}[1]{\hat{#1}}
\newcommand{\genericprobability}{p}
\newcommand{\simplex}{\Delta}

\newcommand{\rationalityconstant}{\lambda}
\newcommand{\entropy}{H}
\newcommand{\randomversion}[1]{\bm{#1}}
\newcommand{\empricalreturn}{IR}
\newcommand{\stackelbergreturn}{SR}
\newcommand{\onesmatrix}{J}

\newcommand{\softmax}{\sigma}
\newcommand{\function}{f}
\newcommand{\leaderfunction}{\function^{\leader}}
\newcommand{\followerfunction}{\function^{\follower}}
\newcommand{\leaderactionspace}{\mathcal{A}}
\newcommand{\followeractionspace}{\mathcal{B}}
\newcommand{\gamestatespace}{\mathcal{S}}
\newcommand{\transitionfunc}{\mathcal{P}}
\newcommand{\genericstate}{s}
\newcommand{\leadergenericcontaction}{a}
\newcommand{\followergenericcontaction}{b}
\newcommand{\leadercontdist}{\pi^{\leader}}
\newcommand{\followercontdist}{\pi^{\follower}}
\newcommand{\leaderparameter}{x}
\newcommand{\followerparameter}{y}
\newcommand{\leadercontdistspace}{X}
\newcommand{\followercontdistspace}{Y}
\newcommand{\categorizationfunction}{c}
\newcommand{\leadertypes}{\mathcal{T}}

%% file: intro.tex
\section{Introduction}
Autonomous agents repeatedly interact with other agents, e.g., humans and other autonomous systems, in their environments during their operations. Often, the intentions and strategies of an autonomous agent are not fully known by the other agents, and the other agents rely on statistical inference from past interactions when they react to the actions of the autonomous agent. For example, an autonomous car interacts with pedestrians who intend to cross the road, and pedestrians do not have full knowledge of the car's strategy. Consequently, acting in an inferable way is essential for autonomous agents.

We model the interaction between the autonomous agent and the other agent with a Stackelberg game. In this game, the autonomous agent is the \textit{leader} that commits to a strategy (e.g., a software), and the other agent is the \textit{follower} that reacts to the leader's strategy. The game is repeated between the agents. While the leader follows the same strategy at every interaction, the follower's strategy can change between interactions. For the autonomous car example, a version of the car's software defines a fixed strategy over actions, stopping and proceeding.

In a conventional Stackelberg game with mixed strategies, the follower does not know the leader's action but knows the leader's strategy from which the leader's actions are drawn~\cite{conitzer2016stackelberg}. We, on the other hand, consider that the follower does not have full information of the leader's strategy and relies on the observations from the previous interactions. More specifically, at every interaction, the follower reacts to the empirical action distribution from the previous interactions. For example, in the car-pedestrian scenario, the pedestrian would act based on the frequency that the car stopped in the previous interactions.

The leader's optimal strategy may be mixed for a conventional Stackelberg game~\cite{conitzer2016stackelberg}. 
However, in the inference setting that we consider, this strategy may not be optimal since the follower reacts to the empirically observed strategy of the leader, not the actual strategy. As a result, the leader might be better off using a less stochastic strategy since such strategies would be more inferable. 
As such, the leader's expected return in the inference setting might be lower than its expected return in the full information setting. We call the return gap between these settings the leader's \textit{inferability gap}. 

We show that for a repeated static Stackelberg game with parametric action spaces, the leader's cumulative inferability gap is bounded under some Lipschitz continuity assumptions. As a corollary of this result, we show that for a repeated static bimatrix Stackelberg game, if the follower has bounded rationality (modeled by the maximum entropy response), the leader's cumulative inferability gap is bounded. We also show that the inferability gap is bounded for repeated static bimatrix Stackelberg games with fully rational followers and dynamic bimatrix Stackelberg games with myopic, boundedly rational followers. These upper bounds are functions of both the stochasticity level of the leader's strategy (i.e., the total variance of the leader's strategy assuming that the follower's estimator is efficient) and the number of interactions. As the stochasticity level of the leader's strategy decreases, the inferability gap vanishes. In the extreme case where the leader's strategy is deterministic, the leader does not suffer from any inferability gap after a single interaction; the expected return in the inference setting is the same as the expected return in the full information setting. The inferability gap at interaction \(k\) is at most \(\mathcal{O}(\nicefrac{1}{\sqrt{k}})\) (ignoring the other terms), implying that \(\mathcal{O}(\nicefrac{1}{\epsilon^{2}})\) interactions are sufficient to achieve \(\epsilon\) inferability gap.

Motivated by this bound, we use the stochasticity level as a regularization term in the leader's objective function to find optimal strategies for the inference setting. Numerical experiments show that the leader indeed suffers from an inferability gap, and the strategies generated by the regularized objective function lead to improved transient returns compared to the strategies optimal for the full information setting. We also conduct a human subject study in a simulated driving environment where the participants interact with an autonomous car. Experiment results show that despite being suboptimal for the full information setting, strategies with lower stochasticity levels lead to improved return in interactions with humans.

Additionally, as a converse result, we provide an example bimatrix Stackelberg game where the inferability gap at interaction \(k\) is at least \(\epsilon\) if \(k\) is not at the order of \(\mathcal{O}(\nicefrac{1}{\epsilon^{2}})\) under the full rationality assumption for the follower. 

We also analyze a spectrum of bimatrix Stackelberg games and identify a set of games where the (near)-optimal strategy of the leader may suffer from a large inferability gap in the inference setting. More precisely, we show that there are strategies for the leader that limit the inferability gap if the game is almost cooperative or competitive, i.e., respectively, the difference or the sum of the leader's and follower's returns are bounded.

\textit{Related work:}
The closest work is the preliminary conference version~\cite{karabag2023encouraging} of this paper.
Building on the results of the conference version, we provide analyses of the inferability gap in repeated static Stackelberg games with parametric policies, repeated static bimatrix Stackelberg games with fully rational followers, repeated static bimatrix Stackelberg games with classifying followers, and repeated dynamic discrete Stackelberg games with myopic boundedly rational followers. We show that for almost cooperative or competitive bimatrix Stackelberg games, in the inference setting, there exist strategies that are nearly as good as the optimal strategies in the full information setting. We also provide an additional example to emphasize the importance of inferable strategies in a semi-cooperative setting modeled as a Stackelberg game with parametric strategies.

Bimatrix Stackelberg games with a commitment to mixed strategies have been extensively studied in the literature under the assumption that the follower has full knowledge of the leader's strategy~\cite{conitzer2016stackelberg,conitzer2006computing,von2010leadership}. For these games, an optimal strategy for the leader can be computed in polynomial time via linear programming (assuming that the follower breaks ties in favor of the leader)~\cite{conitzer2006computing}. The paper \cite{korzhyk2011solving} considers Stackelberg games with partial observability where the follower observes the leader's strategy with some probability and does not otherwise. We consider a different observability setting where the follower gets observations from the leader's strategy. Papers \cite{pita2010robust,karwowski2023sequential} also consider this observation setting. To account for the follower's partial information, \cite{pita2010robust,karwowski2023sequential,yin2011risk} consider a robust set that represents the possible realizations of the leader's strategy and maximizes the leader's worst-case return by solving a robust optimization problem. We follow a different approach and try to maximize the leader's expected return under observations by relating it to the return under the full information assumption.

We provide a lower bound on the leader's return that involves the stochasticity level (inferability) of the leader's strategy. To our knowledge, a bound in this spirit does not exist for Stackelberg games with observations. Works \cite{farokhi2017fisher,farokhi2019ensuring,karabag2019least} increase the stochasticity level of the control policy (the leader's strategy in our context) to improve the non-inferability in different contexts. We consider a stochasticity metric that coincides with the Fisher information metric considered in~\cite{farokhi2017fisher,farokhi2019ensuring,karabag2019least}. However, unlike these works, which focus on minimizing information and providing unachievability results, we provide an achievability result. 

Human-robot interactions are more efficient if the human knows the robot's intent. Conveying intent information via movement is explored to create legible behavior \cite{bodden2018flexible,dragan2013legibility}. These works are often concerned with creating a single trajectory that is distant from the trajectories under other intents. Different from the legible behavior literature that considers inferability during a single interaction, we consider statistical inferability over repeated interactions. 

The leader's problem in our setting is a bilevel optimization problem under data uncertainty~\cite{beck2023survey}. Works \cite{patriksson1997stochastic,patriksson1999stochastic,lin2010stochastic} consider stochastic bilevel optimization problems where first the leader commits to a strategy before the data uncertainty is resolved, then the data uncertainty is resolved, and finally, the follower makes a decision with known data. In our problem, the data distribuiton depends on the leader's decision\footnote{For single-level stochastic optimization problems, this setting is referred to as non-oblivious stochastic optimization~\cite{hassani2019stochastic}. }, whereas \cite{patriksson1997stochastic,patriksson1999stochastic,lin2010stochastic} consider a fixed distribution. 

We represent the boundedly rational follower using the maximum entropy model (also known as Boltzmann rationality model or quantal response)~\cite{mckelvey1995quantal,mehr2023maximum}. Alternatively, \cite{pita2010robust,zare2020bilevel} consider boundedly rational followers using the anchoring theory~\cite{fox2005subjective} or \(\epsilon\)-optimal follower models.

%% file: prelims.tex
\section{Notation and Preliminaries on Stackelberg Games}

We use upper-case letters for matrices and bold-face letters for random variables. \(\| \cdot \|\) denotes the L2 norm. \(\simplex^{N}\) denotes the \(N\)-dimensional probability simplex. \(\mathcal{N}(\mu, s^2)\) denotes the normal distribution with mean \(\mu\) and variance \(s^2\).
A function $\function : \mathbb{R}^{N} \to \mathbb{R}^{M}$ is \(L\)-Lipschitz continuous if it satisfies \(\|\function(z)-\function(q)\| \leq L \|z-q\|\) for any $z$ and $q$. The matrix of ones with appropriate dimensions is denoted by \(J\).

For \(z, q \in \simplex^{N}\), the entropy of \(z\), is \(H(z) = \sum_{i=1}^{N} z_{i} \log(1/z_{i}) \) where \(z_{i}\) is the \(i\)-th element of \(z\).
The softmax function \(\softmax_{\rationalityconstant}:\mathbb{R}^{N} \to \simplex^{N}\) is defined as \(\softmax_{\rationalityconstant}(z)_{i} := \frac{\exp(\rationalityconstant z_{i})}{\sum_{j=1}^{N} \exp(\rationalityconstant z_{j})}\) where \(\softmax_{\rationalityconstant}(z)_{i}\) is the \(i\)-th element of \(\softmax_{\rationalityconstant}(z)\).
The softmax function \(\softmax_{\rationalityconstant}\) is \(\rationalityconstant\)-Lipschitz continuous, i.e., it satisfies \(\|\softmax(z)-\softmax(q)\| \leq \rationalityconstant \|z-q\|\)for all \(z,q \in \mathbb{R}^{N}\)~\cite{softmax}.

We define the \textit{stochasticity level} of a discrete probability distribution \(z \in \simplex^{N}\) as \(\nu(z) := \sqrt{\sum_{i=1}^{N} z_{i}(1-z_{i})}\) that is the square root of the trace of the covariance matrix.

We also define a term \(\varphi(z)\) from \cite{weissman2003inequalities} to measure how concentrated \(z\) is. Let \(C \subseteq \lbrace 1, \ldots, N \rbrace\), and \(p({C})\) be the probability that a sampled element belongs to set \(C\) under distribution \(z\).  Define \(p_z = \max_{C \subseteq \lbrace1, \ldots, N \rbrace} \min  (z({C}), 1- z({C}))\). Note that the maximum value of \(p_z\) is less than \(\nicefrac{1}{2}\). Also, note that the maximum value of \(p_z\) is less than one minus the maximum element of \(p_z\). Define \(\varphi(p_z) = \frac{1}{1-2p_z} \log\left(\frac{1-p_z}{p_z}\right).\) Note that the minimum value \(\varphi(z)\) can take is \(2\), and \(\varphi(z)\) is an decreasing function of \(p_z\). As the strategy becomes more deterministic, \(\varphi(z)\) increases.

\begin{remark}
    In the following sections, we present various bilevel optimization problems. For notational and analytical simplicity, we assume that there exists a unique solution to the inner optimization problem, i.e., the follower's problem, and the leader knows this unique response given its strategy.
\end{remark}

We define a general dynamic Stackelberg game model and focus on special cases of this model in the later sections.

\subsection{Dynamic Stackelberg Games with Parametric Mixed Strategies}
A \textit{dynamic Stackelberg game} is a two-player game between a \textit{leader} and a \textit{follower}. The game has the state space \(\gamestatespace\), the leader has action space \(\leaderactionspace\), and the follower has action space \(\followeractionspace\).  \(\leaderfunction : \gamestatespace \times \leaderactionspace \times \followeractionspace \to \mathbb{R}\) is the leader's \textit{utility function} and \(\followerfunction : \gamestatespace \times \leaderactionspace \times \followeractionspace \to \mathbb{R}\) is the follower's utility function. 
At time \(t\), when the leader takes action \(\leadergenericcontaction_{t}\) and the follower takes action \(\followergenericcontaction_{t}\) at state \(\genericstate_{t}\), the leader and follower returns are  \(\leaderfunction(\genericstate_{t}, \leadergenericcontaction_{t}, \followergenericcontaction_{t})\) and \(\followerfunction(\genericstate_{t}, \leadergenericcontaction_{t}, \followergenericcontaction_{t})\) respectively. At time \(t+1\), the game's state follows distribution \(\transitionfunc(\genericstate_{t}, \leadergenericcontaction_{t}, \followergenericcontaction_{t})\).

In a dynamic Stackelberg game with \textit{parametric mixed strategies}, at state \(\genericstate\), the leader and follower have \textit{(stationary) mixed strategies} \(\leadercontdist(\genericstate; \leaderparameter)\) and  \(\followercontdist(\genericstate; \followerparameter)\) parametrized by \(\leaderparameter\in\leadercontdistspace\) and \(\followerparameter\in\followercontdistspace\), respectively. Here, \(\leadercontdist(\genericstate; \leaderparameter)\) and \(\followercontdist(\genericstate; \followerparameter)\) are probability distibutions over \(\leaderactionspace\) and \(\followeractionspace\), respectively. For example, in a parametric game with the same Gaussian action distribution for every state \(\genericstate\), one may have \(\leadercontdist(\genericstate;\leaderparameter) = \mathcal{N}(\theta, 1)\) where \(\leadercontdistspace = \mathbb{R}\) and \(\followercontdist(\genericstate; \followerparameter) = \mathcal{N}(0, \phi)\) where \(\followercontdistspace= \mathbb{R}_{\geq 0}\). When the follower's strategy is fixed, the game is a Markov decision process (MDP) for the leader, and vice versa.

Let \(\randomversion{\genericstate}_{t}\), \(\randomversion{\leadergenericcontaction}_{t}\), and \(\randomversion{\followergenericcontaction}_{t}\) denote the random versions of \(\genericstate_{t}\),  \(\leadergenericcontaction_{t}\), and \(\followergenericcontaction_{t}\), respectively. At time \({t}\), the leader's \textit{expected utility} is \[ \expectation{\leaderfunction(\genericstate_{t},\randomversion{\leadergenericcontaction}_{t},\randomversion{\followergenericcontaction}_{t}) }{\randomversion{\leadergenericcontaction}_{t} \sim \leadercontdist(\genericstate_{t};\leaderparameter), \randomversion{\followergenericcontaction}_{t} \sim \followercontdist(\genericstate_{t}; \followerparameter)},\] and the follower's expected utility is \[ \expectation{\followerfunction(\genericstate_{t},\randomversion{\leadergenericcontaction}_{t},\randomversion{\followergenericcontaction_{t}}) }{\randomversion{\leadergenericcontaction}_{t} \sim \leadercontdist(\genericstate_{t};\leaderparameter), \randomversion{\followergenericcontaction}_{t} \sim \followercontdist(\genericstate_{t}; \followerparameter)}.\]  \textcolor{black}{Let \(\tau\) be the random, first hitting time to a return-free, invariant set \(\gamestatespace_{end}\) of states. The players reach \(\gamestatespace_{end}\) after \(\tau\) steps and cannot leave this set under any sequence of actions. We consider that the interaction between the players ends upon reaching \(\gamestatespace_{end}\). For example, in a navigation scenario \(\gamestatespace_{end}\) could be the target locations. In a cyber system, \(\gamestatespace_{end}\) could represent the state where users are disconnected. The time \(\tau\) (potentially) depends on the players' strategies. For example, in the navigation scenario, the time that it takes for an agent to reach its target may depend on its strategy and the other player's strategy.}
The leader's \textit{expected cumulative utility} is \[ \mathbb{E}\left[\sum_{t=0}^{\tau}\leaderfunction(\randomversion{\genericstate}_{t},\randomversion{\leadergenericcontaction}_{t},\randomversion{\followergenericcontaction}_{t})\right],\] and the follower's expected cumulative utility is \[ \mathbb{E}\left[\sum_{t=0}^{\tau}\followerfunction(\randomversion{\genericstate}_{t},\randomversion{\leadergenericcontaction}_{t},\randomversion{\followergenericcontaction_{t}})\right],\] where the expectations are over the randomness of states \(\randomversion{\genericstate}_{0}  \ldots\randomversion{\genericstate}_{\tau}\) and actions \(\randomversion{\leadergenericcontaction}_{0}\ldots\randomversion{\leadergenericcontaction}_{\tau}\), \(\randomversion{\followergenericcontaction}_{0}\ldots\randomversion{\followergenericcontaction}_{\tau}\).

In the full information setting, when deciding on its strategy parameter \(\followerparameter\), the follower knows the leader's parameter \(\leaderparameter\). This means the follower knows the probability distribution of the leader's action for every state but does not know the leader's realized actions. 
The follower's goal is to maximize its expected return given the leader's strategy \(\leaderparameter\), by solving: \[\max_{\followerparameter \in \followercontdistspace}  \ \mathbb{E}\left[\sum_{t=0}^{\tau}\followerfunction(\randomversion{\genericstate}_{t},\randomversion{\leadergenericcontaction}_{t},\randomversion{\followergenericcontaction}_{t}) \right]\] where the expectation is over the randomness of states \(\randomversion{\genericstate}_{0}  \ldots\randomversion{\genericstate}_{\tau}\) and actions \(\randomversion{\leadergenericcontaction}_{0}\ldots\randomversion{\leadergenericcontaction}_{\tau}\), \(\randomversion{\followergenericcontaction}_{0}\ldots\randomversion{\followergenericcontaction}_{\tau}\)

The leader's goal is to maximize its 
expected return, i.e., solve the bilevel optimization problem:
\begin{align*}
      \sup_{\leaderparameter \in \leadercontdistspace} \ & \ \mathbb{E}\left[\sum_{t=0}^{\tau}\leaderfunction(\randomversion{\genericstate}_{t},\randomversion{\leadergenericcontaction}_{t},\randomversion{\followergenericcontaction}_{t})\right]  \\  \text{  s.t. }&
     \quad \followerparameter^{*}  = \arg\max_{\followerparameter \in \followercontdistspace}  \mathbb{E}\left[\sum_{t=0}^{\tau}\expectation{\followerfunction(\randomversion{\genericstate}_{t},\randomversion{\leadergenericcontaction}_{t},\randomversion{\followergenericcontaction_{t}}) }{}\right]. 
\end{align*}

We define \begin{align*}
\stackelbergreturn(\leaderparameter):=  \ &\mathbb{E}\left[\sum_{t=0}^{\tau}\leaderfunction(\randomversion{\genericstate}_{t},\randomversion{\leadergenericcontaction}_{t},\randomversion{\followergenericcontaction}_{t})\right]  \\ &\text{  s.t. }
     \quad \followerparameter^{*} = \arg\max_{\followerparameter \in \followercontdistspace}  \mathbb{E}\left[\sum_{t=0}^{\tau}\expectation{\followerfunction(\randomversion{\genericstate}_{t},\randomversion{\leadergenericcontaction}_{t},\randomversion{\followergenericcontaction_{t}}) }{}\right].
 \end{align*}
 We refer to \(\stackelbergreturn(\leaderdistribution)\) as the Stackelberg return of strategy (parameter) \(\leaderdistribution\) under full information and full follower rationality.

If random hitting time \(\tau\) is deterministically \(0\), i.e., the game has a single state, the game is \textit{static} with parametric mixed strategies. In this case, we drop \(\genericstate\) from our notation and use notation \(\leadercontdist(\leaderparameter)\), \(\followercontdist(\followerparameter)\), \(\leaderfunction(\leadergenericcontaction, \followergenericcontaction)\), and \(\followerfunction(\leadergenericcontaction, \followergenericcontaction)\).

If the game has finite state and action spaces, the strategies are distributions over actions. In this case, we drop \(\pi\) and use notation \(\leaderparameter(\genericstate)\) and \(\followerdistribution(\genericstate)\) to denote the action distributions of the leader and the follower at state \(\genericstate\), respectively.

\color{black}

\subsection{Special Case: Static Bimatrix Stackelberg Games}
In a \textit{static bimatrix Stackelberg game}, the leader has \(\leadernumofactions\) (enumerated) actions, and the follower has \(\followernumofactions\) (enumerated) actions. 
We call matrix \(\leaderutilitymatrix \in \mathbb{R}^{\leadernumofactions \times \followernumofactions}\) the leader's \textit{utility matrix} and \(\followerutilitymatrix \in \mathbb{R}^{\leadernumofactions \times \followernumofactions}\) the follower's utility matrix. 
When the leader and the follower take actions \(\leadergenericaction\) and \(\followergenericaction\), the leader and follower returns are  \(\leaderfunction(\leadergenericaction,\followergenericaction)= \leaderutilitymatrix_{\leadergenericaction \followergenericaction}\) and \(\followerfunction(\leadergenericaction,\followergenericaction) = \followerutilitymatrix_{\leadergenericaction \followergenericaction}\) respectively. 

In bimatrix Stackelberg games with \textit{mixed strategies}, the leader has a mixed strategy \(\leaderdistribution \in \simplex^{\leadernumofactions}\), and the follower has mixed strategy \(\followerdistribution \in \simplex^{\followernumofactions}\). We drop the notation \(\leadercontdist\) and \(\followercontdist\) when considering bimatrix Stackelberg games as the parameters directly correspond to the action probabilities.
Let \(\randomversion{\leadergenericaction}\) and \(\randomversion{\followergenericaction}\) denote the random versions of \(\leadergenericaction\) and \(\followergenericaction\), respectively. The leader's \textit{expected utility} is \(\leaderdistribution^{\top} \leaderutilitymatrix \followerdistribution = \expectation{\leaderutilitymatrix_{\randomversion{\leadergenericaction},\randomversion{\followergenericaction}}}{\randomversion{\leadergenericaction} \sim \leaderdistribution, \randomversion{\followergenericaction} \sim \followerdistribution},\) and the follower's expected utility is \(\leaderdistribution^{\top} \followerutilitymatrix \followerdistribution = \expectation{\followerutilitymatrix_{\randomversion{\leadergenericaction},\randomversion{\followergenericaction}}}{\randomversion{\leadergenericaction} \sim \leaderdistribution, \randomversion{\followergenericaction} \sim \followerdistribution}.\)
When deciding on strategy \(\followerdistribution\), the follower knows the leader's strategy \(\leaderdistribution\). The follower knows the leader's action probabilities but does not know the leader's realized action. 

\subsubsection{Fully Rational Followers} \label{sec:staticdiscretefullysetting}
The follower's goal is to maximize its expected return given \(\leaderdistribution\) by solving: \[\max_{\followerdistribution \in \simplex^{\followernumofactions}} \leaderdistribution^{\top} \followerutilitymatrix \followerdistribution.\] 

The leader's goal is to maximize its expected return, i.e., solve the bilevel optimization problem:
\begin{align*}
      \sup_{\leaderdistribution \in \simplex^{\leadernumofactions}} \ & \ \leaderdistribution^{\top} \leaderutilitymatrix \followerdistribution^{*}  \\ \text{  s.t. }
     &\quad \followerdistribution^{*}  = \arg\max_{\followerdistribution \in \simplex^{\followernumofactions}} \leaderdistribution^{\top} \followerutilitymatrix \followerdistribution .
\end{align*}
Note that an optimal solution may not exist for this problem.

We define \begin{align*}
\stackelbergreturn(x):=  &  \ \leaderdistribution^{\top} \leaderutilitymatrix \followerdistribution^{*}  \\ & \text{  s.t. }
 \quad \followerdistribution^{*}  = \arg\max_{\followerdistribution \in \simplex^{\followernumofactions}} \leaderdistribution^{\top} \followerutilitymatrix \followerdistribution .
 \end{align*}
 We refer to \(\stackelbergreturn\) as the Stackelberg return under full information and full follower rationality. Note that \(\stackelbergreturn(x)\) is a piecewise linear function of \(x\).

\subsubsection{Boundedly Rational Followers with Maximum Entropy Response}
Bounded rationality models represent the decision-making process of an agent with limited information or information processing capabilities and are often used to model the decision-making process of humans~\cite{rubinstein1998modeling}. We consider the maximum entropy model (Boltzmann rationality) to represent boundedly rational followers~\cite{braun2014information}.

Given the leader's strategy \(\leaderdistribution\), the boundedly rational follower solves the following optimization problem \[\max_{\followerdistribution \in \simplex^{\followernumofactions}} \leaderdistribution^{\top} \followerutilitymatrix \followerdistribution + \frac{1}{\rationalityconstant} \entropy(\followerdistribution)\] where \(\rationalityconstant\) denotes the follower's rationality level. Note that for \(\rationalityconstant \in (0,\infty)\), the optimal solution for the above problem is unique since the objective function is strictly concave and is given by \(\softmax_{\rationalityconstant}(\followerutilitymatrix^{\top} \leaderdistribution)\)~\cite{softmax}. 
As \(\rationalityconstant \to 0\), the follower does not take its expected utility \(\leaderdistribution^{\top} \followerutilitymatrix \followerdistribution\) into account and takes all available actions uniformly randomly. As \(\rationalityconstant \to \infty\), the follower becomes fully rational.  Given that the follower is boundedly rational with level \(\rationalityconstant  \in (0,\infty)\), the leader's goal is to maximize its expected utility, i.e., solve \[\max_{\leaderdistribution \in \simplex^{\leadernumofactions}}  \leaderdistribution^{\top} \leaderutilitymatrix \followerdistribution^{*}\] such that \(\followerdistribution^{*} = \softmax_{\rationalityconstant}(\followerutilitymatrix^{\top} \leaderdistribution). \) 
In this setting, we define 
\[\stackelbergreturn(x):= \leaderdistribution^{\top} \leaderutilitymatrix \followerdistribution^{*} \text{ where } \followerdistribution^{*} = \softmax_{\rationalityconstant}(\followerutilitymatrix^{\top} \leaderdistribution).\] as the Stackelberg return under full information and bounded follower rationality with maximum entropy response.

\subsubsection{Boundedly Rational Followers with Maximum Likelihood Classificatrion}
A classification function \(\categorizationfunction: \simplex^{\leadernumofactions} \to \leadertypes\) maps the leader's strategy to discrete set \(\leadertypes = \lbrace \leaderdistribution^{1}, \ldots, \leaderdistribution^{|\leadertypes|} \rbrace   \subset \simplex^{\leadernumofactions}\) of predefined strategies. For example, a pedestrian considers drivers to be either aggressive or defensive. In detail, \(\categorizationfunction(\leaderdistribution)\) is the maximum likelihood estimate given data \(\leaderdistribution\) and the parameter space \(\mathcal{T}\), i.e., \(\categorizationfunction(\leaderdistribution) = \arg \max_{\leaderdistribution^{i} \in \mathcal{T} } \Pr(\leaderdistribution | \leaderdistribution^{i})\).
A boundedly rational follower that classifies solves the following optimization problem \[\max_{\followerdistribution \in \simplex^{\followernumofactions}} \categorizationfunction(\leaderdistribution)^{\top} \followerutilitymatrix \followerdistribution\] where \(\categorizationfunction(\leaderdistribution)\) denotes the follower's classification of the leader's type.  Given the follower's classification function \(\categorizationfunction\), the leader's goal is to maximize its expected utility, i.e., solve \begin{align*}
      \sup_{\leaderdistribution \in \simplex^{\leadernumofactions}} \ &\leaderdistribution^{\top} \leaderutilitymatrix \followerdistribution^{*}  \\ \text{  s.t. }
     &\quad \followerdistribution^{*}  =\arg\max_{\followerdistribution \in \simplex^{\followernumofactions}} \categorizationfunction(\leaderdistribution)^{\top} \followerutilitymatrix \followerdistribution .
\end{align*}
In this setting, we define \begin{align*}
\stackelbergreturn(x):=  & \ \leaderdistribution^{\top} \leaderutilitymatrix \followerdistribution^{*}  \\ 
 & \text{  s.t. }
 \quad \followerdistribution^{*} =\arg\max_{\followerdistribution \in \simplex^{\followernumofactions}} \categorizationfunction(\leaderdistribution)^{\top} \followerutilitymatrix \followerdistribution,
 \end{align*}
 as the Stackelberg return under full information and bounded follower rationality with maximum likelihood classification.

\color{black}

%% file: problem_formulation.tex
\section{Problem Formulation: Repeated Stackelberg Games With Inference}

\begin{table}[]
    \caption{Overview of the problem settings and results}
    \label{tab:achieveabilityresults}
\begin{center}
\begin{tabular}{||m{1.1cm} | m{1cm} | m{1.2cm} | m{1.3cm} | m{1.7cm}||} 
\hline
\multicolumn{5}{||c||}{ Achievability Results (Upper bounds on $\stackelbergreturn(\leaderdistribution) - \empricalreturn_{\genericinteraction}(\leaderdistribution)$) } 
\\
 \hline
 Problem section & Result section & Horizon & Strategy class & Follower type \\ [0.5ex] 
 \hline
 \ref{sec:probstaticparametric} & \ref{sec:resstaticparametric} & Static & Parametric & Rational \\ 
  \hline
 \ref{sec:probdiscretestaticfully} & \ref{sec:resdiscretestaticfully} & Static & Discrete & Rational \\
 \hline
   \ref{sec:probdiscretestaticentropy} & \ref{sec:resdiscretestaticentropy} & Static & Discrete & Max. Ent. \\
 \hline
 \ref{sec:probdiscretestaticcategorizing} & \ref{sec:resdiscretestaticcategorizing} & Static& Discrete & Classifying  \\
 \hline
\ref{sec:probdynamicdiscreteentropy} & \ref{sec:resdynamicdiscreteentropy} & Dynamic & Discrete & Myopic Max. Ent.  \\ [1ex] 
 \hline \hline
 \multicolumn{5}{||c||}{ Converse Result (Lower bound on $\stackelbergreturn(\leaderdistribution) - \empricalreturn_{\genericinteraction}(\leaderdistribution)$) } 
\\
 \hline
 Problem section & Result section & Horizon & Strategy class & Follower type \\ [0.5ex] 
 \hline
\ref{sec:probdiscretestaticfully} & \ref{sec:staticfullyconverse} & Static & Discrete & Rational \\ 
  \hline
  \hline
 \multicolumn{5}{||c||}{ Achievability Result (Upper bound on $\max_{\leaderdistribution} \stackelbergreturn(\leaderdistribution) - \max_{\leaderdistribution}\empricalreturn_{\genericinteraction}(\leaderdistribution)$) } 
\\
 \hline
 Setting section & Result section & Horizon & Strategy class & Follower type \\ [0.5ex] 
 \hline
\ref{sec:staticdiscretefullysetting} & \ref{sec:importanceofinf} & Static & Discrete & Rational \\ 
  \hline
\end{tabular}
\end{center}
\end{table}

In this section, we formulate decision-making problems for the leader where the follower statistically infers the strategy of the leader through repeated interactions. Let \(\empricalreturn_{\genericinteraction}(\leaderdistribution)\) be the expected return of the leader at interaction \(\genericinteraction\) in the inference setting. \(\stackelbergreturn(\leaderdistribution)\) is the expected return of the leader in the full information setting. Since \(\empricalreturn_{\genericinteraction}(\leaderdistribution)\) depends on the statistical inference of \(\leaderdistribution\) by the follower, it is not necessarily the same with \(\stackelbergreturn(\leaderdistribution)\). We refer the difference \(\stackelbergreturn(\leaderdistribution) - \empricalreturn_{\genericinteraction}(\leaderdistribution)\) as the \textit{inferability gap} of the leader at interaction \(\genericinteraction\). We are interested in analyzing the inferability gap to find strategies that remain performant in the inference setting. 

In detail, the leader's strategy affects the follower's optimal strategy in two ways in Stackelberg games with inference.  
First, as in the conventional Stackelberg game formulation, the leader's strategy determines the expected return for different follower actions, i.e., \(\leaderdistribution^{\top} \followerutilitymatrix\) in the bimatrix Stackelberg games.
This affects the optimal strategies for the follower. 
Second, unlike in the full information Stackelberg game, the leader's strategy \(\leaderdistribution\) modifies the distribution of its empirical action distribution \(\emprical{\randomversion{\leaderdistribution}}_{\genericinteraction}\) and, consequently, the follower's strategy \(\randomversion{\followerdistribution}_{\genericinteraction}\).

A strategy with a high Stackelberg return under full information may be highly suboptimal in a Stackelberg game with inference. Different realizations of \(\emprical{\randomversion{\leaderdistribution}}_{\genericinteraction}\) lead to different solutions for \(\randomversion{\followerdistribution}_{\genericinteraction}\). If \(\leaderdistribution\) is poorly inferred by the follower, the follower's strategy \(\randomversion{\followerdistribution}_{\genericinteraction}\) may yield poor returns when simultaneously played with \(\leaderdistribution\).  In the inference setting, an optimal strategy \(\leaderdistribution^{*}\) will strike a balance between having a high Stackelberg return under full information and efficiently conveying information about itself to the follower.

We analyze inferability gap for different problem settings. We first consider a static game with parametric action spaces. Then, we consider a static bimatrix game over discrete action spaces with fully rational and boundedly rational followers. Finally, we consider a dynamic game over discrete state and action spaces. Table \ref{tab:achieveabilityresults} summarizes of the settings and results.

\begin{remark}
    In some settings, the leader's strategy could be deterministic in a higher-dimensional state space. However, the follower may perceive the strategy as mixed in a lower-dimensional state space. For example, the car's decision may be a deterministic function of the car's exact distance to the pedestrian. However, a pedestrian with imperfect information about the car's state, i.e., the exact distance, may think the car has a mixed strategy. While we do not explicitly consider games with imperfect information, the problems and results can easily be extended to these settings.
\end{remark}

\textcolor{black}{
As mentioned in Remark 1, for notational and analytical convenience, in the remaining sections, we assume that the follower's optimal response exists and is unique. We note that if the follower's response is not unique, one can formulate the leader's problem by considering the worst-case outcome as given in Definitions 3.26 and 3.27 of \cite{bacsar1998dynamic}, i.e., among the strategies that maximize its objective function, the follower chooses a strategy that minimizes the leader's return (under full information or inference). The problem formulations and the results derived in the remaining sections easily extend to the worst-case formulation. 
We also note that for the followers with strongly convex utility functions, the optimal response is unique even without the assumption. For example, the follower's response is unique for the boundedly rational followers with maximum entropy response.}

\subsection{Repeated Static Parametric Stackelberg Games with Inference Against Fully Rational Followers} \label{sec:probstaticparametric}
Consider a Stackelberg game with mixed strategies that is repeated \( K \) times.   
The follower knows the leader's parametric strategy space but does not know the leader's fixed mixed strategy (and its associated parameter) \(\leadercontdist(\leaderparameter)\). Instead, the follower infers the leader's parameter \(\leaderparameter\) from observations of the previous interactions.
At interaction \(\genericinteraction \),   let  \(\emprical{\leaderparameter}_{\genericinteraction} \) be the estimator of the leader's parameter by the follower based on the leader's previous actions \(\leadergenericcontaction_{1},\ldots, \leadergenericcontaction_{k-1}\).

Under these assumptions, the follower's strategy $\followercontdist(\followerparameter_\genericinteraction)$ at interaction $\genericinteraction$ depends on the leader's actions in the previous $\genericinteraction-1$ interactions. For this reason, $\followercontdist(\followerparameter_\genericinteraction)$ changes with \(\genericinteraction\).
For a fully rational follower based on the estimate \(\emprical{\leaderparameter}_{\genericinteraction}\), we donote the follower's optimal parameter with \(\followerparameter_\genericinteraction^{*}\) such that
\( \followercontdist(\followerparameter_\genericinteraction^{*}) = \arg\max_{\followerparameter \in \followercontdistspace} \followerfunction(\leadercontdist(\emprical{\leaderparameter}_{k}), \followercontdist(\followerparameter)).\) 

The leader's decision-making problem is to a priori select a strategy $\leadercontdist(\leaderparameter^{*})$ that maximizes its expected cumulative return under inference, i.e., assuming that the follower rationally responds to the estimator $\leadercontdist(\emprical{\leaderparameter}_{\genericinteraction})$ of $\leadercontdist(\leaderparameter^{*})$ at each interaction \(\genericinteraction\). Let \(\randomversion{\leadergenericcontaction}_{\genericinteraction}\), \(\emprical{\randomversion{\leaderparameter}}_{\genericinteraction}\), and \(\randomversion{\followerdistribution}_{\genericinteraction}\) be random variables denoting the unrealized versions of \(\leadergenericcontaction_{\genericinteraction}\), \(\emprical{\leaderdistribution}_{\genericinteraction}\) and \(\followerdistribution_{\genericinteraction}\), respectively. If exists,
the leader's optimal strategy $x \in \leadercontdistspace$  maximizes  
\begin{align*}
& 
\expectation{\sum_{\genericinteraction=1}^{\numofinteractions} \leaderfunction(\randomversion{\leadergenericcontaction}_{k},\randomversion{\followergenericcontaction}_{k})}{} \\& 
\quad \text{s.t. } \randomversion{\followerdistribution}_{\genericinteraction}^{*} =  \arg\max_{\followercontdist \in \followercontdistspace} \expectation{\followerfunction(\randomversion{\leadergenericcontaction},\randomversion{\followergenericcontaction})}{\randomversion{\leadergenericcontaction} \sim \leadercontdist(\emprical{\randomversion{\leaderdistribution}}_{\genericinteraction}), \randomversion{\followergenericcontaction} \sim \followercontdist(\followerparameter)}.
\end{align*}
Here, $\randomversion{\leadergenericcontaction}_{k} \sim \leadercontdist(\leaderparameter), \randomversion{\followergenericcontaction}_{k} \sim \leadercontdist(\randomversion{\followerparameter}^{*}_{\genericinteraction})$, and the randomness of \(\randomversion{\followerdistribution}_{\genericinteraction}^{*}\) is over the leader's actions \(\randomversion{\leadergenericcontaction}_{1}, \ldots, \randomversion{\leadergenericcontaction}_{\genericinteraction-1}\). 
The leader solves this decision problem prior to taking any action, meaning their future actions \(\randomversion{\leadergenericcontaction}_{1}, \ldots, \randomversion{\leadergenericcontaction}_{\numofinteractions}\) are random variables that are all independently sampled from \(\leaderdistribution^{*}\). Since the follower's estimation  \(\emprical{\randomversion{\leaderparameter}}_{\genericinteraction}\) is a function of these future actions and the follower's strategy \(\randomversion{\followerdistribution}_{\genericinteraction}^{*}\) is a function of \(\emprical{\randomversion{\leaderdistribution}}_{\genericinteraction}\), they are both random variables as well and therefore bolded.

\subsection{Repeated Static Bimatrix Stackelberg Games with Inference Against Fully or Boundedly Rational Followers}

Consider a bimatrix Stackelberg game with mixed strategies that is repeated \( K \) times. The follower infers the leader's strategy from observations of the previous interactions. For example, at interaction \(\genericinteraction \), let \(\emprical{\leaderdistribution}_{\genericinteraction} \) be the \textit{plug-in sample mean estimation} of the leader's strategy. 
Specifically, if the leader takes actions $\leadergenericaction_1, \ldots, \leadergenericaction_{k-1}$ at the previous $k-1$ interactions,
\[
\left(\emprical{\leaderdistribution}_{\genericinteraction} \right)_{l} = \frac{\#_{t=1}^{k-1} \text{ }  (i_t = l) }{k-1},
\]
where \( \left(\emprical{\leaderdistribution}_{\genericinteraction} \right)_{l}\) is the $l^\text{th}$ element of vector \(\emprical{\leaderdistribution}_{\genericinteraction} \) and \( \#(\cdot) \) counts the number of times the input is true. 

The follower's optimal strategy $y_k^{*}$ at interaction $k$ depends on the leader's actions in the previous $k-1$ interactions. For this reason, $y_k^{*}$ changes over interactions.
For example, for a fully rational follower,
\( \followerdistribution_\genericinteraction^{*} =\arg\max_{\followerdistribution \in \simplex^{\followernumofactions}} \emprical{\leaderdistribution}^{\top}_{k} \followerutilitymatrix \followerdistribution.\)

Next, we consider the following formulations of the leader's problem for different levels of follower rationality.
\subsubsection{Fully rational follower} \label{sec:probdiscretestaticfully}
The leader's decision-making problem is to a priori select a strategy $\leaderdistribution^{*}$ that maximizes its expected cumulative return under inference, i.e., assuming that the follower rationally responds to the plug-in of $\leaderdistribution^{*}$ at each interaction \(\genericinteraction\). Let \(\randomversion{\leadergenericaction}_{\genericinteraction}\), \(\emprical{\randomversion{\leaderdistribution}}_{\genericinteraction}\), and \(\randomversion{\followerdistribution}_{\genericinteraction}^{*}\) be random variables denoting the unrealized versions of \(\leadergenericaction_{\genericinteraction}\), \(\emprical{\leaderdistribution}_{\genericinteraction}\) and \(\followerdistribution_{\genericinteraction}^{*}\), respectively. 
The leader's decision problem is to maximize
\begin{align*}
&\expectation{\sum_{\genericinteraction=1}^{\numofinteractions}  \leaderdistribution^{\top} \leaderutilitymatrix \randomversion{\followerdistribution}_{\genericinteraction}^{*}}{} & \text{ s.t. } \randomversion{\followerdistribution}_{\genericinteraction}^{*} =\arg\max_{\followerdistribution \in \simplex^{\followernumofactions}} \emprical{\randomversion{\leaderdistribution}}^{\top}_{\genericinteraction}\followerutilitymatrix \followerdistribution.
\end{align*}

The expectation is over the randomness in the leader's (random) actions \(\randomversion{\leadergenericaction}_{1}, \ldots, \randomversion{\leadergenericaction}_{\numofinteractions}\).

\subsubsection{Boundedly rational follower with maximum entropy response} \label{sec:probdiscretestaticentropy}
The leader's decision problem is to maximize
\begin{align*}
&\expectation{\sum_{\genericinteraction=1}^{\numofinteractions}  \leaderdistribution^{\top} \leaderutilitymatrix \randomversion{\followerdistribution}_{\genericinteraction}^{*}}{} &\text{ s.t. } \randomversion{\followerdistribution}_{\genericinteraction}^{*}  = \softmax_{\rationalityconstant}(\followerutilitymatrix^{\top} \emprical{\randomversion{\leaderdistribution}}_{k}).
\end{align*}

\subsubsection{Boundedly rational follower with maximum likelihood classification} \label{sec:probdiscretestaticcategorizing}
The leader's decision problem is to maximize
\begin{align*}
&\expectation{\sum_{\genericinteraction=1}^{\numofinteractions}  \leaderdistribution^{\top} \leaderutilitymatrix \randomversion{\followerdistribution}_{\genericinteraction}^{*}}{} &\text{s.t. } \randomversion{\followerdistribution}_{\genericinteraction}^{*} =\arg\max_{\followerdistribution \in \simplex^{\followernumofactions}} \categorizationfunction(\emprical{\randomversion{\leaderdistribution}}_{\genericinteraction})^{\top}\followerutilitymatrix \followerdistribution.
\end{align*}

\subsection{Repeated Discrete Dynamic Parametric Stackelberg Games with Inference Against Boundedly Rational, Myopic Followers} \label{sec:probdynamicdiscreteentropy}Assume that the state space \(\gamestatespace\), the leader's action space \(\leaderactionspace\), and the follower's action space \(\followeractionspace\) are finite. After every interaction \(\genericstate_{0} \leadergenericcontaction_{0} \followergenericcontaction_{0} \ldots\genericstate_{\tau} \leadergenericcontaction_{\tau} \followergenericcontaction_{\tau}\), the follower updates its estimation of the leader's strategy using the plug-in sample mean estimator for every state using all previous sample actions. At interaction \(k\), we denote the follower's estimation of the leader's strategy for state \(\genericstate\) with \(\emprical{\leaderdistribution}_{\genericinteraction}(\genericstate)\). At interaction \(k\), a \textit{boundedly rational, miyopic} follower approximately maximizes its one-step return. For every state \(\genericstate\), we define utility matrices \(\leaderutilitymatrix(\genericstate)_{ij}:= \leaderfunction(\genericstate, i, j) = \) and \(\followerutilitymatrix(\genericstate)_{ij}:= \followerfunction(\genericstate, i, j)\).
The leader's decision problem is to maximize
\begin{align*}
 &\mathbb{E}\left[ \sum_{k=1}^{K} \sum_{t=0}^{\tau} \leaderparameter(\randomversion{\genericstate}_{t})^{\top} \leaderutilitymatrix(\randomversion{\genericstate}_{t}) \randomversion{\followerdistribution}_{\genericinteraction}^{*}(\randomversion{\genericstate}_{t})\right]   \\ & 
\quad \text{s.t. } \randomversion{\followerdistribution}_{\genericinteraction}^{*}(\randomversion{\genericstate}_{t}) = \softmax_{\rationalityconstant}(\followerutilitymatrix(\randomversion{\genericstate}_{t})^{\top} \emprical{\randomversion{\leaderdistribution}}_{\genericinteraction}(\randomversion{\genericstate}_{t})).
 \end{align*} 

We note that different from the static case, the leader aims to optimize its expected return for the whole interaction.

\color{black}
\subsection{An Example Where Inferable Behavior is Necessary: Autonomous Car and Pedestrian Interaction} \label{section:carpedestrianexample}
This section describes a motivating example of the interaction between an autonomous car and a fully-rational pedestrian. The interaction is modeled as a bimatrix game. Similar scenarios have been considered in \cite{di2020liability,millard2018pedestrians}. 

Consider an autonomous car moving in its lane. A pedestrian is dangerously close to the road and aims to cross. The car has the right of way and wants to proceed, as an unnecessary stop is inefficient. However, if the pedestrian decides to cross the car may need to make a dangerous emergency stop. In the event the pedestrian crosses, they may get fined for jaywalking. The pedestrian and the car must make simultaneous decisions that will determine the outcome. Since the autonomous car's software is fixed prior to deployment, the car's decision is drawn from a fixed strategy that does not change over time.

    \begin{table}[h]
    \caption{Utilities for the car and pedestrian interaction}\label{tab:bimarix}
    \centering
\begin{tabular}{cccc}
          &                              & \multicolumn{2}{l}{Pedestrian's actions}                 \\ \cline{3-4} 
          & \multicolumn{1}{c|}{}        & \multicolumn{1}{c|}{Wait}  & \multicolumn{1}{c|}{Cross}  \\ \cline{2-4} 
\multicolumn{1}{c|}{\multirow{2}{*}{\begin{tabular}[c]{@{}c@{}}Car's\\ actions\end{tabular}}} & \multicolumn{1}{c|}{Stop} & \multicolumn{1}{c|}{(0,2)} & \multicolumn{1}{c|}{(0,1)} \\ \cline{2-4} 
\multicolumn{1}{c|}{}                                                                         & \multicolumn{1}{c|}{Proceed}    & \multicolumn{1}{c|}{(2,0)}  & \multicolumn{1}{c|}{(-8,1)} \\ \cline{2-4} 
\end{tabular}
\end{table}

For the pedestrian, crossing has a value \(r^{\follower}_{\mathrm{cr}} = 2\), and potentially getting fined for jaywalking has \(r^{\follower}_{\mathrm{jw}} = -1\) value. For the car, proceeding without a stop has a value \(r^{\leader}_{\mathrm{pr}} = 2\), and making an emergency stop has \(r^{\leader}_{\mathrm{em}} = -8\) value. 

 \textbf{Scenario 1:} The car stops, and the pedestrian waits for the car. The pedestrian's return is \(r^{\follower}_{\mathrm{cr}} = 2\) since the car's stop allows them to cross. The car's return is \(0\) since it does not proceed and stops unnecessarily. 

 \textbf{Scenario 2:} 
  The car stops, but the pedestrian crosses before the car yields the right of way. In this case, the pedestrian's return is \(r^{\follower}_{\mathrm{cr}} + r^{\follower}_{\mathrm{jw}} = 1\) since they cross the road but risk being fined for jaywalking. Once again, the car receives a return of \(0\) since it does not proceed.  

 \textbf{Scenario 3:} The car proceeds, and the pedestrian  waits. The car's return is \(r^{\leader}_{\mathrm{pr}} = 2\) since it makes no unnecessary stops. The pedestrian  gets a return of \(0\)  since it can not cross.

\textbf{Scenario 4:} The car proceeds, and the pedestrian crosses. The pedestrian's return is \(r^{\follower}_{\mathrm{cr}}+r^{\follower}_{\mathrm{jw}} = 1\). 
While the car proceeds, it makes an emergency stop due to the crossing pedestrian, resulting in a return of \( r^{\leader}_{\mathrm{em}} =-8\).

Assume that the car stops with probability \(\genericprobability\) and proceeds with probability \(1-\genericprobability\). If the pedestrian knows the probability \(\genericprobability\), the pedestrian would wait if \(2 \genericprobability +0(1-\genericprobability) > 1 \genericprobability +1(1-\genericprobability)\), i.e., \(\genericprobability >0.5\). Knowing that the pedestrian would wait when \(\genericprobability>0.5\), the car gets a return of \(0 \genericprobability + 2(1-\genericprobability)\).  Knowing that the pedestrian would cross when \(\genericprobability <0.5\), the car gets a return of \(0 \genericprobability  -8(1-\genericprobability)\). Hence, it is optimal for the car to choose a \(\genericprobability\) such that \(\genericprobability > 0.5\) and \(\genericprobability \approx 0.5\). While such a strategy is optimal and has a return of \(\approx 1\) for the car, it relies on the fact that the pedestrian has full information of the car's strategy. Such a strategy may not be optimal if the pedestrian does not know \(\genericprobability\) and relies on observations. 

Consider a scenario where the pedestrian and car will interact a certain number of times. The pedestrian estimates the car's fixed strategy using observations from previous interactions. 
If, in most of the previous interactions, the car stopped, the pedestrian would expect the car to stop in the next interaction. 
Knowing that the pedestrian relies on observations, the car should choose an easily inferable strategy. 

If the pedestrian has a good estimate \(\emprical{\genericprobability}\) of the car's strategy, the pedestrian will act optimally with respect to the car's actual strategy. On the flip side, the car will get a return that is close to the full information case. 
For example, consider that \(\genericprobability = 1\). In this case, the pedestrian has the correct estimate \(\emprical{\genericprobability} = 1\) after a single interaction, and the car will get a return of \(0\) in the subsequent interactions. 

If the car's strategy is not easily inferable, then the car may suffer from an \textit{inferability gap}. For example, if \(\genericprobability\) is such that \(\genericprobability > 0.5\) and \(\genericprobability \approx 0.5\), the car has an expected return of \(\approx -1.5\) in the second interaction. This is significantly lower compared to the expected return of \(1\) in the full information case. The difference is because the pedestrian's estimate will be \(\emprical{\genericprobability} = 0\) with probability \(1-\genericprobability \approx 0.5\). In those events, the pedestrian will cross, and the car will get \(-8\) return if it proceeds. Overall, a strategy that maximizes the car's expected return over a finite number of interactions should take the pedestrian's estimation errors into account.

A similar inferability gap happens for the followers that classify. 
    Consider a pedestrian that classifies cars into two types: cars that stop with probability \(0.99\) and cars that proceed with probability \(0.99\). For these classes, the optimal action for the pedestrian is to cross and wait, respectively. For such a classification, in the full information setting, it is still optimal for the car to choose a \(\genericprobability\) such that \(\genericprobability > 0.5\) and \(\genericprobability \approx 0.5\). However, under such a strategy, the follower's classification of the leader's strategy alternates frequently. Similar to the fully rational setting, the leader's expected return under inference is \(\approx -1.5\) in the second interaction.

%% file: performancebounds.tex
\section{Performance Bounds Under Inference}

In this section, we compare the leader's expected utility under repeated interactions with inference with the leader's expected utility under repeated interactions with full information in different settings. With an abuse of notation, for different settings, we use \(\empricalreturn_{\genericinteraction}(\leaderdistribution)\) to denote the leader's expected return at interaction \(k\) under inference and \(\stackelbergreturn(\leaderdistribution) \) to denote the leader's expected return under full information under strategy (parameter) \(\leaderdistribution\). The leader's expected return at the first interaction is arbitrary since the follower does not have any action samples. Hence, we are interested in analyzing the expected cumulative return for interactions \(\genericinteraction=2, \ldots, \numofinteractions\). Let \(v^{\mathrm{Stck}}(\leaderdistribution)\) be the leader's expected cumulative return in the full information setting and \(v^{\mathrm{Infr}}(\leaderdistribution)\) be the leader's expected cumulative return in the inference setting for interactions  \(\genericinteraction=2, \ldots, \numofinteractions\).
Due to the linearity of expectation, the expected cumulative return can be represented as a sum of expected returns of every interaction, i.e.,  \(v^{\mathrm{Stck}}(\leaderdistribution) = (\numofinteractions-1) \stackelbergreturn(\leaderdistribution)\) and \(v^{\mathrm{Infr}}(\leaderdistribution) = \sum_{\genericinteraction=2}^{\numofinteractions} \empricalreturn_{\genericinteraction}(\leaderdistribution)\). Consequently, we are interested in analyzing the cumulative gap \((\numofinteractions-1) \stackelbergreturn(\leaderdistribution) - \sum_{\genericinteraction=2}^{\numofinteractions} \empricalreturn_{\genericinteraction}(\leaderdistribution)\).

For the bimatrix games, without loss of generality, we assume that the utility matrices are normalized, i.e.,   \( \max_{i, j} \followerutilitymatrix_{i,j}-\min_{i, j} \followerutilitymatrix_{i,j} = \max_{i, j} \leaderutilitymatrix_{i,j}-\min_{i, j} \leaderutilitymatrix_{i,j} = 1\).

\color{black}
\noindent \textbf{Summary of the technical results:}
In Sections \ref{sec:resstaticparametric}-\ref{sec:resdynamicdiscreteentropy}, we provide lower bounds on the leader's return under inference $IR_{k}(x)$ for various settings depending on the action spaces, the horizon length, and the follower type. The bounds share two elements other than the follower's return under full information $SR(x)$: (i) the interaction number\footnote{\textcolor{black}{Theorem \ref{thm:paramstatic} does not directly depend on the interaction number. However, the MSE of the follower's estimate inherently depends on the interaction number. Consider that the follower has an efficient estimator of the leader's parameter. In this case, the MSE of the follower's estimate is inversely proportional to the number of interactions since the leader's actions are i.i.d.. Consequently, one can obtain a bound that grows as $\mathcal{O}(\sqrt{K})$ and depends on the stochasticity level of the leader's action distribution.}} and (ii) the stochasticity (concentrability) of the leader's strategy. 

\textit{Dependency on the interaction number: }At interaction \(k+1\), the inferability gap scales with \(\nicefrac{1}{\sqrt{k}}\) meaning that it decreases with each interaction, and the cumulative inferability gap grows as $\sqrt{\mathcal{O}(K)}$. 

\textit{Dependency on the stochasticity (or concentrability) of the leader's strategy:} As the leader's action distribution becomes less stochastic, the upper bound on the inferability loss decreases, implying that the return under inference gets closer to the return under full information. The bounds are asymptotically vanishing: as the leader's strategy becomes deterministic, the inferability gap approaches zero. We remark that while the inferability gap approaches zero, deterministic strategies do not necessarily maximize the return under inference. 

For repeated games where the leader's return may depend on the inferability of its strategy, the strategy maximizing the return under full information is not necessarily optimal if the number of interactions is limited. In these scenarios, one can regulate the stochasticity level of the strategy to improve the return under inference. While such regulated strategies are, in general, suboptimal in the limit (when the follower has full information of the leader's strategy), they may maximize the cumulative return under inference in a limited number of interactions. We demonstrate the effects of regulating the total variance through a numerical example in Section \ref{sec:experimentrandomgames}.

\color{black}

\subsection{Achievability Bound for Repeated Parametric Static Stackelberg Games Against Fully Rational Followers} \label{sec:resstaticparametric}

Let \(\emprical{\randomversion{\leaderdistribution}}_{\genericinteraction}\) be the random variable denoting follower's estimation of the leader's parameter \(\leaderdistribution\) after \(\genericinteraction-1\) interactions based on the leader's previous actions \(\randomversion{\leadergenericcontaction}_{1}, \ldots, \randomversion{\leadergenericcontaction}_{\genericinteraction-1}\).

In this setting, at interaction \(\genericinteraction\), we have \begin{align*}
\empricalreturn_{\genericinteraction}(\leaderdistribution) &= 
\expectation{\leaderfunction(\randomversion{\leadergenericcontaction}_{k},\randomversion{\followergenericcontaction}_{k})}{\randomversion{\leadergenericcontaction}_{k} \sim \leadercontdist(\leaderparameter), \randomversion{\followergenericcontaction}_{k} \sim \followercontdist(\randomversion{\followerparameter}^{*}_{\genericinteraction})}\\& 
\quad \text{s.t. } \randomversion{\followerdistribution}_{\genericinteraction}^{*} =  \arg\max_{\followerdistribution \in \followercontdistspace} \expectation{\followerfunction(\randomversion{\leadergenericcontaction},\randomversion{\followergenericcontaction})}{\randomversion{\leadergenericcontaction} \sim \leadercontdist(\emprical{\randomversion{\leaderdistribution}}_{\genericinteraction}), \randomversion{\followergenericcontaction} \sim \followercontdist(\followerparameter)},
\end{align*}
and
\begin{align*}
\stackelbergreturn(\leaderdistribution) &= 
\expectation{\leaderfunction(\randomversion{\leadergenericcontaction}_{k},\randomversion{\followergenericcontaction}_{k})}{\randomversion{\leadergenericcontaction}_{k} \sim \leadercontdist(\leaderparameter), \randomversion{\followergenericcontaction}_{k} \sim \followercontdist(\followerparameter^{*})} \\& 
\quad \text{s.t. } \followerdistribution^{*} =  \arg\max_{\followerdistribution \in \followercontdistspace} \expectation{\followerfunction(\randomversion{\leadergenericcontaction},\randomversion{\followergenericcontaction})}{\randomversion{\leadergenericcontaction} \sim \leadercontdist(\leaderdistribution), \randomversion{\followergenericcontaction} \sim \followercontdist(\followerparameter)}.
\end{align*}

If \(\emprical{\randomversion{\leaderdistribution}}_{\genericinteraction}\) converges to \(\leaderdistribution\) and the follower's optimal response is a continuous mapping, then the follower's response under inference converges to the follower's response under full information. Furthermore, if the leader's objective function is a continuous mapping, then the leader's expected return under inference \(\empricalreturn_{\genericinteraction}(\leaderdistribution)\) converges to the leader's return under full information \(\stackelbergreturn(\leaderdistribution)\). However, with a finite number of interactions, these returns are not necessarily the same, and the leader may suffer from an \textit{inferability gap}.

The following result shows that the inferability gap is upper bounded if the follower's estimation of the leader's parameter has a bounded mean squared error (MSE), the follower's response is Lipschitz continuous, and the leader's objective function is Lipschitz continuous. 

\begin{theorem} \label{thm:paramstatic}
    Let $MSE_{k}\geq 0$ be constants for \(k =2, 3,\ldots, K\). For a repeated parametric static Stackelberg game, if 
    \begin{enumerate}
        \item \(\mathbb{E} [ \| \emprical{\randomversion{\leaderdistribution}}_{\genericinteraction} - \leaderdistribution \|^{2} ] \leq MSE_{k}\), \label{condition:mse}

        \item \(\expectation{\leaderfunction(\randomversion{\leadergenericcontaction},\randomversion{\followergenericcontaction})}{\randomversion{\leadergenericcontaction} \sim \leadercontdist(\leaderparameter), \randomversion{\followergenericcontaction} \sim \followercontdist(\followerparameter)}\) is an \(L^{\leader}\)-Lipschitz function of \(\followerparameter\), and \label{condition:leaderlipschitz}

        \item \(\arg\max_{\followerdistribution \in \followercontdistspace} \expectation{\followerfunction(\randomversion{\leadergenericcontaction},\randomversion{\followergenericcontaction})}{\randomversion{\leadergenericcontaction} \sim \leadercontdist(\leaderdistribution'), \randomversion{\followergenericcontaction} \sim \followercontdist(\followerparameter)}\) is an \(L^{\follower}\)-Lipschitz function of \(\leaderdistribution'\) for all \(\leaderdistribution' \in \leadercontdistspace\), \label{condition:followerlipschitz}
    \end{enumerate}
    then \((K-1)\stackelbergreturn(\leaderdistribution) - \sum_{k=2}^{K} \empricalreturn_{\genericinteraction}(\leaderdistribution) \leq  \sum_{k=2}^{K}  L^{\leader} L^{\follower}  \sqrt{MSE_{k}}.\)

\end{theorem}

\begin{proof}[Proof of Theorem \ref{thm:paramstatic}]
Condition \ref{condition:followerlipschitz} implies  \[\| \followerdistribution^{*} - \randomversion{\followerdistribution}^{*}_{k} \| \leq L^{\follower} \| \leaderdistribution - \emprical{\randomversion{\leaderdistribution}}_{k} \| .\]
Combining this inequality with condition \ref{condition:leaderlipschitz}, we get 
\begin{align*}
&\left\lvert\expectation{\leaderfunction(\randomversion{\leadergenericcontaction},\randomversion{\followergenericcontaction})}{\randomversion{\leadergenericcontaction} \sim \leadercontdist(\leaderparameter), \randomversion{\followergenericcontaction} \sim \leadercontdist(\followerparameter)} - \expectation{\leaderfunction(\randomversion{\leadergenericcontaction},\randomversion{\followergenericcontaction})}{\randomversion{\leadergenericcontaction} \sim \leadercontdist(\leaderparameter), \randomversion{\followergenericcontaction} \sim \leadercontdist(\randomversion{\followerdistribution}^{*}_{k})}\right\rvert
\\
&\leq L^{\leader} L^{\follower} \| \leaderdistribution - \emprical{\randomversion{\leaderdistribution}}_{k}\| 
\end{align*}
which implies
\begin{align*}
&\expectation{\leaderfunction(\randomversion{\leadergenericcontaction},\randomversion{\followergenericcontaction})}{\randomversion{\leadergenericcontaction} \sim \leadercontdist(\leaderparameter), \randomversion{\followergenericcontaction} \sim \leadercontdist(\randomversion{\followerdistribution}^{*}_{k})} \geq 
\\
&\expectation{\leaderfunction(\randomversion{\leadergenericcontaction},\randomversion{\followergenericcontaction})}{\randomversion{\leadergenericcontaction} \sim \leadercontdist(\leaderparameter), \randomversion{\followergenericcontaction} \sim \leadercontdist(\followerparameter)} - L^{\leader} L^{\follower} \| \leaderdistribution - \emprical{\randomversion{\leaderdistribution}}_{k}\|. 
\end{align*}

Let \(h(t)\) be the p.d.f. of \(\| \emprical{\randomversion{\leaderdistribution}}_{\genericinteraction} - \leaderdistribution \|\). We have
\begin{subequations}
    \begin{align}
    &\empricalreturn_{\genericinteraction}(\leaderdistribution) = \expectation{ \expectation{\leaderfunction(\randomversion{\leadergenericcontaction},\randomversion{\followergenericcontaction})}{\randomversion{\leadergenericcontaction} \sim \leadercontdist(\leaderparameter), \randomversion{\followergenericcontaction} \sim \leadercontdist(\randomversion{\followerdistribution}^{*}_{k})} }{}  \\
    &= \int_{0}^{\infty} \expectation{ \expectation{\leaderfunction(\randomversion{\leadergenericcontaction},\randomversion{\followergenericcontaction})}{\randomversion{\leadergenericcontaction} \sim \leadercontdist(\leaderparameter), \randomversion{\followergenericcontaction} \sim \leadercontdist(\randomversion{\followerdistribution}^{*}_{k})} \Big| \|\emprical{\randomversion{\leaderdistribution}}_{\genericinteraction} - \leaderdistribution \| = t  }{}   h(t) dt \label{eqn:integraldefofexpectation}
    \\
    & \geq \int_{0}^{\infty} \left(\expectation{\leaderfunction(\randomversion{\leadergenericcontaction},\randomversion{\followergenericcontaction})}{\randomversion{\leadergenericcontaction} \sim \leadercontdist(\leaderparameter), \randomversion{\followergenericcontaction} \sim \leadercontdist(\followerparameter)} - L^{\leader} L^{\follower} \| \leaderdistribution - \emprical{\randomversion{\leaderdistribution}}_{k}\|\right) h(t) dt \label{eqn:useofbound}
    \\
    & =  \stackelbergreturn(\leaderdistribution) - L^{\leader} L^{\follower}\int_{0}^{\infty} t h(t) dt \label{eqn:pdfaddsupto1}
\end{align}
\end{subequations}
where \eqref{eqn:integraldefofexpectation} is due to the definition of expectation, \eqref{eqn:useofbound} is due to the above bound, and \eqref{eqn:pdfaddsupto1} is due to the definition of \(\stackelbergreturn\) and \(\int_{t=0}^{\infty} h(t) = 1\). Note that 
\begin{align} \label{ineq:mseineq}
    \mathbb{E} [ \| \emprical{\randomversion{\leaderdistribution}}_{\genericinteraction} - \leaderdistribution \|^{2} ] &= \int_{0}^{\infty}  t^{2} h(t) dt \geq  \left( \int_0^{\infty} t h(t) dt \right)^{2} 
\end{align}
due to Jensen's inequality. Condition \ref{condition:mse} in \eqref{eqn:pdfaddsupto1} and \eqref{ineq:mseineq} imply
\begin{align*}
    \empricalreturn_{\genericinteraction} (\leaderdistribution) &\geq \stackelbergreturn(\leaderdistribution) - L^{\leader} L^{\follower}  \sqrt{\left( \int_{0}^{\infty} t^{2} h(t) dt\right) }  \\
    &= \stackelbergreturn(\leaderdistribution) - L^{\leader} L^{\follower} \sqrt{MSE_{k}}
\end{align*}
Summation from \(k=2\) to \(K\) yields the desired result.
\end{proof}

\color{black}

\subsection{Achievability Bound for Repeated Bimatrix Stackelberg Games Against Boundedly Rational Followers with Maximum Entropy Response} \label{sec:rbmaxent} \label{sec:resdiscretestaticentropy}

 In this setting, we define
\[\empricalreturn_{\genericinteraction}(\leaderdistribution):= \expectation{  \leaderdistribution^{\top} \leaderutilitymatrix \randomversion{\followerdistribution}_{\genericinteraction} }{} \quad \text{s.t. } \randomversion{\followerdistribution}_{\genericinteraction} = \softmax_{\rationalityconstant}(\followerutilitymatrix^{\top} \emprical{\randomversion{\leaderdistribution}}_{k}), \] and \[\stackelbergreturn(x):= \leaderdistribution^{\top} \leaderutilitymatrix \followerdistribution^{*} \text{ s.t. } \followerdistribution^{*} = \softmax_{\rationalityconstant}(\followerutilitymatrix^{\top} \leaderdistribution).\] 

The inferability gap in this setting can be bounded directly using Theorem \ref{thm:paramstatic} due to three facts: (i)
Since the leader's actions \(\leadergenericaction_{1}, \ldots, \leadergenericaction_{k-1}\) are i.i.d., the plug-in estimator satisfies \(\expectation{\|\randomversion{\emprical{\leaderdistribution}_{\genericinteraction}} - \leaderdistribution\|^2}{} = \nicefrac{\nu(\leaderdistribution)^2}{\genericinteraction - 1}.\) (ii) The leader's expected return is a \(\nicefrac{\sqrt{\followernumofactions}}{2}\)-Lipschitz continuous function of the follower's strategy. (iii) The follower's maximum entropy optimal response is a \(\nicefrac{\rationalityconstant\sqrt{ \followernumofactions \leadernumofactions }}{2}\)-Lipschitz continuous function of the leader's empirical action distribution.

\begin{corollary}\label{thm:achieveability}
For a repeated bimatrix Stackelberg games with a boundedly rational follower with maximum entropy response,
\[(K-1)\stackelbergreturn(\leaderdistribution) - \sum_{k=2}^{K} \empricalreturn_{\genericinteraction}(\leaderdistribution) \leq  \sum_{k=2}^{K} \frac{ \rationalityconstant \followernumofactions \sqrt{\leadernumofactions } \nu(\leaderdistribution)}{4\sqrt{(k-1) }}.\]
\end{corollary}

We use the following lemmas to prove the corollary. \textcolor{black}{We give the proofs to the lemmas in the appendix.}

\begin{lemma} \label{lemma:varianceofemprical}
    \[\mathbb{E} [ \| \emprical{\randomversion{\leaderdistribution}}_{\genericinteraction} - \leaderdistribution \|^{2} ] = \frac{\nu(\leaderdistribution)^2}{\genericinteraction-1}.\]
\end{lemma}

\begin{lemma} \label{lemma:returncloseifxclose}
Let \(\followerdistribution_{k} = \softmax_{\rationalityconstant}(\followerutilitymatrix^{\top} \emprical{\leaderdistribution}_{k})\) and \(\followerdistribution = \softmax_{\rationalityconstant}(\followerutilitymatrix^{\top} \leaderdistribution)\).
    \begin{equation*}
    | \leaderdistribution^{\top} \leaderutilitymatrix \followerdistribution - \leaderdistribution^{\top} \leaderutilitymatrix \followerdistribution_{k} | \leq \frac{\sqrt{\followernumofactions}}{2} \| \followerdistribution - \followerdistribution_{k} \|.
\end{equation*}
\end{lemma}

\begin{lemma} \label{lemma:ycloseifxclose}
    \begin{equation}
    \| \softmax_{\rationalityconstant}(\followerutilitymatrix^{\top} \leaderdistribution) -\softmax_{\rationalityconstant}(\followerutilitymatrix^{\top} \emprical{\leaderdistribution}_{k}) \| \leq \frac{\lambda \sqrt{\followernumofactions  \leadernumofactions } }{2}  \| \emprical{\leaderdistribution}_{\genericinteraction} - \leaderdistribution \| .
\end{equation}
\end{lemma}

\begin{proof}[Proof of Corollary \ref{thm:achieveability}]
 The result directly follows from Theorem \ref{thm:paramstatic} since Lemmas \ref{lemma:varianceofemprical}, \ref{lemma:returncloseifxclose}, and \ref{lemma:ycloseifxclose}, satisfy conditions 1, 2, and 3 in Theorem \ref{thm:paramstatic}, respectively.   
\end{proof}

\begin{remark}
    There are \(\frac{(\genericinteraction+\leadernumofactions-2)!}{(\genericinteraction-1)!(\leadernumofactions-1)!} \approx (\genericinteraction-1)^{\leadernumofactions-1}\) (assuming \(\genericinteraction \gg \leadernumofactions)\) different values of \(\emprical{\leaderdistribution}_{k}\). Computing the exact value of \(\empricalreturn\) may require evaluating the expected return under all possible realizations of \(\emprical{\leaderdistribution}_{k}\).
\end{remark}
The cumulative inferability gap grows sublinearly, i.e, \[\sum_{k=2}^{K} \frac{\rationalityconstant \followernumofactions \sqrt{\leadernumofactions  } \nu(\leaderdistribution)}{4\sqrt{(k-1) }} = \mathcal{O}\left(\sqrt{K}\rationalityconstant \nu(\leaderdistribution)\right).\]

We note that the MSE of the follower's estimator depends on the stochasticity level of the leader's strategy. As the leader's strategy becomes deterministic, i.e., \(\nu(\leaderdistribution) \to 0\), the inferability gap vanishes to \(0\). In the extreme case where the leader's strategy is deterministic \(\nu(\leaderdistribution)=0\), the leader does not suffer from an inferability gap. As the follower becomes irrational, i.e., \(\rationalityconstant \to 0\), the inferability gap vanishes to \(0\), and when the follower is fully irrational, \(\rationalityconstant = 0\), the leader does not suffer from an inferability gap since the follower's strategy is uniformly random and does not depend on observations. 

The leader's optimal strategy under inference depends on various factors. Such a strategy should have a balance between having a high Stackelberg return under full information and having a minimal inferability gap, i.e., efficiently conveying information about itself to the follower.

\subsection{Achievability Bound for Repeated Bimatrix Stackelberg Games against Fully Rational Followers} \label{sec:resdiscretestaticfully}
In this setting, we define
\begin{align*}
\empricalreturn_{\genericinteraction}(\leaderdistribution) :=& 
\expectation{
       \leaderdistribution^{\top} \leaderutilitymatrix \randomversion{\followerdistribution}_{\genericinteraction}^{*}}{}    \quad       \text{s.t. } \randomversion{\followerdistribution}_{\genericinteraction}^{*} = \arg\max_{\followerdistribution \in \simplex^{\followernumofactions}} \emprical{\randomversion{\leaderdistribution}}^{\top}_{\genericinteraction}\followerutilitymatrix \followerdistribution,
\end{align*}
and 
\begin{align*}
\stackelbergreturn(x):=  &  \ \leaderdistribution^{\top} \leaderutilitymatrix \followerdistribution^{*}   & \text{  s.t. }
 \quad \followerdistribution^{*}  = \arg\max_{\followerdistribution \in \simplex^{\followernumofactions}} \leaderdistribution^{\top} \followerutilitymatrix \followerdistribution.
 \end{align*}

 Unlike the boundedly rational followers, a fully rational follower's optimal strategy is not a continuous function of the leader's strategy. For example, the pedestrian's optimal strategy as a function of the leader's strategy has a discontinuity in the motivating example. Hence, Theorem \ref{thm:achieveability} is not directly applicable to this setting. 

While the follower's optimal strategy is not a continuous function of the leader's strategy, it is a piecewise constant function of it. In words, the follower's optimal strategy is the same within subdomains of the leader's strategy simplex. Let \(e_{i}\) denote a probability vector such that the $i$-th entry is 1 and the others are 0. Strategy profile \(e_{i}\) is optimal if and only if \(\leaderdistribution^{\top} \leaderutilitymatrix e_{i} \geq \leaderdistribution^{\top} \leaderutilitymatrix e_{j}\) for all \(j \in [\followernumofactions]\). Let \(C_{i} = \lbrace x \ |\  \forall j \in [\followernumofactions], \leaderdistribution^{\top} \leaderutilitymatrix e_{i} \geq \leaderdistribution^{\top} \leaderutilitymatrix e_{j}\rbrace \), i.e.,  the set of leader strategies such that strategy profile \(e_{i}\) is optimal for the follower. 
 
 In the inference setting, if the leader has a strategy profile \(x\) such that \(e_{i}\) is uniquely optimal, i.e., \(x\in C_{i}\) and \(x \not \in C_{j}\) for every \(j \in [\followernumofactions] \setminus \lbrace i \rbrace\),  then \(\randomversion{\followerdistribution}^{*}_{\genericinteraction}\) almost surely converges to \(e_{i}\) as \(\randomversion{\emprical{\leaderdistribution}}_{\genericinteraction}\) falls in only \(C_{i}\) with probability \(1\). 

Inspired by the piecewise constant property and the convergence of \(\randomversion{\emprical{\leaderdistribution}}_{\genericinteraction}\) to \(\leaderdistribution\), we show that the inferability gap is bounded for a finite number of actions. We first note that the boundary between \(C_{i}\) and \(C_{j}\) is a hyperplane for every \(i\neq j \in [\followernumofactions]\). Let \(e_{i}\) be uniquely optimal for \(\leaderdistribution\) and \(d(x)\) denote the L2 distance to the closest boundary, i.e., \(\min_{j \in [\followernumofactions] \setminus \lbrace i \rbrace} \min_{x' \in C_{j}} \| x - x' \|\). We note that \(d(x)\) has an analytical expression. We have \(\expectation{\|\randomversion{\emprical{\leaderdistribution}_{\genericinteraction}} - \leaderdistribution\|^2}{} = \nicefrac{\nu(\leaderdistribution)^2}{\genericinteraction - 1}\), which implies \(\Pr({\|\randomversion{\emprical{\leaderdistribution}_{\genericinteraction}} - \leaderdistribution\| \leq d(x)}) = \nicefrac{\nu(\leaderdistribution)}{\sqrt{\genericinteraction - 1}d(x)}\) due to the Markov bound. Using this, we get the following bound on the inferability gap.

\begin{theorem}
    \label{thm:fullyrational} For a repeated bimatrix Stackelberg game with a fully rational follower, we have
\[(K-1)\stackelbergreturn(\leaderdistribution) - \sum_{k=2}^{K} \empricalreturn_{\genericinteraction}(\leaderdistribution) \leq  \sum_{k=2}^{K} \frac{ \nu(\leaderdistribution)}{ d(\leaderdistribution) \sqrt{(k-1) } }.\]
    
\end{theorem}

\begin{proof}[Proof of Theorem \ref{thm:fullyrational}]
Let \(i\) be the optimal action for the follower given the leader's strategy \(\leaderdistribution\). We have 
\begin{align*}
    \empricalreturn_{\genericinteraction}(\leaderdistribution) &= (1-\Pr(\emprical{\randomversion{\leaderdistribution}}_{\genericinteraction} \not \in C_{i})) \expectation{
       \leaderdistribution^{\top} \leaderutilitymatrix \randomversion{\followerdistribution}_{\genericinteraction}^{*} | \emprical{\randomversion{\leaderdistribution}}_{\genericinteraction} \in C_{i}}{} 
       \\
       &\ +  \Pr(\emprical{\randomversion{\leaderdistribution}}_{\genericinteraction} \not\in C_{i}) \expectation{
       \leaderdistribution^{\top} \leaderutilitymatrix \randomversion{\followerdistribution}_{\genericinteraction}^{*} | \emprical{\randomversion{\leaderdistribution}}_{\genericinteraction} \not\in C_{i}}{}
\end{align*}

We note that \(\emprical{\randomversion{\leaderdistribution}}_{\genericinteraction} \in C_{i}\) implies that \(\randomversion{\followerdistribution}_{\genericinteraction}^{*} = e_{i} = \followerdistribution\) and hence  \(\expectation{
       \leaderdistribution^{\top} \leaderutilitymatrix \randomversion{\followerdistribution}_{\genericinteraction}^{*} | \emprical{\randomversion{\leaderdistribution}}_{\genericinteraction} \in C_{i}}{}  = 
       \leaderdistribution^{\top} \leaderutilitymatrix \followerdistribution = \stackelbergreturn(\leaderdistribution)\).
We have 
\begin{align*}
    &\stackelbergreturn(\leaderdistribution) - \empricalreturn_{\genericinteraction}(\leaderdistribution)
    \\
    &= \Pr(\emprical{\randomversion{\leaderdistribution}}_{\genericinteraction} \not\in C_{i}) (\expectation{
       \leaderdistribution^{\top} \leaderutilitymatrix \randomversion{\followerdistribution}_{\genericinteraction}^{*} | \emprical{\randomversion{\leaderdistribution}}_{\genericinteraction} \in C_{i}}{}  - \expectation{
       \leaderdistribution^{\top} \leaderutilitymatrix \randomversion{\followerdistribution}_{\genericinteraction}^{*} | \emprical{\randomversion{\leaderdistribution}}_{\genericinteraction} \not\in C_{i}}{})
\end{align*}
Note that
\begin{align*}
&\max_{\leaderdistribution \in \simplex^{\leadernumofactions}} \expectation{
        \leaderdistribution^{\top} \leaderutilitymatrix \randomversion{\followerdistribution}_{\genericinteraction}^{*} | \emprical{\randomversion{\leaderdistribution}}_{\genericinteraction} \in C_{i}}{}  - \expectation{
       \leaderdistribution^{\top} \leaderutilitymatrix \randomversion{\followerdistribution}_{\genericinteraction}^{*} | \emprical{\randomversion{\leaderdistribution}}_{\genericinteraction} \not\in C_{i}}{}
       \\
       &\leq \max_{i,j} \leaderutilitymatrix_{ij} - \min_{i,j} \leaderutilitymatrix_{i,j}= 1.
\end{align*}
Therefore, \(\stackelbergreturn(\leaderdistribution) - \empricalreturn_{\genericinteraction}(\leaderdistribution) \leq \Pr(\emprical{\randomversion{\leaderdistribution}}_{\genericinteraction} \not\in C_{i}) \). The event \(\emprical{\randomversion{\leaderdistribution}}_{\genericinteraction} \not\in C_{i}\) happens only if \(\|\randomversion{\emprical{\leaderdistribution}_{\genericinteraction}} - \leaderdistribution\| \geq d(\leaderdistribution)\) which happens with prob. at most \(\frac{\expectation{\|\randomversion{\emprical{\leaderdistribution}_{\genericinteraction}} - \leaderdistribution\|}{}}{d(x)}\) due to the Markov bound. We have \[\stackelbergreturn(\leaderdistribution) - \empricalreturn_{\genericinteraction}(\leaderdistribution) \leq \frac{\expectation{\|\randomversion{\emprical{\leaderdistribution}_{\genericinteraction}} - \leaderdistribution\|}{}}{d(\leaderdistribution)}.\] Due to Jensen's inequality we get \[\stackelbergreturn(\leaderdistribution) - \empricalreturn_{\genericinteraction}(\leaderdistribution) \leq \frac{\expectation{\|\randomversion{\emprical{\leaderdistribution}_{\genericinteraction}} - \leaderdistribution\|}{}}{d(\leaderdistribution)} \leq \frac{\sqrt{\expectation{\|\randomversion{\emprical{\leaderdistribution}_{\genericinteraction}} - \leaderdistribution\|^2}{}}}{d(\leaderdistribution)} .\] Using Lemma \ref{lemma:varianceofemprical} and summation from \(k=2\) to \(K\) yields the desired result.
\end{proof}

\color{black}
We note that we can derive an asymptotically tighter bound using the Chernoff bound for the convergence of \(\emprical{\leaderdistribution}_{\genericinteraction}\) to \(\leaderdistribution\) instead of the Markov bound. Using $\|\randomversion{\emprical{\leaderdistribution}_{\genericinteraction}} - \leaderdistribution\| \leq \|\randomversion{\emprical{\leaderdistribution}_{\genericinteraction}} - \leaderdistribution\|_{1}$ and Theorem 2.1 of \cite{weissman2003inequalities}, we can get the result \[\sum_{k=2}^{K}\left(\stackelbergreturn(\leaderdistribution) -  \empricalreturn_{\genericinteraction}(\leaderdistribution) \right) \leq  \sum_{k=2}^{K} 2^{\leadernumofactions} \exp\left(-\frac{(k-1) \varphi(x) d(\leaderdistribution)^2 }{4}\right),\] which implies the cumulative inferability gap is bounded by a constant depending on $m$, $\varphi(\leaderdistribution)$, and $d(\leaderdistribution)$ regardless of $K$.
\color{black}

\subsection{Achievability Bound for Repeated Bimatrix Stackelberg Games Against Classifying Boundedly Rational Followers} \label{sec:resdiscretestaticcategorizing}
In this setting, we define
\begin{align*}
\empricalreturn_{\genericinteraction}(\leaderdistribution) :=& 
\expectation{
       \leaderdistribution^{\top} \leaderutilitymatrix \randomversion{\followerdistribution}_{\genericinteraction}^{*}}{}    \quad       \text{s.t. } \randomversion{\followerdistribution}_{\genericinteraction}^{*} = \arg\max_{\followerdistribution \in \simplex^{\followernumofactions}} \categorizationfunction(\emprical{\randomversion{\leaderdistribution}}_{\genericinteraction})^{\top}\followerutilitymatrix \followerdistribution,
\end{align*}
and 
\begin{align*}
\stackelbergreturn(x):=  &  \ \leaderdistribution^{\top} \leaderutilitymatrix \followerdistribution^{*}   & \text{  s.t. }
 \quad \followerdistribution^{*}  = \arg\max_{\followerdistribution \in \simplex^{\followernumofactions}} \categorizationfunction(\leaderdistribution)^{\top} \followerutilitymatrix \followerdistribution
 \end{align*} where \(\categorizationfunction\) is the maximum likelihood classification function of the follower. 

 The classification function \(\categorizationfunction\) assigns data \(\leaderdistribution\) to \(\leaderdistribution^{i} \in \mathcal{T}\) if \(\Pr(\leaderdistribution|\leaderdistribution^{i}) > \Pr(\leaderdistribution|\leaderdistribution^{j})\) for all \(j \in [l] \setminus \lbrace i \rbrace\), i.e., the likelihood of data \(\leaderdistribution\) is maximized under \(\leaderdistribution^{i}\) among strategies in \(\mathcal{T} = \lbrace \leaderdistribution^{1}, \ldots, \leaderdistribution^{l} \rbrace\). We note that \(\log \Pr(\leaderdistribution|\leaderdistribution^{i}) = \langle \leaderdistribution, \log(\leaderdistribution^{i}) \rangle\), and \(\Pr(\leaderdistribution|\leaderdistribution^{i}) > \Pr(\leaderdistribution|\leaderdistribution^{j})\) implies \(\log {\Pr(\leaderdistribution|\leaderdistribution^{i})} - \log{\Pr(\leaderdistribution|\leaderdistribution^{j})} > 0 \), i.e., \(\langle \leaderdistribution, \log(\leaderdistribution^{i}) - \log ( \leaderdistribution^{j})\rangle \geq 0\). In words, similar to the fully rational setting, the classification boundaries are hyperplanes, and the optimal response of the follower is a piecewise linear function of the leader's strategy. Let \(C_{i} = \lbrace x \ |\  \forall j \in [\followernumofactions], \leaderdistribution^{\top} \log(\leaderdistribution^{i}) \geq \leaderdistribution^{\top} \log(\leaderdistribution^{j})\rbrace \), i.e.,  the set of leader's strategies that is classified into \(\leaderdistribution^{i}\). 
 
 In the inference setting, %
 if the leader has a strategy profile \(\leaderdistribution\) such that \( \arg \max_{\leaderdistribution^{i} \in \mathcal{T} } \Pr(\leaderdistribution|\leaderdistribution^{i})\) is unique,  then \(\randomversion{\followerdistribution}^{*}_{\genericinteraction}\) almost surely converges to \(\followerdistribution^{*}  = \arg\max_{\followerdistribution \in \simplex^{\followernumofactions}} \categorizationfunction(\leaderdistribution)^{\top} \followerutilitymatrix \followerdistribution\) as \(\randomversion{\emprical{\leaderdistribution}}_{\genericinteraction}\) falls in only \(C_{i}\) with probability \(1\). Let \(d(\leaderdistribution)\) denote the L2 distance to the closest boundary, i.e., \(\min_{j \in [l] \setminus \lbrace i \rbrace} \min_{\leaderdistribution' \in C_{j}} \| \leaderdistribution - x' \|\). Theorem \ref{thm:fullyrational} applies to this case with the new definition of \(d(\leaderdistribution)\).

\color{black}

\subsection{Achievability Bound for a Repeated Discrete Dynamic Stackelberg Games Against Myopic Boundedly Rational Followers with Maximum Entropy Response}  \label{sec:resdynamicdiscreteentropy}

In this setting, we define

\begin{align*}
&\empricalreturn_{\genericinteraction}(x):= \ \mathbb{E}\left[\sum_{t=0}^{\tau}\leaderfunction(\randomversion{\genericstate}_{t},\randomversion{\leadergenericcontaction}_{t},\randomversion{\followergenericcontaction}_{t}) \bigg| \followergenericcontaction_{t} \sim \randomversion{\followerdistribution}_{\genericinteraction}^{*}(\randomversion{\genericstate}_{t})\right]  \\ & \text{s.t. }
     \ \randomversion{\followerdistribution}_{\genericinteraction}^{*}(\randomversion{\genericstate}_{t}) = \softmax_{\rationalityconstant}(\followerutilitymatrix(\randomversion{\genericstate}_{t})^{\top} \emprical{\randomversion{\leaderdistribution}}_{\genericinteraction}(\randomversion{\genericstate}_{t}))
 \end{align*} 
 and
 \begin{align*}
&\stackelbergreturn(x):= \  \mathbb{E}\left[\sum_{t=0}^{\tau}\leaderfunction(\randomversion{\genericstate}_{t},\randomversion{\leadergenericcontaction}_{t},\randomversion{\followergenericcontaction}_{t}) \bigg| \followergenericcontaction_{t} \sim \followerdistribution^{*}(\randomversion{\genericstate}_{t})\right]   \\ & \text{s.t. }
     \ \followerdistribution^{*}(\randomversion{\genericstate}_{t}) = \softmax_{\rationalityconstant}(\followerutilitymatrix(\randomversion{\genericstate}_{t})^{\top} \leaderdistribution(\randomversion{\genericstate}_{t})).
 \end{align*} 

To show the closeness of \(\empricalreturn_{k}(x)\) and \(\stackelbergreturn(x)\), we follow a similar approach model-based off-policy evaluation for MDPs~\cite{karabag2023sample}. We note that for a myopic follower, the dynamic game is an MDP for the leader. With the increasing number of sample actions from a state \(\genericstate\), the follower's estimate \(\emprical{\randomversion{\leaderdistribution}}_{\genericinteraction}(\genericstate)\) converges to the leader's strategy \(\leaderdistribution(\genericstate)\). Since the follower has the maximum entropy response and is myopic, the follower's response under inference \(\randomversion{\followerdistribution}_{\genericinteraction}^{*}(\genericstate)\) converges to the follower's response under full information \(\followerdistribution^{*}(\genericstate)\). Using the closeness of \(\randomversion{\followerdistribution}_{\genericinteraction}^{*}(\genericstate)\) and \(\followerdistribution^{*}(\genericstate)\), we can show the the closeness of \(\empricalreturn_{k}(x)\) and \(\stackelbergreturn(x)\) if the game has the contraction property. We use a series of lemmas to prove this result and give the proof sketches in the appendix.

To ensure the boundedness of \(\empricalreturn(x)\) and \(\stackelbergreturn(x)\), we make the following contraction assumption.
\begin{assumption} \label{assumption:contraction}
    For all \(\followerdistribution\), \(\sum_{q\in \gamestatespace_{end}}\transitionfunc(\genericstate_{}, \leadergenericcontaction, \followergenericcontaction, q) \geq 1-\gamma\), where \(\gamestatespace_{end}\) is a return-free, invariant set of states.
\end{assumption}
This contraction assumption means that at every time step, the interaction ends with probability at least \(1-\gamma\) and is similar to having a discount factor of \(\gamma\). \textcolor{black}{For example, in a navigation example, it can mean that the target is reached or the battery runs out with a certain probability at every time step.} Under the contraction property, the well-known simulation lemma~\cite{strehl2008analysis,karabag2023sample} shows that if two strategies are close for every state, then the returns of these strategies are close.
\begin{lemma}[Lemma 1 from \cite{strehl2008analysis}] \label{lemma:simulation}
    Under Assumption \ref{assumption:contraction}, if  \(\|\randomversion{\followerdistribution}_{\genericinteraction}^{*}(\genericstate) - \followerdistribution^{*}(\genericstate)\|_{1} \leq \epsilon\) for all \(\genericstate \in \gamestatespace\), then \[|\stackelbergreturn(x) - \empricalreturn(x)|\leq \frac{\gamma \epsilon}{(1-\gamma)^2}.\] 
\end{lemma}

 To invoke Lemma \ref{lemma:simulation}, we need to show the closeness of  \(\randomversion{\followerdistribution}_{\genericinteraction}^{*}(\genericstate)\) and \(\followerdistribution^{*}(\genericstate)\). We note that Lemma \ref{lemma:ycloseifxclose} immediately applies to bound the closeness between \(\randomversion{\followerdistribution}_{\genericinteraction}^{*}(\genericstate)\) and \(\followerdistribution^{*}(\genericstate)\) as a function of \(\|\emprical{\randomversion{\leaderdistribution}}_{\genericinteraction}(\genericstate) - \leaderdistribution(\genericstate) \|\). However, different from Section \ref{sec:rbmaxent}, the sample actions of the leader are not independently sampled (since they come from interactions of a dynamic game), and Lemma \ref{lemma:varianceofemprical} does not apply. To bound \(\|\emprical{\randomversion{\leaderdistribution}}_{\genericinteraction}(\genericstate) - \leaderdistribution(\genericstate)\|\), we use a result from \cite{karabag2023sample} that can handle random stopping times.

 \begin{lemma}[Lemma 3 of \cite{karabag2023sample}, Theorem 2.1 of \cite{weissman2003inequalities}] \label{lemma:enoughestimation}
 Let \(w\) be the number of sample leader actions from state \(\genericstate\) in the first \(k-1\) interactions. With probability at least \(1-\delta\), \[\|\emprical{\randomversion{\leaderdistribution}}_{\genericinteraction}(\genericstate) - \leaderdistribution(\genericstate)\|_{1} \leq \epsilon \sqrt{\frac{2}{\varphi(\leaderdistribution(s))}}\] if 
 \(w \geq \frac{40 |\leaderactionspace|}{\epsilon^2} \log\left(\frac{1}{\epsilon}\right) \log\left(\frac{3}{5\delta}\right).\)
 \end{lemma}

To ensure that each state \(\genericstate\) has enough samples, we have the following assumption.
\begin{assumption}\label{assumption:exploring}
       Let \(\rho>0\) and \(\mu > 0\) be constants. Under any \(\followerdistribution\), for all \(\genericstate \in \gamestatespace \setminus \gamestatespace_{end}\),
    \[\Pr(\randomversion{\genericstate}_{0} = \genericstate \vee \randomversion{\genericstate}_{1} = \genericstate \ldots) \geq \rho\] and     \[\Pr(\randomversion{\genericstate}_{t+1} = \genericstate \vee \randomversion{\genericstate}_{t+2} = \genericstate \ldots | \randomversion{\genericstate}_{t} = \genericstate) \geq \mu.\] 
\end{assumption}
Assumption \ref{assumption:exploring} ensures that every non-terminal state is reached and resampled with high probability, and the occupancy measure of each state is at least \(\nicefrac{\rho}{(1-\mu)}\).
The following result is due to the concentration of Bernoulli and geometric random variables and shows that the number of samples from each state is lower bounded under Assumption \ref{assumption:exploring}. 
\begin{lemma}[Unification of Lemmas 4 and 5 from \cite{karabag2023sample}] \label{lemma:enoughsamples}
Under Assumption \ref{assumption:exploring}, if \[\genericinteraction \geq \frac{6 \max(8 w (1-\mu), \log(\nicefrac{2}{\delta}))}{\rho}\] the number of sample leader actions from state \(\genericstate\) in the first \(k-1\) interactions is at least \(w\) with probability at least \(1-\delta\).
\end{lemma}

Combining the above results, we get the following result on the gap between \(\stackelbergreturn_{k}(\leaderdistribution)\) and \(\empricalreturn(\leaderdistribution)\).
\begin{theorem} \label{thm:dynamic}
    Under Assumptions \ref{assumption:contraction} and \ref{assumption:exploring}, if \[k \geq \frac{6 \max(\frac{320 |\leaderactionspace|}{\epsilon^2} \log(\nicefrac{1}{\epsilon}) \log(\nicefrac{9 |\gamestatespace|}{5\delta}) (1-\mu), \log(\nicefrac{3 |\gamestatespace|}{\delta}))}{\rho},\] then with probability \(1-\delta\), \[ \stackelbergreturn(\leaderdistribution) - \empricalreturn_{\genericinteraction}(\leaderdistribution) \leq \frac{\rationalityconstant |\followeractionspace| \sqrt{|\leaderactionspace|} \gamma \epsilon}{\sqrt{2} (1-\gamma)^2 \min_{s} \sqrt{\varphi(\leaderdistribution(s))}}.\]
\end{theorem}
\begin{proof}[Proof of Theorem \ref{thm:dynamic}]
   Setting the failure probability to \(\nicefrac{\delta}{3|\gamestatespace|}\) in Lemma \ref{lemma:enoughestimation} and the failure probability to \(\nicefrac{2\delta}{3|\gamestatespace|}\) in Lemma \ref{lemma:enoughsamples}, and using the union bound, we get that \[\|\leaderdistribution_{\genericinteraction}(\randomversion{\genericstate}_{t}) - \emprical{\randomversion{\leaderdistribution}}_{\genericinteraction}(\randomversion{\genericstate}_{t})\|_{1} \leq \epsilon\sqrt{\frac{2}{\varphi(\leaderdistribution(s))}}\] for all \(\genericstate \in \gamestatespace \setminus \gamestatespace_{end}\) with probability at least \(1-\delta\). Using \(\|z\|_{2} \leq \|z\|_{1} \leq \sqrt{d} \|z\|_{2}\) and Lemma \ref{lemma:ycloseifxclose}, we have \[\|\followerdistribution^{*}(\randomversion{\genericstate}_{t}) - \emprical{\randomversion{\followerdistribution}}_{k}(\randomversion{\genericstate}_{t})\|_{1} \leq \frac{\rationalityconstant |\followeractionspace| \sqrt{|\leaderactionspace|} \epsilon}{\sqrt{{2\varphi(\leaderdistribution(s))}}}\] for all \(\genericstate \in \gamestatespace \setminus \gamestatespace_{end}\) with probability at least \(1-\delta\). Using this result in Lemma \ref{lemma:simulation} gives the desired result. 
\end{proof}

The inferability gap has the term \(\sqrt{\varphi(\leaderdistribution(s))}\) that measures the concentration of the leader's strategy for state \(s\). The term \(\frac{1}{\min_{s} \varphi(\leaderdistribution(s))}\) in the bound is analogous to the stochasticity level \(\nu\): As the strategies for every state get more deterministic, \(\frac{1}{\min_{s} \varphi(\leaderdistribution(s))}\) decreases making the inferability gap diminish. In the extreme case where the leader's strategy is deterministic for every state, we have \(1/\min_{s} \sqrt{\varphi(\leaderdistribution(s))} = 0\), which implies that the inferability gap is \(0\) with high probability after a certain number of interactions. 
 
Similar to the static setting \(\tilde{\mathcal{O}}(\nicefrac{1}{\epsilon^2})\) interactions are sufficient to achieve an inferability gap of \(\epsilon\) with high probability.

We also note that the result for the fully rational followers can also be extended to this setting by using the closeness between \(\randomversion{\followerdistribution}_{\genericinteraction}^{*}(\genericstate)\) and \(\followerdistribution^{*}(\genericstate)\) as described in Section \ref{sec:resdiscretestaticfully}.

%% file: converseresult.tex
\subsection{Converse Bound for Repeated Static Bimatrix Stackelberg Games with Fully Rational Followers}
\label{sec:staticfullyconverse}

Theorem \ref{thm:fullyrational} shows that with a fully rational follower, the gap between the leader's expected return in the full information setting and in the inference setting,  
\(\stackelbergreturn(\leaderdistribution) -\empricalreturn_{\genericinteraction}(\leaderdistribution)\), is at most at the order of \(\mathcal{O}(\nicefrac{1}{\sqrt{\genericinteraction}})\) at interaction \(\genericinteraction\).
In other words, after \(\mathcal{O}(\nicefrac{1}{\epsilon^{2}})\) interactions, we have \(\stackelbergreturn_(\leaderdistribution) - \empricalreturn_{\genericinteraction}(\leaderdistribution) \leq \epsilon\). In this section, we give an example for the fully rational follower setting that matches the upper bound: \(\mathcal{O}(\nicefrac{1}{\epsilon^{2}})\) interactions are required to achieve \(\stackelbergreturn(\leaderdistribution) - \empricalreturn_{\genericinteraction}(\leaderdistribution) \leq \epsilon\). We consider 
\begin{equation} \label{eqn:defofAandB}
    \leaderutilitymatrix = \begin{bmatrix}
0 & 0\\ 
1 & 0
\end{bmatrix}, \quad \followerutilitymatrix = \begin{bmatrix}
2 & 1\\ 
0 & 1
\end{bmatrix}.
\end{equation}
For these choices of \(\leaderutilitymatrix\) and \(\followerutilitymatrix\), we have \[\left(\sup_{\leaderdistribution \in \simplex^{\leadernumofactions}} \  \leaderdistribution^{\top} \leaderutilitymatrix \followerdistribution^{*} \ \text{  s.t. }
 \ \followerdistribution^{*}  =\arg\max_{\followerdistribution \in \simplex^{\followernumofactions}} \leaderdistribution^{\top} \followerutilitymatrix \followerdistribution  \right)=  \frac{1}{2}\]

\begin{proposition} \label{prop:converse}
    Let \(\leaderutilitymatrix\) and \(\followerutilitymatrix\) be as defined in \eqref{eqn:defofAandB}. For every \(\epsilon \in (0, 1/2)\) and \(\leaderdistribution \in \simplex^{2}\) such that \(\stackelbergreturn(\leaderdistribution) \geq (\nicefrac{1}{2}) - \epsilon\), if \[\genericinteraction \leq \frac{1-20\epsilon + 128\epsilon^2+80\epsilon^3 - 400 \epsilon^4}{32\epsilon^2},\] then \(\stackelbergreturn(\leaderdistribution) - \empricalreturn_{k}(\leaderdistribution) \geq \epsilon.\)
\end{proposition}

\begin{proof}[Proof of Proposition \ref{prop:converse}]
    We follow a proof similar to the proofs for bandit lower bounds. For a strategy with high return under full information, we choose an alternative, close strategy. We show that the alternative strategy has a low return in the inference setting, and these strategies have similar returns in the inference setting since they are not distinguishable. Therefore, the return of the strategy with high return under full information is low in the inference setting.

    We first define some notation. Let \(E^{1} = \lbrace \leaderdistribution | \leaderdistribution \in \simplex^{2}, \leaderdistribution^{\top} [1,1] > 0 \rbrace \) and \(E^{2} = \simplex^{2} \setminus E^{1}\). Note that \(E^{1}\) is the set of leader strategies for which action \(1\) is optimal for the follower, and \(E^{2}\) is the set of leader strategies for which action \(2\) is optimal for the follower. We consider two strategies: \[\leaderdistribution = \left[\frac{1}{2} + \epsilon, \, \frac{1}{2} - \epsilon \right] \text{ and }  z = \left[\frac{1}{2} - \epsilon, \, \frac{1}{2} + \epsilon \right].\]  These strategies are both $\epsilon$ away from the follower's decision boundary at $[\nicefrac{1}{2}, \nicefrac{1}{2}]$. We prove the statement in four steps:
    \begin{enumerate}
        \item Show that the expected return of \(\leaderdistribution\) under full information \(\stackelbergreturn(\leaderdistribution)\) is lower bounded.
        \item Show that the expected return of \(z\) under inference \(\empricalreturn_{\genericinteraction}(z)\) is upper bounded.
        \item Show that the expected returns of \(z\) and \(\leaderdistribution\) under inference, \(\empricalreturn_{\genericinteraction}(z)\) and \(\empricalreturn_{\genericinteraction}(\leaderdistribution)\) are close.
        \item Combine the above results.
    \end{enumerate}

    \textbf{Step 1:}     Note that \(x \in E^{2}\). Consequently, \(\stackelbergreturn(x) = (\nicefrac{1}{2})  - \epsilon\) which trivially implies    
    \begin{equation} \label{eqn:stackelbergreturnx}
        \stackelbergreturn(x) \geq \frac{1}{2} - \epsilon. 
    \end{equation}

    \textbf{Step 2:} Next, we upper bound \(\empricalreturn_{k}(z)\). We have
    \begin{align*}
        \empricalreturn_{\genericinteraction}(z) &= \sum_{i=1}^{2} \sum_{j=1}^{2} \Pr(\randomversion{\leadergenericaction}_{k} = i|z) \Pr(\randomversion{\followergenericaction}_{k} = j|z) \leaderutilitymatrix_{i,j} 
        \\
        =&  \left(\frac{1}{2} + \epsilon \right)\Pr(\randomversion{\followergenericaction}_{k} = 1|z) 
    \end{align*}
    since \(\leaderutilitymatrix_{1,1}=\leaderutilitymatrix_{1,2}=\leaderutilitymatrix_{2,2} = 0\), \(\leaderutilitymatrix_{2,1} = 1 \), and \(\Pr(\randomversion{\leadergenericaction}_{k} = 2|z) = (\nicefrac{1}{2}) + \epsilon \). We note that \(\randomversion{\followergenericaction}_{k} = 1\) if and only if \(\emprical{\randomversion{z}}_{k} \in E^{1}\), i.e., the empirical distribution of the leader's actions belongs to \(E^{1}\). It implies that \(\Pr(\emprical{\randomversion{z}}_{k} \in E^{1}|z) = \Pr(\randomversion{\followergenericaction}_{k} = 1|z)\).  \(\Pr(\emprical{\randomversion{z}}_{k} \in E^{1}|z) \leq \nicefrac{1}{2}\) since \(z\) has more bias towards action \(2\). We have 
    \begin{equation} \label{eqn:empricalreturnz}
        \empricalreturn_{\genericinteraction}(z) \leq \frac{1}{2} \left(\frac{1}{2} + \epsilon \right) = \frac{1}{4} + \frac{\epsilon}{2}.
    \end{equation}

    \textbf{Step 3:}
    Finally, we upper bound \(|\empricalreturn_{\genericinteraction}(\leaderdistribution) - \empricalreturn_{\genericinteraction}(z)|\). Note that 
    \begin{subequations} \label{eqn:empricalreturndiff}
            \begin{align}
        &\empricalreturn_{\genericinteraction}(\leaderdistribution) - \empricalreturn_{\genericinteraction}(z) \nonumber
        \\
        & =\Pr(\randomversion{\leadergenericaction}_{k} = 2|\leaderdistribution) \Pr(\randomversion{\followergenericaction}_{k} = 2|\leaderdistribution) -\Pr(\randomversion{\leadergenericaction}_{k} = 2|z) \Pr(\randomversion{\followergenericaction}_{k} = 2|z) \nonumber
        \\
        &=  \left(\frac{1}{2} - \epsilon \right)\Pr(\randomversion{\followergenericaction}_{k} = 2|\leaderdistribution)- \left(\frac{1}{2} + \epsilon\right)\Pr(\randomversion{\followergenericaction}_{k} = 2|z) \nonumber
        \\
        & \leq \frac{1}{2}(\Pr(\randomversion{\followergenericaction}_{k} = 1|\leaderdistribution) - \Pr(\randomversion{\followergenericaction}_{k} = 1|z)).
    \end{align}
    \end{subequations}

    Let \(\mathcal{D}_{\hat{\randomversion{z}}_{k}}\) and \(\mathcal{D}_{\hat{\randomversion{\leaderdistribution}}_{k}}\) be the distributions of \(\hat{\randomversion{z}}_{k}\) and \(\hat{\randomversion{\leaderdistribution}}_{k}\), respectively. Also, let \(KL(\mathcal{D}^{1} || \mathcal{D}^{2})\) denote the KL divergence between distributions \(\mathcal{D}^{1}\) and \(\mathcal{D}^{2}\), and \(Be(p)\) denote the Bernoulli random variable with parameter \(p\). We have 
\begin{align*}
    KL (\mathcal{D}_{\hat{\randomversion{z}}_{k}} ||\mathcal{D}_{\hat{\randomversion{\leaderdistribution}}_{k}} ) = (k-1) KL\left(Be\left(\frac{1}{2} - \epsilon\right) || Be\left(\frac{1}{2} + \epsilon\right) \right) 
\end{align*} 
since the leader's actions are withdrawn from \(z\) ( and from \(\leaderdistribution\)) at every interaction independently. Using the bound given in Theorem 1 of \cite{dragomir2000some}, we get
\(
    KL(\mathcal{D}_{\hat{\randomversion{z}}_{k}} ||\mathcal{D}_{\hat{\randomversion{\leaderdistribution}}_{k}} ) \leq  \frac{16 \epsilon^2 (k-1)}{1-4\epsilon^2}.
\)

Since \(\randomversion{\followergenericaction}_{k}\) is a function of the emprical distribution \(\hat{\randomversion{z}}_{k}\) (\(\hat{\randomversion{\leaderdistribution}}_{k}\)), using the data processing inequality~\cite{cover1999elements}, we get 
\[KL(Be(\Pr(\randomversion{\followergenericaction}_{k} = 1|z))||Be(\Pr(\randomversion{\followergenericaction}_{k} = 1|\leaderdistribution)) ) \leq KL(\mathcal{D}_{\hat{\randomversion{z}}_{k}} ||\mathcal{D}_{\hat{\randomversion{\leaderdistribution}}_{k}} )\] By  Pinsker's inequality and the above inequalities, 
\[|\Pr(\randomversion{\followergenericaction}_{k} = 1|z) - \Pr(\randomversion{\followergenericaction}_{k} = 1|\leaderdistribution)| \leq \sqrt{\frac{8 \epsilon^2 (k-1)}{1-4\epsilon^2}}.\] Combining this with \eqref{eqn:empricalreturndiff}, we get 
\begin{equation} \label{eqn:empricalreturndifffinal}
    \empricalreturn_{\genericinteraction}(\leaderdistribution) - \empricalreturn_{\genericinteraction}(z) \leq \sqrt{\frac{2 \epsilon^2 (k-1)}{1-4\epsilon^2}}.
\end{equation}

\textbf{Step 4:} Combining \eqref{eqn:stackelbergreturnx}, \eqref{eqn:empricalreturnz}, and \eqref{eqn:empricalreturndifffinal} yields \[\stackelbergreturn(x) - \empricalreturn_{\genericinteraction}(x) \geq \frac{1}{4} - \frac{3\epsilon}{2} - \sqrt{\frac{2 \epsilon^2 (k-1)}{1-4\epsilon^2}}.\]
The bound is a monotone function of \(k\). Setting the right-hand side to \(\epsilon\) and solving for \(k\) yields the desired result.
\end{proof}

For small enough \(\epsilon\), the term \(\nicefrac{1}{(32\epsilon^2)}\) dominates the other terms. If there are \(o(\nicefrac{1}{\epsilon^2})\) interactions, then the leader's expected return under inference is at least \(\epsilon\) worse than its return under full information. 

The strategies with near-optimal Stackelberg returns, i.e.,  \(\stackelbergreturn(x) \geq (\nicefrac{1}{2}) - \epsilon\), will have poor returns under inference since they are close to the decision boundary where the follower abruptly changes its strategy. The returns for the strategies close to the boundary will be poor since the empirical distribution may be on the other side of the decision boundary.

\textcolor{black}{
The alternative bound from in Section \ref{sec:resdiscretestaticfully} implies that for an inferability gap of $\epsilon$, \(\frac{m \log 2 + \log(1/\epsilon)}{d(x)^2 \varphi(x)}\) interactions are sufficient. Using $\varphi(x) \geq 2$, we get that if $k \geq \frac{m \log 2 + \log(1/\epsilon)}{2d(x)^2 } + 1$ then $SR(x) - IR_{k}(x) \leq \epsilon$. For the example used in Proposition $1$, we have $d(x) = \epsilon$, which implies if $k \geq \frac{m \log 2 + \log(1/\epsilon)}{2\epsilon^2 } + 1 = \Tilde{\mathcal{O}}(1/\epsilon^2)$ then $SR(x) - IR_{k}(x) \leq \epsilon$. On the other hand, Proposition $1$ shows that if \(k \leq \frac{1-20\epsilon + 128\epsilon^2+80\epsilon^3 - 400 \epsilon^4}{32\epsilon^2} = \Tilde{\mathcal{O}}(1/\epsilon^2),\) then \(SR(x) - IR_{k}(x) \geq \epsilon.\) Therefore, the bound derived using the concentrability metric $\varphi(x)$ is optimal up to logarithmic factors.
}

%% file: limitedgap.tex
\subsection{Static Bimatrix Games with Limited Inferability Gaps
} \label{sec:importanceofinf}

The inferability of mixed strategies is particularly important in general-sum games where the objectives of the players are weakly positively correlated. Consider the static bimatrix game setting. Due to the positive correlation between the utility matrices $\leaderutilitymatrix$ and $\followerutilitymatrix$, it is useful for the leader to be correctly inferred by the follower. On the other hand, since there is only a weak correlation between $\leaderutilitymatrix$ and $\followerutilitymatrix$, i.e., \(\leaderutilitymatrix \neq \followerutilitymatrix\), the leader's optimal strategy may still be mixed.

Proposition \ref{prop:spectrum} shows that for bimatrix Stackelberg games,  there exist strategies with a limited inferability gap (regardless of the interaction number) if the game is almost cooperative or competitive. In detail, we have
\[\left( \max_{\leaderdistribution} \empricalreturn_{\genericinteraction}(\leaderdistribution) \right) = \left(\max_{\leaderdistribution} \ \leaderdistribution^{\top} \leaderutilitymatrix \randomversion{\followerdistribution}_{\genericinteraction}^{*}
\quad \text{s.t. } \randomversion{\followerdistribution}_{\genericinteraction}^{*} =\arg\max_{\followerdistribution \in \simplex^{\followernumofactions}} \emprical{\randomversion{\leaderdistribution}}^{\top}_{\genericinteraction}\followerutilitymatrix \followerdistribution\right),\] and \[\left(\max_{\leaderdistribution} \stackelbergreturn(\leaderdistribution)\right) = \left(\max_{\leaderdistribution} \ \leaderdistribution^{\top} \leaderutilitymatrix \followerdistribution^{*}
\quad \text{s.t. } \followerdistribution^{*} =\arg\max_{\followerdistribution \in \simplex^{\followernumofactions}} \leaderdistribution^{\top}\followerutilitymatrix \followerdistribution\right).\] For an almost cooperative or competitive game, there exists a strategy for the leader such that the expected return of this strategy in the inference setting, is approximately at the level of the expected return \(\max_{\leaderdistribution} \stackelbergreturn(\leaderdistribution)\) of optimal strategies in the full information setting. Consequently, $\max_{\leaderdistribution} \empricalreturn_{\genericinteraction}(\leaderdistribution)$ is greater than \(\max_{\leaderdistribution} \stackelbergreturn(\leaderdistribution)\) minus a constant depending on the cooperativeness or competitiveness of the game. To formally define the cooperativeness and competitiveness level of the game, we decompose \(\leaderutilitymatrix\) and \(\followerutilitymatrix\).
Let \(U^{c} = \leaderutilitymatrix/2 + \followerutilitymatrix/2 \) (the cooperative objective) and \(U^{z} =  \leaderutilitymatrix/2 - \followerutilitymatrix/2\) (the zero-sum objective). The leader's utility matrix is \(U^{c} + U^{z}\), and the follower's utility matrix is \(U^{c} - U^{z}\). Let \(\bar{U}^{c} =  (U^{c} - c^{c})/\alpha^{c}\) and \(\bar{U}^{z} =  (U^{z}- c^{z})/\alpha^{z}\) be the shifted and normalized versions of \(U^{c}\) and \(U^{z}\), respectively, such that \( \max_{i, j} \bar{U}^{c}_{i,j} = 1\), \(\min_{i, j} \bar{U}^{c}_{i,j} = 0 \), \( \max_{i, j} \bar{U}^{z}_{i,j} = 1\), and \(\min_{i, j} \bar{U}^{z}_{i,j} = 0 \). The game is zero-sum if \(\alpha^{c} = 0\) and fully cooperative if \(\alpha^{z} = 0\).

\begin{proposition} \label{prop:spectrum}
    For a repeated static bimatrix Stackelberg game and $k \geq 2$:
    \begin{enumerate}  
    \item 
    \[\left(  \max_{\leaderdistribution} \empricalreturn_{\genericinteraction}(\leaderdistribution) \right) \geq \left( \max_{\leaderdistribution} \stackelbergreturn(\leaderdistribution) \right) - 2\alpha^{c}.\] \label{case:competitive}

            \item If the maximum element of \(\followerutilitymatrix_{ij}\) is unique and the follower's estimate of the leader's strategy is unbiased,  
    \[\left(  \max_{\leaderdistribution} \empricalreturn_{\genericinteraction}(\leaderdistribution) \right) \geq \left(  \max_{\leaderdistribution} \stackelbergreturn(\leaderdistribution) \right) - 2\alpha^{z}.\]  \label{case:cooperative}
    \end{enumerate}
\end{proposition}
\begin{proof}[Proof of Proposition \ref{prop:spectrum}] We prove the statements separately.

Statement \ref{case:competitive}:
Consider that the leader optimizes for the negative of the follower's objective, i.e., \(-B = -U^{c} + U^{z}\). In this case, the leader's optimal strategy \(x^{*}\) is a Nash equilibrium strategy for the zero-sum objective function  \(-U^{c} + U^{z}\).  \(\leaderdistribution^{*}\) is also a Nash equilibrium strategy for the zero-sum objective function  \(-U^{c} + U^{z} + 2c^{c} J\) since \(2c^{c} J\) is constant. Define \[\followerdistribution^{*} =\arg\max_{\followerdistribution \in \simplex^{\followernumofactions}} \leaderdistribution^{\top}\followerutilitymatrix\followerdistribution =\arg\max_{\followerdistribution \in \simplex^{\followernumofactions}} \leaderdistribution^{\top}(U^{c} - U^{z})\followerdistribution, \]
\[\followerdistribution^{**} =\arg\max_{\followerdistribution \in \simplex^{\followernumofactions}} (\leaderdistribution^{*})^{\top}\followerutilitymatrix\followerdistribution =\arg\max_{\followerdistribution \in \simplex^{\followernumofactions}} (\leaderdistribution^{*})^{\top}(U^{c} - U^{z})\followerdistribution ,\] and \[\randomversion{\followerdistribution}^{*}_{\genericinteraction}=\arg\max_{\followerdistribution\in \simplex^{\followernumofactions}} \randomversion{\emprical{\leaderdistribution}}_{\genericinteraction}^{\top}\followerutilitymatrix\followerdistribution =\arg\max_{\followerdistribution\in \simplex^{\followernumofactions}} \randomversion{\emprical{\leaderdistribution}}_{\genericinteraction}^{\top}(U^{c} - U^{z})\followerdistribution.\] Since \(\followerdistribution^{**}\) is a Nash equilibrium strategy, for all \(k\geq 2\),  
\begin{align*}
    &\expectation{ (\leaderdistribution^{*})^{\top} (-U^{c}+c^{c} J +U^{z}+c^{c} J)\followerdistribution^{**}}{}
    \\
&\leq  \expectation{  (\leaderdistribution^{*})^{\top} (-U^{c}+c^{c} J +U^{z}+c^{c} J)\randomversion{\followerdistribution}^{*}_{\genericinteraction}}{}
\end{align*}
Since all entries of \(U^{c}- c^{c} J\) are positive, we have 
\begin{align*}
      &\expectation{  (\leaderdistribution^{*})^{\top} (-U^{c} +c^{c} J+U^{z} +c^{c}J)\randomversion{\followerdistribution}^{*}_{\genericinteraction}}{}
      \\
      &\leq   \expectation{  (\leaderdistribution^{*})^{\top} (U^{c} -c^{c} J+U^{z} + c^{c} J)\randomversion{\followerdistribution}^{*}_{\genericinteraction}}{}
    \\
    & = \expectation{  (\leaderdistribution^{*})^{\top} (U^{c} +U^{z})\randomversion{\followerdistribution}^{*}_{\genericinteraction}}{}
\end{align*} which is the expected return of the strategy \(\leaderdistribution^{*}\) under inference at interaction \(\genericinteraction\). Since \(\leaderdistribution^{*}\) is a feasible strategy, we have 
\begin{align*}
    \left(  \max_{\leaderdistribution} \empricalreturn_{\genericinteraction}(\leaderdistribution) \right) &\geq  \expectation{ (\leaderdistribution^{*})^{\top} (-U^{c} +U^{z} + 2 c^{c} J)\followerdistribution^{**}}{}
    \\
    &= \expectation{ (\leaderdistribution^{*})^{\top} (-U^{c} +U^{z})\followerdistribution^{**}}{} + 2 c^{c} 
\end{align*}

For all \(k\geq 2\), we have
\begin{align*}
    \max_{\leaderdistribution} \expectation{ \leaderdistribution^{\top} \leaderutilitymatrix \followerdistribution^{*}}{}
&=     \max_{\leaderdistribution} \expectation{ \leaderdistribution^{\top} (-U^{c} +2U^{c}+U^{z})\followerdistribution^{*}}{}
\\
&\leq     \max_{\leaderdistribution} \expectation{ \leaderdistribution^{\top} (-U^{c} +U^{z})\followerdistribution^{*}+ 2\alpha^{c} + 2c^{c}}{}
\\
&=   \expectation{ (\leaderdistribution^{*})^{\top} (-U^{c} +U^{z})\followerdistribution^{**}+ 2\alpha^{c}+ 2c^{c}}{}\,.
\end{align*}
The inequality is because for every \(\leaderdistribution\) there exists a unique \(\followerdistribution^{*}\), and \(\max_{\leaderdistribution, \followerdistribution^{*}} 2 \leaderdistribution^{\top} U^{c} \followerdistribution^{*} 
 = 2\alpha^{c} + 2c^{c}\).
 Hence, we have 
 \[\left( \max_{\leaderdistribution} \stackelbergreturn(\leaderdistribution) \right) \leq \expectation{ (\leaderdistribution^{*})^{\top} (-U^{c} +U^{z})\followerdistribution^{**} + 2\alpha^{c} + 2c^{c}}{}.\]

Combining the bounds on \( \max_{\leaderdistribution} \stackelbergreturn(\leaderdistribution) \) and \( \max_{\leaderdistribution} \empricalreturn_{\genericinteraction}(\leaderdistribution) \) yields the desired result.

Statement \ref{case:cooperative}: Note that if \(\followerutilitymatrix\) has a unique maximum element then \(U^{c} - U^{z} + 2c^{z}J\) and \(\alpha^{c} \bar{U}^{c} - \alpha^{z} \bar{U}^{z} + (c^{c} + c^{z})J\) also have unique maximum elements at the same locations since they are shifted versions of \(\followerutilitymatrix\) by \(2c^{z}J\). Let \((i,j)\) be the index of the maximum element. 

Consider that the leader optimizes for the follower's objective function (with a constant offset of \(2c^{z}\)), i.e., \(U^{c} - U^{z} + 2c^{z}J\), and its strategy \(\leaderdistribution^{*}\) plays action \(i\) deterministically. In this case, the leader's return is \(\alpha^{c} \bar{U}^{c}_{ij} - \alpha^{z} \bar{U}^{z}_{ij} + c^{c} + c^{z}\) in the full information setting since the follower will play action \(j\). Similarly in the inference setting, the leader's return is \(\alpha^{c} \bar{U}^{c}_{ij} - \alpha^{z} \bar{U}^{z}_{ij}+c^{c} + c^{z}\) in every interaction after the first interaction since the follower will infer the leader's strategy with no error and play action \(j\). Note that \(\alpha^{c} \bar{U}^{c}_{ij} - \alpha^{z} \bar{U}^{z}_{ij} + c^{c} + c^{z} \geq \alpha^{c} - \alpha^{z} + c^{c} + c^{z}\) since \[\max_{k,l} \alpha^{c} \bar{U}^{c}_{kl} -\alpha^{z} \bar{U}^{z}_{kl} \geq \max_{k,l} \alpha^{c} \bar{U}^{c}_{kl} -  \max_{k,l} 
 \alpha^{z} \bar{U}^{z}_{kl} \geq \alpha^{c} - \alpha^{z}.\] Consequently, we have \(\expectation{ (\leaderdistribution^{*})^{\top} \leaderutilitymatrix \randomversion{\followerdistribution}^{*}_{\genericinteraction}}{} \geq (\alpha^{c} - \alpha^{z}+c^{c} + c^{z})\) which implies that \(\left(  \max_{\leaderdistribution} \empricalreturn_{\genericinteraction}(\leaderdistribution) \right) \geq (\alpha^{c} - \alpha^{z}+c^{c} + c^{z})\).

We have \(\leaderutilitymatrix = \alpha^{c} \bar{U}^{c} +\alpha^{z} \bar{U}^{z} +(c^{c} + c^{z})J\) and \(\max_{k,l} \alpha^{c} \bar{U}^{c}_{kl} +\alpha^{z} \bar{U}^{z}_{kl} \leq \max_{k,l} \alpha^{c} \bar{U}^{c}_{kl} +\max_{k,l}  \alpha^{z} \bar{U}^{z}_{kl} = \alpha^{c} + \alpha^{z} \) which imply that \(\left(  \max_{\leaderdistribution} \stackelbergreturn(\leaderdistribution) \right) \leq  \alpha^{c} +\alpha^{z} +c^{c} + c^{z}\). 

Arranging the inequalities for \(\left(  \max_{\leaderdistribution} \empricalreturn_{\genericinteraction}(\leaderdistribution) \right)\) and \(\left(  \max_{\leaderdistribution} \empricalreturn_{\genericinteraction}(\leaderdistribution) \right)\) yields the results for statement \ref{case:cooperative}.  
\end{proof}

If \(\alpha^{z}\) is small (i.e., the competitive aspect of the game is insignificant,) the leader can optimize for the follower's objective function instead of its own objective function. By considering a joint objective function, we can observe that in the full information setting, there exists a deterministic strategy for the leader that is near-optimal for the leader's own objective (assuming that the follower has a unique optimal response to the leader's strategy.) Since the strategy is deterministic, the leader does not suffer from an inferability gap in the inference setting compared to the full information setting. Since the ignored part of the objective function is insignificant, the expected return of this strategy in the inference setting is near the optimal return in the full information setting. If \(\alpha^{c}\) is small, (i.e., the cooperative aspect of the game is insignificant,) the leader can optimize for the opposite of the follower's objective function instead of its own objective function. By considering the zero-sum equilibrium, we can observe that in the full information setting, there exists a mixed strategy for the leader that is near-optimal in the absolute sense for the leader's own objective. In the inference setting, the follower may not be able to infer the leader's strategy fully. However, since any error in inferring the leader's strategy leads to a different strategy for the follower, it can only decrease the follower's returns as the leader's strategy is a zero-sum equilibrium strategy. Consequently, it can only increase the leader's returns, and the leader does not suffer from an inferability gap compared to the full information setting. Since the ignored part of the objective function is insignificant, the expected return of this strategy in the inference setting is near the optimal return in the full information setting.

%% file: numerical_examples.tex
\section{Numerical Examples}
In this section, we evaluate the effect of inference on repeated bimatrix Stackelberg games and a repeated Stackelberg game with parametric action spaces. For the bimatrix games, we consider the aforementioned car-pedestrian interaction and randomly generated bimatrix games. For clarity of presentation, we plot the average return \(\frac{1}{K}\sum_{k=2}^{K} \empricalreturn_{\genericinteraction}(\leaderdistribution)\), which is the expected cumulative return up to interaction \(\numofinteractions\) divided by \(\numofinteractions\). We approximate the expectation with repeated simulations. 

\subsection{Car-Pedestrian Interactions}

\begin{figure}
 \vspace*{10pt}
    \centering
    \input{plots/car_pedestrian_lambda_5}
    \input{plots/car_pedestrian}  
    \caption{The car's average return (averaged over $10^4$ simulations) in the pedestrian-car example. 
    Solid lines represent the average return for different strategies where $p$ is the probability of the car stopping.  
    Dashed lines represent the average return per interaction under full information, i.e., $\leaderdistribution^{\top} \leaderutilitymatrix\sigma_{\lambda}(B^{\top}x) $ for $x = [p, 1-p ].$ (Top) $\lambda = 5$. (Bottom) $\lambda =100$. }
    \label{fig:CP_matrix}
\end{figure}
We consider the bimatrix game presented in Table \ref{tab:bimarix} with a boundedly rational follower with maximum entropy response. We simulate the gameplay under inference for $100$ interactions with rationality constants $\rationalityconstant = 5$ and $\rationalityconstant = 100$. The car's strategy is determined by $p$, i.e., the probability that the car stops. For \(\rationalityconstant=5\) and \(\rationalityconstant=100\), the optimal \(p\) are \(0.77\) and \(0.53\) in the full information setting, respectively. 
We show the results in Fig. \ref{fig:CP_matrix} for different values of \(p\).

For $\lambda=100$, all strategies receive higher average returns as the number of interactions increases as the pedestrian's estimation improves. In the long run, $p = 0.53$, the optimal strategy for the car in the full information setting, would achieve the highest return.
However, after 100 interactions, this strategy is still underperforming compared to more deterministic strategies.
This is because a small error in the pedestrian's estimation \(\emprical{\randomversion{p}}_{\genericinteraction}\) results in large changes in the pedestrian's strategy, demonstrating the impact inference has on the leader's return.
On the other hand, as we expected, the strategies with higher stopping probabilities achieve higher transient returns:   
More deterministic strategies are easier for the pedestrian to infer at a small number of interactions $K$, and any error in \(\emprical{\randomversion{p}}_{\genericinteraction}\) results in only small changes to the pedestrian's strategy.
\textcolor{black}{
For $\lambda=5$, we also observe that more deterministic strategies such as $p=0.8$ or $p=0.9$ achieve higher returns in the transient period than $p=0.9$, which is optimal for the full information case. In this case, the pedestrian is less rational, caring less about the leader's strategy, and takes actions more uniformly randomly. Consequently, the inferability gaps are smaller, aligned with the bound given in Corollary \ref{thm:achieveability}. }

\subsection{General-sum Tug of War}
We consider a static Stackelberg game with parametric mixed strategies to demonstrate the importance of inferability. The leader's strategies are normal distributions \(\mathcal{N}(\mu^{\leader}, (s^{\leader})^2)\) parametrized by \((\mu^{\leader}, s^{\leader}) \in \mathbb{R}^2\). Similarly, the follower's are normal distributions \(\mathcal{N}(\mu^{\follower}, (s^{\follower})^2)\) parametrized by \((\mu^{\follower}, s^{\follower}) \in \mathbb{R}^2\). The leader's action is \(\leadergenericcontaction \sim \mathcal{N}(\mu^{\leader}, (s^{\leader})^2)\), and the follower's action is \(\followergenericcontaction \sim \mathcal{N}(\mu^{\follower}, (s^{\follower})^2)\). Let $0\leq c_{low} < c_{high} $ be constants. The return of the leader is \(1\) if \(\leadergenericcontaction + \followergenericcontaction \in [c_{low}, c_{high}]\), and \(0\) otherwise. Symmetrically around \(0\), the return of the follower is \(1\) if \(\leadergenericcontaction + \followergenericcontaction \in [-c_{high}, -c_{low}]\), and \(0\) otherwise. 

We note that the leader and follower try to pull the sum of their actions in different directions. This aspect of the game resembles a competitive ``tug of war''. On the other hand, for inconclusive outcomes, i.e., \(\leadergenericcontaction + \followergenericcontaction \in (-c_{low}, c_{low})\), or extreme outcomes, i.e., \(\leadergenericcontaction + \followergenericcontaction \in (-\infty, -c_{high}) \cup (c_{high}, \infty)\), both parties receive the same low return of \(0\). This aspect of the game is cooperative in that both parties aim to avoid these regions.

\begin{figure}[t]
    \centering
\input{plots/tug_of_war}
    \caption{The leader's expected return in the tug of war game for different mean and standard deviation values ($c_{low}=0.01$ and $c_{high}=5$). Lower values of standard deviation lead to higher returns for the lower number of interactions. \textcolor{black}{The leader's mean parameter has no effect on the expected return.}}
    \label{fig:tugofwar}
\end{figure}
Due to the independent sampling of \(\leadergenericcontaction\) and \(\followergenericcontaction\), we have \(\leadergenericcontaction + \followergenericcontaction \sim \mathcal{N}(\mu^{\leader}+\mu^{\follower}, (s^{\leader})^2 + (s^{\follower})^2)\). 
\textcolor{black}{One can observe that given the leader's strategy \(\mathcal{N}(\mu^{\leader}, (s^{\leader})^2)\), the follower's optimal strategy is \(\mathcal{N}((-c_{high}-c_{low})/2 - \mu^{\leader}, 0)\) since it maximizes the probability of \(\leadergenericcontaction + \followergenericcontaction \in [-c_{high}, -c_{low}]\) which is the follower's expected return. Given the follower's optimal strategy, we have \(\leadergenericcontaction + \followergenericcontaction \sim \mathcal{N}((-c_{high}-c_{low})/2, (s^{\leader})^2)\). The leader's expected return is the probability that \(\leadergenericcontaction + \followergenericcontaction \in [c_{low}, c_{high}]\) and the leader's mean parameter \(\mu^{\leader}\) has no effect on the leader's expected return. } As shown in Figure \ref{fig:tugofwar}, \(\stackelbergreturn(\mu^{\leader}, s^{\leader}) = \Pr(\leadergenericcontaction + \followergenericcontaction \in [c_{low}, c_{high}])\) is a unimodal function of \(s^{\leader}\) between \([0, \infty)\) and attains its maximum for a finite positive value of \(s^{\leader}\). Consequently, the leader's optimal strategy is mixed.

Consider a scenario where the leader and the follower will interact a certain number of times. The follower estimates the leader's mean from the previous interactions using the plug-in estimator. \textcolor{black}{After $k$ interactions, the follower's estimate of the leader's mean parameter is $\hat{\randomversion{\mu}}^{l}_{k}$ that is distributed as \(\mathcal{N}(\mu^{\leader}, (s^{\leader})^2)/k\). The follower's action $b_{k+1}$ is distributed as $\mathcal{N}(  (-c_{high}-c_{low})/2 - \hat{\randomversion{\mu}}^{l}_{k}, 0) = \mathcal{N}((-c_{high}-c_{low})/2 - \mu^{\leader}, (s^{\leader})^2)/k)$. Since the leader's action $\randomversion{a}_{k+1}$ is distributed as \(\mathcal{N}(\mu^{\leader}, (s^{\leader})^2)\), and $\randomversion{a}_{k+1}$ and $b_{k+1}$ are independently sampled, we have $a_{k+1} + b_{k+1} \sim \mathcal{N} ((-c_{high}-c_{low})/2, (s^{\leader})^2) + (s^{\leader})^2)/k)$, and the leader's return expected return at $(k+1)$-th interaction is the probability that $a_{k+1} + b_{k+1} \in [c_{low}, c_{high}]$.  We note that the leader's mean parameter \(\mu^{\leader}\) has no effect on the expected return since the follower's unbiased estimate cancels any changes to $\mu^{\leader}$. On the other hand, the expected return depends on the leader's variance parameter.} If the leader's strategy has high variance, i.e., it is not easily inferable, then the follower's estimation will be inaccurate. In return, the leader will suffer from an inferability gap. Figure \ref{fig:tugofwar} shows that strategies with lower variance achieve a higher return for lower number of interactions. 

\color{black}

\begin{figure}
 \vspace*{10pt}
    \centering
    \input{plots/optimal_C_random_matrix}
    \caption{The leader's average return for the randomly generated bimatrix games. 
    Solid lines represent the average return for the bound's local maxima for different values of the regularization constant $c$.  
    Dashed lines represent the average return per interaction under full information, i.e., $(x^*(c )))^{\top} \leaderutilitymatrix\sigma_{\lambda}(B^{\top}x^*(c )) $.}
    \label{fig:random_matrix}
\end{figure}

\subsection{Randomly Generated Bimatrix Games} \label{sec:experimentrandomgames}
We evaluate the performance under inference for randomly generated bimatrix games when the follower is boundedly rational with maximum entropy response.
From the achievability bound given in Corollary  \ref{thm:achieveability}, 
\begin{equation*}
    (K-1)\stackelbergreturn(\leaderdistribution) - c \, \nu(\leaderdistribution)  \leq   \sum_{k=2}^{K}\empricalreturn_{\genericinteraction}(\leaderdistribution)  
\end{equation*}
for some constant $c$ depending on \(\numofinteractions\).
We use \(\nu\) as a regularizer and optimize the bound for fixed values of $c$: 
\begin{equation*}
    \leaderdistribution^{*}(c) = \arg\max_{\leaderdistribution \in \simplex^{\leadernumofactions}} \stackelbergreturn(\leaderdistribution) - c \text{ } (\nu(\leaderdistribution)) ^{2}.
\end{equation*} 
and compare the performance of leader strategies for different values of $c$. We replace $\nu$ with $\nu^{2}$ in the optimization problem since the gradients of \(\nu^2\) are Lipschitz continuous. 
\textcolor{black}{
We note that even when $c = 0$, this is a nonconvex optimization problem. On the other hand, the objective is Lipschitz continuous thanks to the softmax response. To find a local optimum, we use gradient descent with decaying stepsize. 
We use the leader's optimal strategy from the Stackelberg game with a fully rational follower as the starting point for the gradient descent. We note that for the bimatrix game under consideration, by exploiting that the pure strategies are optimal for the follower, such a strategy can be obtained by solving $\followernumofactions$ linear programs~\cite{conitzer2006computing}. Each of the linear programs fixes the follower's action and constrains the leader's strategies to a polytope to ensure that the considered follower action is optimal.}

In this example, we randomly generate bimatrix games.  For each bimatrix game, the entries of the leader's  utility matrix \(  \leaderutilitymatrix\) are uniformly randomly distributed between  \(0 \) and \(1.\)
The follower's utility matrix \( \followerutilitymatrix  = \nicefrac{\leaderutilitymatrix}{2}  + \nicefrac{ C}{2}\), where $C$ is a uniformly randomly distributed matrix between  \(0 \) and \(1.\) This construction makes \(\leaderutilitymatrix\) and \(\followerutilitymatrix\) weakly positively correlated highlighting the importance of mixed strategies and inferability as explained in Section \ref{sec:importanceofinf}.

We randomly generate \(10,000 \) $ 4 \times 4$ bimatrix games. For each random bimatrix game, we find the leader's strategy \(\leaderdistribution^{*}(c) \) for \(c = 0, 1, 10,\) and \(100 \).
For each bimatrix game, we simulate play for $100$ interactions with rationality constant $\rationalityconstant = 100.$
We repeat the simulations 100 times, and the leader's return is averaged at each interaction over these simulations. 
Then, the leader's average return until interaction \(k\) for each bimatrix is averaged at each interaction \(k\) over all bimatrix games. Results are shown in Fig. \ref{fig:random_matrix}.

In these simulations, higher regularization constants correspond to more inferable (less stochastic) strategies, as more weight is given to the stochasticity level of a strategy. The optimal strategies for the full information setting \(\leaderdistribution^{*}(c=0) \) (the optimal strategy with no stochasticity regularization) achieves a higher average expected return in the long run (after $45$ interactions) since the follower's estimation accuracy improves with more interactions. However, these strategies still suffer inferability gap after $100$ interactions. 
For the first $ 45$ interactions, the regularization constant $c= 100$ yields higher average returns, and the average return reaches its final value even after the first interaction since the generated strategies are deterministic and estimated by the follower perfectly.

%% file: plots/car_pedestrian_lambda_5.tex
\begin{tikzpicture}
\definecolor{darkgray176}{RGB}{176,176,176}
\definecolor{lightgray204}{RGB}{204,204,204}
\begin{axis}[
legend columns=3,
legend cell align={left},
legend style={
  fill opacity=0.8,
  draw opacity=1,
  text opacity=1,
  at={(0.97,0.03)},
  anchor=south east,
  draw=lightgray204
},
width=1\columnwidth,
height=0.45\columnwidth,
tick align=outside,
tick pos=left,
x grid style={darkgray176},
xlabel={\small Interactions (excluding the first interaction) },
xmin=-.9, xmax=100,
xtick style={color=black},
y grid style={darkgray176},
ylabel={\small Average Car Return},
ymin=-1, ymax=.55,
ytick style={color=black},
xmajorgrids,
ymajorgrids
]
\addplot [ultra thick, matlabred]
table {%
1 -0.7998
2 -0.7718
3 -0.7012
4 -0.7078
5 -0.6468
6 -0.6154
7 -0.5724
8 -0.5852
9 -0.5574
10 -0.5262
11 -0.5334
12 -0.5068
13 -0.5426
14 -0.4606
15 -0.4746
16 -0.4412
17 -0.4412
18 -0.4272
19 -0.3728
20 -0.4388
21 -0.3752
22 -0.4138
23 -0.4016
24 -0.3902
25 -0.4486
26 -0.3868
27 -0.4166
28 -0.4394
29 -0.3894
30 -0.3902
31 -0.3594
32 -0.36
33 -0.4114
34 -0.3298
35 -0.4076
36 -0.3562
37 -0.3574
38 -0.3932
39 -0.355
40 -0.3666
41 -0.3656
42 -0.4178
43 -0.4106
44 -0.3726
45 -0.4116
46 -0.4058
47 -0.3534
48 -0.374
49 -0.4192
50 -0.3556
51 -0.3462
52 -0.3464
53 -0.3758
54 -0.3338
55 -0.3764
56 -0.3534
57 -0.335
58 -0.3424
59 -0.3942
60 -0.337
61 -0.3324
62 -0.348
63 -0.3416
64 -0.3606
65 -0.3474
66 -0.3588
67 -0.3454
68 -0.386
69 -0.3052
70 -0.314
71 -0.297
72 -0.3014
73 -0.2984
74 -0.3392
75 -0.2988
76 -0.3288
77 -0.3204
78 -0.3164
79 -0.3496
80 -0.2914
81 -0.3468
82 -0.2614
83 -0.311
84 -0.3686
85 -0.3272
86 -0.3234
87 -0.3226
88 -0.306
89 -0.3416
90 -0.2674
91 -0.3446
92 -0.3504
93 -0.325
94 -0.322
95 -0.3196
96 -0.3316
97 -0.3102
98 -0.356
99 -0.346
};
\addlegendentry{$p = 0.6$}
\addplot [ultra thick, matlabred, dashed, forget plot]
table {%
1 -0.275765685479981
2 -0.275765685479981
3 -0.275765685479981
4 -0.275765685479981
5 -0.275765685479981
6 -0.275765685479981
7 -0.275765685479981
8 -0.275765685479981
9 -0.275765685479981
10 -0.275765685479981
11 -0.275765685479981
12 -0.275765685479981
13 -0.275765685479981
14 -0.275765685479981
15 -0.275765685479981
16 -0.275765685479981
17 -0.275765685479981
18 -0.275765685479981
19 -0.275765685479981
20 -0.275765685479981
21 -0.275765685479981
22 -0.275765685479981
23 -0.275765685479981
24 -0.275765685479981
25 -0.275765685479981
26 -0.275765685479981
27 -0.275765685479981
28 -0.275765685479981
29 -0.275765685479981
30 -0.275765685479981
31 -0.275765685479981
32 -0.275765685479981
33 -0.275765685479981
34 -0.275765685479981
35 -0.275765685479981
36 -0.275765685479981
37 -0.275765685479981
38 -0.275765685479981
39 -0.275765685479981
40 -0.275765685479981
41 -0.275765685479981
42 -0.275765685479981
43 -0.275765685479981
44 -0.275765685479981
45 -0.275765685479981
46 -0.275765685479981
47 -0.275765685479981
48 -0.275765685479981
49 -0.275765685479981
50 -0.275765685479981
51 -0.275765685479981
52 -0.275765685479981
53 -0.275765685479981
54 -0.275765685479981
55 -0.275765685479981
56 -0.275765685479981
57 -0.275765685479981
58 -0.275765685479981
59 -0.275765685479981
60 -0.275765685479981
61 -0.275765685479981
62 -0.275765685479981
63 -0.275765685479981
64 -0.275765685479981
65 -0.275765685479981
66 -0.275765685479981
67 -0.275765685479981
68 -0.275765685479981
69 -0.275765685479981
70 -0.275765685479981
71 -0.275765685479981
72 -0.275765685479981
73 -0.275765685479981
74 -0.275765685479981
75 -0.275765685479981
76 -0.275765685479981
77 -0.275765685479981
78 -0.275765685479981
79 -0.275765685479981
80 -0.275765685479981
81 -0.275765685479981
82 -0.275765685479981
83 -0.275765685479981
84 -0.275765685479981
85 -0.275765685479981
86 -0.275765685479981
87 -0.275765685479981
88 -0.275765685479981
89 -0.275765685479981
90 -0.275765685479981
91 -0.275765685479981
92 -0.275765685479981
93 -0.275765685479981
94 -0.275765685479981
95 -0.275765685479981
96 -0.275765685479981
97 -0.275765685479981
98 -0.275765685479981
99 -0.275765685479981
};
\addplot [ultra thick, matlabyellow]
table {%
1 -0.3078
2 -0.3116
3 -0.181
4 -0.161
5 -0.0974
6 -0.0356
7 -0.0356
8 0.0258
9 0.0078
10 0.0464
11 0.0738
12 0.0594
13 0.0942
14 0.1018
15 0.0756
16 0.087
17 0.1202
18 0.131
19 0.1244
20 0.1248
21 0.1442
22 0.1346
23 0.1488
24 0.104
25 0.1208
26 0.1882
27 0.1386
28 0.1436
29 0.179
30 0.1724
31 0.1938
32 0.1162
33 0.188
34 0.1618
35 0.1762
36 0.191
37 0.1806
38 0.173
39 0.184
40 0.1732
41 0.2078
42 0.1604
43 0.219
44 0.169
45 0.1696
46 0.2158
47 0.1976
48 0.1952
49 0.2004
50 0.2162
51 0.227
52 0.1954
53 0.1864
54 0.1984
55 0.211
56 0.2276
57 0.2116
58 0.224
59 0.1802
60 0.2098
61 0.2196
62 0.1948
63 0.1814
64 0.2526
65 0.2242
66 0.1976
67 0.2124
68 0.1954
69 0.2082
70 0.2012
71 0.2206
72 0.2012
73 0.2234
74 0.1952
75 0.2084
76 0.198
77 0.2348
78 0.2122
79 0.2324
80 0.204
81 0.2328
82 0.202
83 0.235
84 0.2184
85 0.2294
86 0.2468
87 0.2268
88 0.2208
89 0.206
90 0.2194
91 0.227
92 0.1996
93 0.2056
94 0.2098
95 0.2334
96 0.2284
97 0.2294
98 0.2052
99 0.2416
};
\addlegendentry{$p = 0.7$}
\addplot [ultra thick, matlabyellow, dashed, forget plot]
table {%
1 0.242391233933647
2 0.242391233933647
3 0.242391233933647
4 0.242391233933647
5 0.242391233933647
6 0.242391233933647
7 0.242391233933647
8 0.242391233933647
9 0.242391233933647
10 0.242391233933647
11 0.242391233933647
12 0.242391233933647
13 0.242391233933647
14 0.242391233933647
15 0.242391233933647
16 0.242391233933647
17 0.242391233933647
18 0.242391233933647
19 0.242391233933647
20 0.242391233933647
21 0.242391233933647
22 0.242391233933647
23 0.242391233933647
24 0.242391233933647
25 0.242391233933647
26 0.242391233933647
27 0.242391233933647
28 0.242391233933647
29 0.242391233933647
30 0.242391233933647
31 0.242391233933647
32 0.242391233933647
33 0.242391233933647
34 0.242391233933647
35 0.242391233933647
36 0.242391233933647
37 0.242391233933647
38 0.242391233933647
39 0.242391233933647
40 0.242391233933647
41 0.242391233933647
42 0.242391233933647
43 0.242391233933647
44 0.242391233933647
45 0.242391233933647
46 0.242391233933647
47 0.242391233933647
48 0.242391233933647
49 0.242391233933647
50 0.242391233933647
51 0.242391233933647
52 0.242391233933647
53 0.242391233933647
54 0.242391233933647
55 0.242391233933647
56 0.242391233933647
57 0.242391233933647
58 0.242391233933647
59 0.242391233933647
60 0.242391233933647
61 0.242391233933647
62 0.242391233933647
63 0.242391233933647
64 0.242391233933647
65 0.242391233933647
66 0.242391233933647
67 0.242391233933647
68 0.242391233933647
69 0.242391233933647
70 0.242391233933647
71 0.242391233933647
72 0.242391233933647
73 0.242391233933647
74 0.242391233933647
75 0.242391233933647
76 0.242391233933647
77 0.242391233933647
78 0.242391233933647
79 0.242391233933647
80 0.242391233933647
81 0.242391233933647
82 0.242391233933647
83 0.242391233933647
84 0.242391233933647
85 0.242391233933647
86 0.242391233933647
87 0.242391233933647
88 0.242391233933647
89 0.242391233933647
90 0.242391233933647
91 0.242391233933647
92 0.242391233933647
93 0.242391233933647
94 0.242391233933647
95 0.242391233933647
96 0.242391233933647
97 0.242391233933647
98 0.242391233933647
99 0.242391233933647
};
\addplot [ultra thick, matlabgreen]
table {%
1 -0.1058
2 -0.0814
3 0.0294
4 0.1008
5 0.1084
6 0.1518
7 0.1678
8 0.17
9 0.1962
10 0.199
11 0.2128
12 0.2132
13 0.2264
14 0.212
15 0.2246
16 0.2526
17 0.2586
18 0.2344
19 0.2624
20 0.2712
21 0.2804
22 0.2362
23 0.2608
24 0.2742
25 0.2546
26 0.2804
27 0.2916
28 0.2822
29 0.27
30 0.2754
31 0.281
32 0.2752
33 0.2836
34 0.2748
35 0.307
36 0.2916
37 0.2972
38 0.2874
39 0.2832
40 0.2936
41 0.279
42 0.3092
43 0.2538
44 0.2934
45 0.284
46 0.2754
47 0.2988
48 0.2776
49 0.3164
50 0.2906
51 0.3014
52 0.3032
53 0.277
54 0.2866
55 0.3024
56 0.3014
57 0.2968
58 0.2986
59 0.291
60 0.2998
61 0.3198
62 0.2892
63 0.3006
64 0.3198
65 0.2786
66 0.2736
67 0.2926
68 0.308
69 0.3066
70 0.3132
71 0.2968
72 0.2762
73 0.3236
74 0.3152
75 0.3094
76 0.3112
77 0.3082
78 0.3006
79 0.307
80 0.29
81 0.329
82 0.27
83 0.3082
84 0.291
85 0.2696
86 0.2994
87 0.3062
88 0.312
89 0.3284
90 0.2986
91 0.313
92 0.3124
93 0.312
94 0.3188
95 0.2792
96 0.3042
97 0.3076
98 0.2868
99 0.3386
};
\addlegendentry{$p = 0.77$}
\addplot [ultra thick, matlabgreen, dashed, forget plot]
table {%
1 0.315161281068908
2 0.315161281068908
3 0.315161281068908
4 0.315161281068908
5 0.315161281068908
6 0.315161281068908
7 0.315161281068908
8 0.315161281068908
9 0.315161281068908
10 0.315161281068908
11 0.315161281068908
12 0.315161281068908
13 0.315161281068908
14 0.315161281068908
15 0.315161281068908
16 0.315161281068908
17 0.315161281068908
18 0.315161281068908
19 0.315161281068908
20 0.315161281068908
21 0.315161281068908
22 0.315161281068908
23 0.315161281068908
24 0.315161281068908
25 0.315161281068908
26 0.315161281068908
27 0.315161281068908
28 0.315161281068908
29 0.315161281068908
30 0.315161281068908
31 0.315161281068908
32 0.315161281068908
33 0.315161281068908
34 0.315161281068908
35 0.315161281068908
36 0.315161281068908
37 0.315161281068908
38 0.315161281068908
39 0.315161281068908
40 0.315161281068908
41 0.315161281068908
42 0.315161281068908
43 0.315161281068908
44 0.315161281068908
45 0.315161281068908
46 0.315161281068908
47 0.315161281068908
48 0.315161281068908
49 0.315161281068908
50 0.315161281068908
51 0.315161281068908
52 0.315161281068908
53 0.315161281068908
54 0.315161281068908
55 0.315161281068908
56 0.315161281068908
57 0.315161281068908
58 0.315161281068908
59 0.315161281068908
60 0.315161281068908
61 0.315161281068908
62 0.315161281068908
63 0.315161281068908
64 0.315161281068908
65 0.315161281068908
66 0.315161281068908
67 0.315161281068908
68 0.315161281068908
69 0.315161281068908
70 0.315161281068908
71 0.315161281068908
72 0.315161281068908
73 0.315161281068908
74 0.315161281068908
75 0.315161281068908
76 0.315161281068908
77 0.315161281068908
78 0.315161281068908
79 0.315161281068908
80 0.315161281068908
81 0.315161281068908
82 0.315161281068908
83 0.315161281068908
84 0.315161281068908
85 0.315161281068908
86 0.315161281068908
87 0.315161281068908
88 0.315161281068908
89 0.315161281068908
90 0.315161281068908
91 0.315161281068908
92 0.315161281068908
93 0.315161281068908
94 0.315161281068908
95 0.315161281068908
96 0.315161281068908
97 0.315161281068908
98 0.315161281068908
99 0.315161281068908
};
\addplot [ultra thick, matlabpurple]
table {%
1 -0.0154
2 -0.003
3 0.0822
4 0.1418
5 0.1472
6 0.1728
7 0.2008
8 0.2328
9 0.236
10 0.2386
11 0.2446
12 0.2312
13 0.2602
14 0.251
15 0.2688
16 0.22
17 0.2556
18 0.2704
19 0.2582
20 0.2732
21 0.2832
22 0.2846
23 0.2782
24 0.2836
25 0.2758
26 0.2832
27 0.2614
28 0.2822
29 0.3008
30 0.2922
31 0.2802
32 0.2928
33 0.2962
34 0.2708
35 0.2946
36 0.2886
37 0.2854
38 0.3132
39 0.3032
40 0.275
41 0.2864
42 0.2666
43 0.2848
44 0.2832
45 0.302
46 0.3068
47 0.2714
48 0.301
49 0.2896
50 0.2944
51 0.304
52 0.3042
53 0.3036
54 0.3016
55 0.299
56 0.2892
57 0.2942
58 0.2934
59 0.29
60 0.2952
61 0.3062
62 0.3022
63 0.2862
64 0.3032
65 0.2822
66 0.2864
67 0.307
68 0.28
69 0.288
70 0.2994
71 0.2892
72 0.301
73 0.2966
74 0.2916
75 0.2952
76 0.2936
77 0.291
78 0.3012
79 0.2956
80 0.301
81 0.2848
82 0.3098
83 0.3042
84 0.2932
85 0.285
86 0.2914
87 0.3092
88 0.3156
89 0.2978
90 0.2936
91 0.2988
92 0.2992
93 0.286
94 0.3126
95 0.2716
96 0.3046
97 0.2738
98 0.3078
99 0.2876
};
\addlegendentry{$p = 0.8$}
\addplot [ultra thick, matlabpurple, dashed, forget plot]
table {%
1 0.305148253644866
2 0.305148253644866
3 0.305148253644866
4 0.305148253644866
5 0.305148253644866
6 0.305148253644866
7 0.305148253644866
8 0.305148253644866
9 0.305148253644866
10 0.305148253644866
11 0.305148253644866
12 0.305148253644866
13 0.305148253644866
14 0.305148253644866
15 0.305148253644866
16 0.305148253644866
17 0.305148253644866
18 0.305148253644866
19 0.305148253644866
20 0.305148253644866
21 0.305148253644866
22 0.305148253644866
23 0.305148253644866
24 0.305148253644866
25 0.305148253644866
26 0.305148253644866
27 0.305148253644866
28 0.305148253644866
29 0.305148253644866
30 0.305148253644866
31 0.305148253644866
32 0.305148253644866
33 0.305148253644866
34 0.305148253644866
35 0.305148253644866
36 0.305148253644866
37 0.305148253644866
38 0.305148253644866
39 0.305148253644866
40 0.305148253644866
41 0.305148253644866
42 0.305148253644866
43 0.305148253644866
44 0.305148253644866
45 0.305148253644866
46 0.305148253644866
47 0.305148253644866
48 0.305148253644866
49 0.305148253644866
50 0.305148253644866
51 0.305148253644866
52 0.305148253644866
53 0.305148253644866
54 0.305148253644866
55 0.305148253644866
56 0.305148253644866
57 0.305148253644866
58 0.305148253644866
59 0.305148253644866
60 0.305148253644866
61 0.305148253644866
62 0.305148253644866
63 0.305148253644866
64 0.305148253644866
65 0.305148253644866
66 0.305148253644866
67 0.305148253644866
68 0.305148253644866
69 0.305148253644866
70 0.305148253644866
71 0.305148253644866
72 0.305148253644866
73 0.305148253644866
74 0.305148253644866
75 0.305148253644866
76 0.305148253644866
77 0.305148253644866
78 0.305148253644866
79 0.305148253644866
80 0.305148253644866
81 0.305148253644866
82 0.305148253644866
83 0.305148253644866
84 0.305148253644866
85 0.305148253644866
86 0.305148253644866
87 0.305148253644866
88 0.305148253644866
89 0.305148253644866
90 0.305148253644866
91 0.305148253644866
92 0.305148253644866
93 0.305148253644866
94 0.305148253644866
95 0.305148253644866
96 0.305148253644866
97 0.305148253644866
98 0.305148253644866
99 0.305148253644866
};
\addplot [ultra thick, matlabblue]
table {%
1 0.1078
2 0.0922
3 0.124
4 0.1514
5 0.1584
6 0.1562
7 0.166
8 0.1664
9 0.181
10 0.173
11 0.1804
12 0.177
13 0.1756
14 0.1674
15 0.1744
16 0.1782
17 0.1784
18 0.1666
19 0.1766
20 0.1668
21 0.1646
22 0.1784
23 0.1866
24 0.1688
25 0.1766
26 0.1804
27 0.1778
28 0.1848
29 0.184
30 0.1702
31 0.1848
32 0.161
33 0.1722
34 0.1886
35 0.1712
36 0.1724
37 0.1812
38 0.1886
39 0.1798
40 0.1758
41 0.1742
42 0.1848
43 0.18
44 0.176
45 0.1826
46 0.1744
47 0.18
48 0.1874
49 0.1748
50 0.1756
51 0.1738
52 0.1782
53 0.1814
54 0.1756
55 0.1808
56 0.1786
57 0.199
58 0.1736
59 0.1656
60 0.1924
61 0.1816
62 0.1826
63 0.1928
64 0.19
65 0.176
66 0.1726
67 0.1724
68 0.1876
69 0.1832
70 0.1952
71 0.1794
72 0.1696
73 0.177
74 0.1742
75 0.1862
76 0.1842
77 0.1816
78 0.1838
79 0.1832
80 0.1722
81 0.1732
82 0.1904
83 0.1716
84 0.1788
85 0.1684
86 0.169
87 0.1944
88 0.1802
89 0.1808
90 0.185
91 0.1798
92 0.185
93 0.1858
94 0.1932
95 0.1864
96 0.1734
97 0.1876
98 0.1842
99 0.1724
};
\addlegendentry{$p = 0.9$}
\addplot [ultra thick, matlabblue, dashed, forget plot]
table {%
1 0.182013790037908
2 0.182013790037908
3 0.182013790037908
4 0.182013790037908
5 0.182013790037908
6 0.182013790037908
7 0.182013790037908
8 0.182013790037908
9 0.182013790037908
10 0.182013790037908
11 0.182013790037908
12 0.182013790037908
13 0.182013790037908
14 0.182013790037908
15 0.182013790037908
16 0.182013790037908
17 0.182013790037908
18 0.182013790037908
19 0.182013790037908
20 0.182013790037908
21 0.182013790037908
22 0.182013790037908
23 0.182013790037908
24 0.182013790037908
25 0.182013790037908
26 0.182013790037908
27 0.182013790037908
28 0.182013790037908
29 0.182013790037908
30 0.182013790037908
31 0.182013790037908
32 0.182013790037908
33 0.182013790037908
34 0.182013790037908
35 0.182013790037908
36 0.182013790037908
37 0.182013790037908
38 0.182013790037908
39 0.182013790037908
40 0.182013790037908
41 0.182013790037908
42 0.182013790037908
43 0.182013790037908
44 0.182013790037908
45 0.182013790037908
46 0.182013790037908
47 0.182013790037908
48 0.182013790037908
49 0.182013790037908
50 0.182013790037908
51 0.182013790037908
52 0.182013790037908
53 0.182013790037908
54 0.182013790037908
55 0.182013790037908
56 0.182013790037908
57 0.182013790037908
58 0.182013790037908
59 0.182013790037908
60 0.182013790037908
61 0.182013790037908
62 0.182013790037908
63 0.182013790037908
64 0.182013790037908
65 0.182013790037908
66 0.182013790037908
67 0.182013790037908
68 0.182013790037908
69 0.182013790037908
70 0.182013790037908
71 0.182013790037908
72 0.182013790037908
73 0.182013790037908
74 0.182013790037908
75 0.182013790037908
76 0.182013790037908
77 0.182013790037908
78 0.182013790037908
79 0.182013790037908
80 0.182013790037908
81 0.182013790037908
82 0.182013790037908
83 0.182013790037908
84 0.182013790037908
85 0.182013790037908
86 0.182013790037908
87 0.182013790037908
88 0.182013790037908
89 0.182013790037908
90 0.182013790037908
91 0.182013790037908
92 0.182013790037908
93 0.182013790037908
94 0.182013790037908
95 0.182013790037908
96 0.182013790037908
97 0.182013790037908
98 0.182013790037908
99 0.182013790037908
};
\addplot [ultra thick, black]
table {%
1 0
2 0
3 0
4 0
5 0
6 0
7 0
8 0
9 0
10 0
11 0
12 0
13 0
14 0
15 0
16 0
17 0
18 0
19 0
20 0
21 0
22 0
23 0
24 0
25 0
26 0
27 0
28 0
29 0
30 0
31 0
32 0
33 0
34 0
35 0
36 0
37 0
38 0
39 0
40 0
41 0
42 0
43 0
44 0
45 0
46 0
47 0
48 0
49 0
50 0
51 0
52 0
53 0
54 0
55 0
56 0
57 0
58 0
59 0
60 0
61 0
62 0
63 0
64 0
65 0
66 0
67 0
68 0
69 0
70 0
71 0
72 0
73 0
74 0
75 0
76 0
77 0
78 0
79 0
80 0
81 0
82 0
83 0
84 0
85 0
86 0
87 0
88 0
89 0
90 0
91 0
92 0
93 0
94 0
95 0
96 0
97 0
98 0
99 0
};
\addlegendentry{$p = 1$}
\addplot [ultra thick, black, dashed, forget plot]
table {%
1 0
2 0
3 0
4 0
5 0
6 0
7 0
8 0
9 0
10 0
11 0
12 0
13 0
14 0
15 0
16 0
17 0
18 0
19 0
20 0
21 0
22 0
23 0
24 0
25 0
26 0
27 0
28 0
29 0
30 0
31 0
32 0
33 0
34 0
35 0
36 0
37 0
38 0
39 0
40 0
41 0
42 0
43 0
44 0
45 0
46 0
47 0
48 0
49 0
50 0
51 0
52 0
53 0
54 0
55 0
56 0
57 0
58 0
59 0
60 0
61 0
62 0
63 0
64 0
65 0
66 0
67 0
68 0
69 0
70 0
71 0
72 0
73 0
74 0
75 0
76 0
77 0
78 0
79 0
80 0
81 0
82 0
83 0
84 0
85 0
86 0
87 0
88 0
89 0
90 0
91 0
92 0
93 0
94 0
95 0
96 0
97 0
98 0
99 0
};
\end{axis}

\end{tikzpicture}

%% file: plots/car_pedestrian.tex
\begin{tikzpicture}

\definecolor{darkgray176}{RGB}{176,176,176}
\definecolor{darkorange}{RGB}{255,140,0}
\definecolor{dodgerblue}{RGB}{30,144,255}
\definecolor{lightgray204}{RGB}{204,204,204}
\definecolor{limegreen}{RGB}{50,205,50}
\definecolor{slateblue}{RGB}{106,90,205}

\begin{axis}[
legend columns=3,
legend cell align={left},
legend style={
  fill opacity=0.8,
  draw opacity=1,
  text opacity=1,
  at={(0.97,0.03)},
  anchor=south east,
  draw=lightgray204
},
width=1\columnwidth,
height=0.45\columnwidth,
tick align=outside,
tick pos=left,
x grid style={darkgray176},
xlabel={\small Interactions (excluding the first interaction) },
xmin=-3.9, xmax=103.9,
xtick style={color=black},
y grid style={darkgray176},
ylabel={\small Average Car Return},
ymin=-1.65, ymax=1.1,
ytick style={color=black},
xmajorgrids,
ymajorgrids
]
\addplot [ultra thick, matlabred]
table {%
1 -1.2522
2 -1.2572
3 -1.2244
4 -1.2026
5 -1.1542
6 -1.1496
7 -1.1268
8 -1.1082
9 -1.1204
10 -1.0816
11 -1.084
12 -1.0442
13 -1.032
14 -1.0544
15 -0.9882
16 -1.0384
17 -0.9874
18 -0.9702
19 -0.9372
20 -0.9684
21 -0.9888
22 -0.937
23 -0.9674
24 -0.9156
25 -0.904
26 -0.917
27 -0.8374
28 -0.8768
29 -0.8652
30 -0.871
31 -0.7932
32 -0.8234
33 -0.8226
34 -0.7956
35 -0.7508
36 -0.7778
37 -0.7548
38 -0.7418
39 -0.7686
40 -0.74
41 -0.6838
42 -0.7232
43 -0.7132
44 -0.6864
45 -0.7208
46 -0.6872
47 -0.6872
48 -0.692
49 -0.63
50 -0.624
51 -0.6014
52 -0.6276
53 -0.6134
54 -0.5578
55 -0.6142
56 -0.573
57 -0.5846
58 -0.572
59 -0.6218
60 -0.6138
61 -0.5894
62 -0.5342
63 -0.5502
64 -0.5708
65 -0.5708
66 -0.549
67 -0.5394
68 -0.529
69 -0.5424
70 -0.5058
71 -0.5476
72 -0.5264
73 -0.5334
74 -0.5046
75 -0.4626
76 -0.4782
77 -0.4762
78 -0.4956
79 -0.4572
80 -0.487
81 -0.4738
82 -0.4878
83 -0.4428
84 -0.4578
85 -0.4202
86 -0.426
87 -0.4492
88 -0.4164
89 -0.4056
90 -0.433
91 -0.3832
92 -0.4142
93 -0.3952
94 -0.3592
95 -0.4204
96 -0.4374
97 -0.4066
98 -0.4062
99 -0.3422
};
\addlegendentry{\small $p = 0.53$}
\addplot [ultra thick, matlabred, dashed, forget plot]
table {%
1 0.928508689462746
2 0.928508689462746
3 0.928508689462746
4 0.928508689462746
5 0.928508689462746
6 0.928508689462746
7 0.928508689462746
8 0.928508689462746
9 0.928508689462746
10 0.928508689462746
11 0.928508689462746
12 0.928508689462746
13 0.928508689462746
14 0.928508689462746
15 0.928508689462746
16 0.928508689462746
17 0.928508689462746
18 0.928508689462746
19 0.928508689462746
20 0.928508689462746
21 0.928508689462746
22 0.928508689462746
23 0.928508689462746
24 0.928508689462746
25 0.928508689462746
26 0.928508689462746
27 0.928508689462746
28 0.928508689462746
29 0.928508689462746
30 0.928508689462746
31 0.928508689462746
32 0.928508689462746
33 0.928508689462746
34 0.928508689462746
35 0.928508689462746
36 0.928508689462746
37 0.928508689462746
38 0.928508689462746
39 0.928508689462746
40 0.928508689462746
41 0.928508689462746
42 0.928508689462746
43 0.928508689462746
44 0.928508689462746
45 0.928508689462746
46 0.928508689462746
47 0.928508689462746
48 0.928508689462746
49 0.928508689462746
50 0.928508689462746
51 0.928508689462746
52 0.928508689462746
53 0.928508689462746
54 0.928508689462746
55 0.928508689462746
56 0.928508689462746
57 0.928508689462746
58 0.928508689462746
59 0.928508689462746
60 0.928508689462746
61 0.928508689462746
62 0.928508689462746
63 0.928508689462746
64 0.928508689462746
65 0.928508689462746
66 0.928508689462746
67 0.928508689462746
68 0.928508689462746
69 0.928508689462746
70 0.928508689462746
71 0.928508689462746
72 0.928508689462746
73 0.928508689462746
74 0.928508689462746
75 0.928508689462746
76 0.928508689462746
77 0.928508689462746
78 0.928508689462746
79 0.928508689462746
80 0.928508689462746
81 0.928508689462746
82 0.928508689462746
83 0.928508689462746
84 0.928508689462746
85 0.928508689462746
86 0.928508689462746
87 0.928508689462746
88 0.928508689462746
89 0.928508689462746
90 0.928508689462746
91 0.928508689462746
92 0.928508689462746
93 0.928508689462746
94 0.928508689462746
95 0.928508689462746
96 0.928508689462746
97 0.928508689462746
98 0.928508689462746
99 0.928508689462746
};
\addplot [ultra thick, matlabyellow]
table {%
1 -0.7922
2 -0.77
3 -0.586
4 -0.5764
5 -0.469
6 -0.4934
7 -0.3584
8 -0.3536
9 -0.2706
10 -0.293
11 -0.1828
12 -0.1844
13 -0.106
14 -0.1248
15 -0.0342
16 -0.068
17 -0.0096
18 -0.0104
19 0.015
20 0.0246
21 0.0508
22 0.0724
23 0.1106
24 0.1102
25 0.1624
26 0.1696
27 0.217
28 0.22
29 0.2672
30 0.2456
31 0.2582
32 0.2862
33 0.2986
34 0.3224
35 0.3516
36 0.3196
37 0.3534
38 0.3772
39 0.3656
40 0.4052
41 0.3936
42 0.416
43 0.4618
44 0.4564
45 0.4676
46 0.474
47 0.496
48 0.4708
49 0.4808
50 0.4562
51 0.508
52 0.5286
53 0.5376
54 0.5092
55 0.5396
56 0.5498
57 0.5264
58 0.5512
59 0.5702
60 0.5626
61 0.555
62 0.5742
63 0.5718
64 0.5732
65 0.6142
66 0.5926
67 0.5988
68 0.6216
69 0.5906
70 0.6398
71 0.6312
72 0.6326
73 0.6236
74 0.6462
75 0.6338
76 0.6324
77 0.6426
78 0.6566
79 0.6302
80 0.6498
81 0.6704
82 0.674
83 0.6822
84 0.6558
85 0.6792
86 0.6758
87 0.6746
88 0.6846
89 0.6726
90 0.688
91 0.668
92 0.703
93 0.6928
94 0.6884
95 0.703
96 0.701
97 0.6896
98 0.7014
99 0.7246
};
\addlegendentry{\small $p = 0.6$}
\addplot [ultra thick, matlabyellow, dashed, forget plot]
table {%
1 0.799999991755385
2 0.799999991755385
3 0.799999991755385
4 0.799999991755385
5 0.799999991755385
6 0.799999991755385
7 0.799999991755385
8 0.799999991755385
9 0.799999991755385
10 0.799999991755385
11 0.799999991755385
12 0.799999991755385
13 0.799999991755385
14 0.799999991755385
15 0.799999991755385
16 0.799999991755385
17 0.799999991755385
18 0.799999991755385
19 0.799999991755385
20 0.799999991755385
21 0.799999991755385
22 0.799999991755385
23 0.799999991755385
24 0.799999991755385
25 0.799999991755385
26 0.799999991755385
27 0.799999991755385
28 0.799999991755385
29 0.799999991755385
30 0.799999991755385
31 0.799999991755385
32 0.799999991755385
33 0.799999991755385
34 0.799999991755385
35 0.799999991755385
36 0.799999991755385
37 0.799999991755385
38 0.799999991755385
39 0.799999991755385
40 0.799999991755385
41 0.799999991755385
42 0.799999991755385
43 0.799999991755385
44 0.799999991755385
45 0.799999991755385
46 0.799999991755385
47 0.799999991755385
48 0.799999991755385
49 0.799999991755385
50 0.799999991755385
51 0.799999991755385
52 0.799999991755385
53 0.799999991755385
54 0.799999991755385
55 0.799999991755385
56 0.799999991755385
57 0.799999991755385
58 0.799999991755385
59 0.799999991755385
60 0.799999991755385
61 0.799999991755385
62 0.799999991755385
63 0.799999991755385
64 0.799999991755385
65 0.799999991755385
66 0.799999991755385
67 0.799999991755385
68 0.799999991755385
69 0.799999991755385
70 0.799999991755385
71 0.799999991755385
72 0.799999991755385
73 0.799999991755385
74 0.799999991755385
75 0.799999991755385
76 0.799999991755385
77 0.799999991755385
78 0.799999991755385
79 0.799999991755385
80 0.799999991755385
81 0.799999991755385
82 0.799999991755385
83 0.799999991755385
84 0.799999991755385
85 0.799999991755385
86 0.799999991755385
87 0.799999991755385
88 0.799999991755385
89 0.799999991755385
90 0.799999991755385
91 0.799999991755385
92 0.799999991755385
93 0.799999991755385
94 0.799999991755385
95 0.799999991755385
96 0.799999991755385
97 0.799999991755385
98 0.799999991755385
99 0.799999991755385
};
\addplot [ultra thick, matlabgreen]
table {%
1 -0.3194
2 -0.3192
3 -0.0054
4 -0.0586
5 0.0884
6 0.1234
7 0.2436
8 0.204
9 0.3104
10 0.3054
11 0.3724
12 0.3828
13 0.4334
14 0.4224
15 0.4656
16 0.4442
17 0.4952
18 0.503
19 0.5146
20 0.4964
21 0.513
22 0.5296
23 0.5344
24 0.5416
25 0.5584
26 0.5746
27 0.5656
28 0.5624
29 0.5772
30 0.5676
31 0.5762
32 0.5798
33 0.559
34 0.582
35 0.592
36 0.5806
37 0.5954
38 0.5734
39 0.6014
40 0.5878
41 0.5832
42 0.5958
43 0.5826
44 0.588
45 0.5906
46 0.6078
47 0.5952
48 0.6
49 0.6022
50 0.6112
51 0.6
52 0.605
53 0.595
54 0.6098
55 0.605
56 0.5684
57 0.6094
58 0.5966
59 0.5984
60 0.5986
61 0.61
62 0.587
63 0.607
64 0.5904
65 0.5838
66 0.6022
67 0.5946
68 0.6022
69 0.5922
70 0.5984
71 0.5936
72 0.5972
73 0.592
74 0.6078
75 0.6024
76 0.5922
77 0.5942
78 0.5992
79 0.6128
80 0.6012
81 0.6002
82 0.5946
83 0.6006
84 0.5868
85 0.6088
86 0.5852
87 0.6146
88 0.6106
89 0.6022
90 0.5938
91 0.5958
92 0.614
93 0.595
94 0.585
95 0.5894
96 0.6002
97 0.6104
98 0.5772
99 0.599
};
\addlegendentry{\small $p = 0.7$}
\addplot [ultra thick, matlabgreen, dashed, forget plot]
table {%
1 0.6
2 0.6
3 0.6
4 0.6
5 0.6
6 0.6
7 0.6
8 0.6
9 0.6
10 0.6
11 0.6
12 0.6
13 0.6
14 0.6
15 0.6
16 0.6
17 0.6
18 0.6
19 0.6
20 0.6
21 0.6
22 0.6
23 0.6
24 0.6
25 0.6
26 0.6
27 0.6
28 0.6
29 0.6
30 0.6
31 0.6
32 0.6
33 0.6
34 0.6
35 0.6
36 0.6
37 0.6
38 0.6
39 0.6
40 0.6
41 0.6
42 0.6
43 0.6
44 0.6
45 0.6
46 0.6
47 0.6
48 0.6
49 0.6
50 0.6
51 0.6
52 0.6
53 0.6
54 0.6
55 0.6
56 0.6
57 0.6
58 0.6
59 0.6
60 0.6
61 0.6
62 0.6
63 0.6
64 0.6
65 0.6
66 0.6
67 0.6
68 0.6
69 0.6
70 0.6
71 0.6
72 0.6
73 0.6
74 0.6
75 0.6
76 0.6
77 0.6
78 0.6
79 0.6
80 0.6
81 0.6
82 0.6
83 0.6
84 0.6
85 0.6
86 0.6
87 0.6
88 0.6
89 0.6
90 0.6
91 0.6
92 0.6
93 0.6
94 0.6
95 0.6
96 0.6
97 0.6
98 0.6
99 0.6
};
\addplot [ultra thick, matlabpurple]
table {%
1 -0.0086
2 0.0038
3 0.2038
4 0.1838
5 0.2872
6 0.286
7 0.3256
8 0.3314
9 0.3554
10 0.3676
11 0.3708
12 0.3896
13 0.3748
14 0.3856
15 0.378
16 0.3822
17 0.3912
18 0.3988
19 0.3916
20 0.3898
21 0.383
22 0.403
23 0.3974
24 0.381
25 0.3918
26 0.3994
27 0.3944
28 0.4016
29 0.3938
30 0.3956
31 0.3944
32 0.415
33 0.4186
34 0.3916
35 0.4086
36 0.4024
37 0.393
38 0.4008
39 0.4112
40 0.3926
41 0.4042
42 0.3994
43 0.4036
44 0.41
45 0.4114
46 0.403
47 0.398
48 0.406
49 0.4076
50 0.407
51 0.3946
52 0.3958
53 0.3854
54 0.3982
55 0.397
56 0.4096
57 0.3942
58 0.3978
59 0.3934
60 0.3906
61 0.3998
62 0.4056
63 0.4006
64 0.4072
65 0.4134
66 0.4078
67 0.413
68 0.3948
69 0.4042
70 0.4144
71 0.4046
72 0.4074
73 0.4044
74 0.3936
75 0.4062
76 0.4094
77 0.405
78 0.416
79 0.3954
80 0.4002
81 0.4178
82 0.404
83 0.4074
84 0.3968
85 0.3914
86 0.421
87 0.395
88 0.3914
89 0.3982
90 0.3942
91 0.4014
92 0.4052
93 0.4012
94 0.409
95 0.3966
96 0.3884
97 0.4068
98 0.3938
99 0.4094
};
\addlegendentry{\small $p = 0.8$}
\addplot [ultra thick, matlabpurple, dashed, forget plot]
table {%
1 0.4
2 0.4
3 0.4
4 0.4
5 0.4
6 0.4
7 0.4
8 0.4
9 0.4
10 0.4
11 0.4
12 0.4
13 0.4
14 0.4
15 0.4
16 0.4
17 0.4
18 0.4
19 0.4
20 0.4
21 0.4
22 0.4
23 0.4
24 0.4
25 0.4
26 0.4
27 0.4
28 0.4
29 0.4
30 0.4
31 0.4
32 0.4
33 0.4
34 0.4
35 0.4
36 0.4
37 0.4
38 0.4
39 0.4
40 0.4
41 0.4
42 0.4
43 0.4
44 0.4
45 0.4
46 0.4
47 0.4
48 0.4
49 0.4
50 0.4
51 0.4
52 0.4
53 0.4
54 0.4
55 0.4
56 0.4
57 0.4
58 0.4
59 0.4
60 0.4
61 0.4
62 0.4
63 0.4
64 0.4
65 0.4
66 0.4
67 0.4
68 0.4
69 0.4
70 0.4
71 0.4
72 0.4
73 0.4
74 0.4
75 0.4
76 0.4
77 0.4
78 0.4
79 0.4
80 0.4
81 0.4
82 0.4
83 0.4
84 0.4
85 0.4
86 0.4
87 0.4
88 0.4
89 0.4
90 0.4
91 0.4
92 0.4
93 0.4
94 0.4
95 0.4
96 0.4
97 0.4
98 0.4
99 0.4
};
\addplot [ultra thick, matlabblue]
table {%
1 0.1148
2 0.1022
3 0.1622
4 0.1714
5 0.1898
6 0.1896
7 0.1932
8 0.204
9 0.203
10 0.2092
11 0.1976
12 0.1994
13 0.2018
14 0.2072
15 0.198
16 0.2116
17 0.1882
18 0.2078
19 0.2024
20 0.1994
21 0.1958
22 0.1962
23 0.2052
24 0.2036
25 0.1978
26 0.195
27 0.21
28 0.1968
29 0.1938
30 0.2044
31 0.1902
32 0.2056
33 0.1998
34 0.2068
35 0.1928
36 0.1978
37 0.2022
38 0.2022
39 0.2016
40 0.2056
41 0.2008
42 0.1992
43 0.192
44 0.2
45 0.1912
46 0.2012
47 0.2152
48 0.2054
49 0.205
50 0.1948
51 0.1986
52 0.1986
53 0.196
54 0.2046
55 0.1936
56 0.2002
57 0.2014
58 0.2016
59 0.2052
60 0.2064
61 0.1968
62 0.187
63 0.212
64 0.2012
65 0.199
66 0.1936
67 0.2064
68 0.2066
69 0.21
70 0.1876
71 0.1946
72 0.1998
73 0.2018
74 0.203
75 0.194
76 0.1958
77 0.2058
78 0.1924
79 0.192
80 0.1982
81 0.1978
82 0.1966
83 0.1994
84 0.1944
85 0.2024
86 0.2108
87 0.2016
88 0.1974
89 0.1972
90 0.1932
91 0.199
92 0.197
93 0.191
94 0.1932
95 0.2048
96 0.2012
97 0.197
98 0.2122
99 0.2044
};
\addlegendentry{\small $p = 0.9$}
\addplot [ultra thick, matlabblue, dashed, forget plot]
table {%
1 0.2
2 0.2
3 0.2
4 0.2
5 0.2
6 0.2
7 0.2
8 0.2
9 0.2
10 0.2
11 0.2
12 0.2
13 0.2
14 0.2
15 0.2
16 0.2
17 0.2
18 0.2
19 0.2
20 0.2
21 0.2
22 0.2
23 0.2
24 0.2
25 0.2
26 0.2
27 0.2
28 0.2
29 0.2
30 0.2
31 0.2
32 0.2
33 0.2
34 0.2
35 0.2
36 0.2
37 0.2
38 0.2
39 0.2
40 0.2
41 0.2
42 0.2
43 0.2
44 0.2
45 0.2
46 0.2
47 0.2
48 0.2
49 0.2
50 0.2
51 0.2
52 0.2
53 0.2
54 0.2
55 0.2
56 0.2
57 0.2
58 0.2
59 0.2
60 0.2
61 0.2
62 0.2
63 0.2
64 0.2
65 0.2
66 0.2
67 0.2
68 0.2
69 0.2
70 0.2
71 0.2
72 0.2
73 0.2
74 0.2
75 0.2
76 0.2
77 0.2
78 0.2
79 0.2
80 0.2
81 0.2
82 0.2
83 0.2
84 0.2
85 0.2
86 0.2
87 0.2
88 0.2
89 0.2
90 0.2
91 0.2
92 0.2
93 0.2
94 0.2
95 0.2
96 0.2
97 0.2
98 0.2
99 0.2
};
\addplot [ultra thick, black]
table {%
1 0
2 0
3 0
4 0
5 0
6 0
7 0
8 0
9 0
10 0
11 0
12 0
13 0
14 0
15 0
16 0
17 0
18 0
19 0
20 0
21 0
22 0
23 0
24 0
25 0
26 0
27 0
28 0
29 0
30 0
31 0
32 0
33 0
34 0
35 0
36 0
37 0
38 0
39 0
40 0
41 0
42 0
43 0
44 0
45 0
46 0
47 0
48 0
49 0
50 0
51 0
52 0
53 0
54 0
55 0
56 0
57 0
58 0
59 0
60 0
61 0
62 0
63 0
64 0
65 0
66 0
67 0
68 0
69 0
70 0
71 0
72 0
73 0
74 0
75 0
76 0
77 0
78 0
79 0
80 0
81 0
82 0
83 0
84 0
85 0
86 0
87 0
88 0
89 0
90 0
91 0
92 0
93 0
94 0
95 0
96 0
97 0
98 0
99 0
};
\addlegendentry{\small $p = 1$}
\addplot [ultra thick, black, dashed, forget plot]
table {%
1 0
2 0
3 0
4 0
5 0
6 0
7 0
8 0
9 0
10 0
11 0
12 0
13 0
14 0
15 0
16 0
17 0
18 0
19 0
20 0
21 0
22 0
23 0
24 0
25 0
26 0
27 0
28 0
29 0
30 0
31 0
32 0
33 0
34 0
35 0
36 0
37 0
38 0
39 0
40 0
41 0
42 0
43 0
44 0
45 0
46 0
47 0
48 0
49 0
50 0
51 0
52 0
53 0
54 0
55 0
56 0
57 0
58 0
59 0
60 0
61 0
62 0
63 0
64 0
65 0
66 0
67 0
68 0
69 0
70 0
71 0
72 0
73 0
74 0
75 0
76 0
77 0
78 0
79 0
80 0
81 0
82 0
83 0
84 0
85 0
86 0
87 0
88 0
89 0
90 0
91 0
92 0
93 0
94 0
95 0
96 0
97 0
98 0
99 0
};
\end{axis}

\end{tikzpicture}

%% file: plots/tug_of_war.tex
\begin{tikzpicture}

\definecolor{darkgray176}{RGB}{176,176,176}
\definecolor{darkorange}{RGB}{255,140,0}
\definecolor{dodgerblue}{RGB}{30,144,255}
\definecolor{lightgray204}{RGB}{204,204,204}
\definecolor{limegreen}{RGB}{50,205,50}
\definecolor{slateblue}{RGB}{106,90,205}

\begin{axis}[
legend columns=3,
legend cell align={left},
legend style={
  fill opacity=0.8,
  draw opacity=1,
  text opacity=1,
  at={(0.97,0.03)},
  anchor=south east,
  draw=lightgray204
},
width=1\columnwidth,
height=0.45\columnwidth,
tick align=outside,
tick pos=left,
x grid style={darkgray176},
xlabel={\small Std. $s^{l}$ },
xmin=2, xmax=7,
xtick style={color=black},
y grid style={darkgray176},
ylabel={\small Return},
ymin=0.15, ymax=.25,
ytick style={color=black},
xmajorgrids,
ymajorgrids
]
\addplot [ultra thick, dashed, black]
table {%
0.1 0
0.16969696969697 0
0.239393939393939 0
0.309090909090909 2.22044604925031e-16
0.378787878787879 1.57267532330252e-11
0.448484848484848 1.02474907448524e-08
0.518181818181818 6.06479870102206e-07
0.587878787878788 9.42507260959236e-06
0.657575757575758 6.54778339752582e-05
0.727272727272727 0.00027197447141547
0.796969696969697 0.000800557868356977
0.866666666666667 0.00185439799043896
0.936363636363636 0.00361649501967798
1.00606060606061 0.00621230562702757
1.07575757575758 0.00969654700517297
1.14545454545455 0.0140590423117871
1.21515151515152 0.0192401880008409
1.28484848484849 0.025148508858806
1.35454545454545 0.0316760966261833
1.42424242424242 0.0387103510769459
1.49393939393939 0.0461419568708431
1.56363636363636 0.0538696979028559
1.63333333333333 0.0618028887819643
1.7030303030303 0.0698621513500226
1.77272727272727 0.0779791316135969
1.84242424242424 0.0860956089884215
1.91212121212121 0.0941623217416139
1.98181818181818 0.102137726893366
2.05151515151515 0.109986829612281
2.12121212121212 0.117680154403493
2.19090909090909 0.125192885718688
2.26060606060606 0.132504176250259
2.33030303030303 0.139596604082138
2.4 0.146455751973128
2.46969696969697 0.153069880472353
2.53939393939394 0.159429668872473
2.60909090909091 0.165528002303689
2.67878787878788 0.17135978822568
2.74848484848485 0.176921790338612
2.81818181818182 0.182212472046305
2.88787878787879 0.187231844876801
2.95757575757576 0.191981319687116
3.02727272727273 0.196463560139456
3.0969696969697 0.200682338970538
3.16666666666667 0.204642398129023
3.23636363636364 0.20834931406291
3.30606060606061 0.211809369413229
3.37575757575758 0.215029432202798
3.44545454545455 0.218016843366919
3.51515151515152 0.220779313204558
3.58484848484849 0.223324827065386
3.65454545454545 0.22566156034963
3.72424242424242 0.227797802693921
3.79393939393939 0.229741891051392
3.86363636363636 0.231502151247627
3.93333333333333 0.233086847502945
4.0030303030303 0.234504139351603
4.07272727272727 0.23576204535501
4.14242424242424 0.236868412994085
4.21212121212121 0.237830894130915
4.28181818181818 0.238656925447788
4.35151515151515 0.239353713298919
4.42121212121212 0.239928222443889
4.49090909090909 0.240387168169446
4.56060606060606 0.240737011346051
4.63030303030303 0.240983956005822
4.7 0.241133949068296
4.76969696969697 0.241192681878762
4.83939393939394 0.241165593260388
4.90909090909091 0.241057873815483
4.97878787878788 0.240874471242844
5.04848484848485 0.240620096467175
5.11818181818182 0.240299230402926
5.18787878787879 0.239916131198815
5.25757575757576 0.23947484183065
5.32727272727273 0.23897919792924
5.3969696969697 0.238432835747149
5.46666666666667 0.23783920018303
5.53636363636364 0.237201552795535
5.60606060606061 0.236522979750292
5.67575757575758 0.235806399653589
5.74545454545455 0.235054571235125
5.81515151515152 0.234270100849784
5.88484848484848 0.233455449774912
5.95454545454546 0.232612941285168
6.02424242424242 0.231744767491804
6.09393939393939 0.230852995937283
6.16363636363636 0.229939575939559
6.23333333333333 0.22900634468322
6.3030303030303 0.228055033057109
6.37272727272727 0.227087271239981
6.44242424242424 0.226104594037406
6.51212121212121 0.225108445974425
6.58181818181818 0.224100186149501
6.65151515151515 0.223081092856156
6.72121212121212 0.222052367979259
6.79090909090909 0.221015141173447
6.86060606060606 0.219970473831435
6.93030303030303 0.218919362850188
7 0.217862744203049
};
\addlegendentry{$SR$}
\addplot [ultra thick, matlabyellow]
table {%
0.1 0
0.16969696969697 0
0.239393939393939 6.46319775476911e-15
0.309090909090909 2.23151141995781e-09
0.378787878787879 1.25757966809109e-06
0.448484848484848 3.29639989682104e-05
0.518181818181818 0.000298035148444728
0.587878787878788 0.00120677986306705
0.657575757575758 0.00337906010528901
0.727272727272727 0.00738796790867797
0.796969696969697 0.0129224329732254
0.866666666666667 0.0198635420616987
0.936363636363636 0.0292571988406493
1.00606060606061 0.0387843372797571
1.07575757575758 0.0486449040233292
1.14545454545455 0.0603849951938412
1.21515151515152 0.0710299615332109
1.28484848484849 0.082750216459174
1.35454545454545 0.0942267023221692
1.42424242424242 0.105877734444584
1.49393939393939 0.116953225436857
1.56363636363636 0.127471631552089
1.63333333333333 0.138073238943164
1.7030303030303 0.147159848444845
1.77272727272727 0.157022064988838
1.84242424242424 0.16489985125613
1.91212121212121 0.174294811684971
1.98181818181818 0.180482738540251
2.05151515151515 0.188941421705067
2.12121212121212 0.194608066111979
2.19090909090909 0.200170568484835
2.26060606060606 0.205558946892461
2.33030303030303 0.210470974785791
2.4 0.216088751031392
2.46969696969697 0.219568504940256
2.53939393939394 0.223750761296296
2.60909090909091 0.227656887268806
2.67878787878788 0.229516095207307
2.74848484848485 0.231456889991304
2.81818181818182 0.23338567274236
2.88787878787879 0.235495331793014
2.95757575757576 0.237698854797652
3.02727272727273 0.237778297239237
3.0969696969697 0.23911423834107
3.16666666666667 0.240750657543085
3.23636363636364 0.241778701189227
3.30606060606061 0.239838840522048
3.37575757575758 0.241129321632647
3.44545454545455 0.241224969301179
3.51515151515152 0.240645701377884
3.58484848484849 0.241523536808997
3.65454545454545 0.240028602554993
3.72424242424242 0.239021511416044
3.79393939393939 0.238658003809336
3.86363636363636 0.237895675478427
3.93333333333333 0.237349418369156
4.0030303030303 0.236313056291781
4.07272727272727 0.234287129343874
4.14242424242424 0.232934220595788
4.21212121212121 0.233314923349403
4.28181818181818 0.231390250595477
4.35151515151515 0.230515386254858
4.42121212121212 0.228150985576344
4.49090909090909 0.227484947876871
4.56060606060606 0.226991186994605
4.63030303030303 0.224697931951719
4.7 0.223137774257613
4.76969696969697 0.222428263931872
4.83939393939394 0.219870021342457
4.90909090909091 0.218377322523504
4.97878787878788 0.21709212935574
5.04848484848485 0.214883745394673
5.11818181818182 0.214511413234117
5.18787878787879 0.212793771625972
5.25757575757576 0.211363266630571
5.32727272727273 0.209293777764462
5.3969696969697 0.208262326969032
5.46666666666667 0.206800085337184
5.53636363636364 0.204888323536936
5.60606060606061 0.203344113226171
5.67575757575758 0.20233384278104
5.74545454545455 0.200614089958703
5.81515151515152 0.199303476142227
5.88484848484848 0.198088897560175
5.95454545454546 0.19648815473839
6.02424242424242 0.194609094369491
6.09393939393939 0.193655661413155
6.16363636363636 0.191669799276471
6.23333333333333 0.189807208914419
6.3030303030303 0.188377042356607
6.37272727272727 0.187622839794748
6.44242424242424 0.186215215910628
6.51212121212121 0.184936819142279
6.58181818181818 0.183685398638865
6.65151515151515 0.181925979394864
6.72121212121212 0.180869736922944
6.79090909090909 0.179526420128185
6.86060606060606 0.17848044506347
6.93030303030303 0.176453951771795
7 0.175501296414511
};
\addlegendentry{$IR_{2}$}
\addplot [ultra thick, matlabgreen]
table {%
0.1 0
0.16969696969697 0
0.239393939393939 5.89528426075958e-19
0.309090909090909 1.25033211717529e-11
0.378787878787879 2.69104785146801e-08
0.448484848484848 2.41251289696182e-06
0.518181818181818 3.64932515084794e-05
0.587878787878788 0.000235934175107078
0.657575757575758 0.000882192744085187
0.727272727272727 0.00231245467822073
0.796969696969697 0.00492817229192887
0.866666666666667 0.00885357916551773
0.936363636363636 0.0142770749785257
1.00606060606061 0.0206818674428781
1.07575757575758 0.0277737366878773
1.14545454545455 0.03653120095384
1.21515151515152 0.0455864051756431
1.28484848484849 0.0551516300537875
1.35454545454545 0.064301412786475
1.42424242424242 0.074886388466874
1.49393939393939 0.0847309106179334
1.56363636363636 0.0948188759012116
1.63333333333333 0.103958529824519
1.7030303030303 0.113616089998806
1.77272727272727 0.123532973110805
1.84242424242424 0.132288847942247
1.91212121212121 0.140811546229934
1.98181818181818 0.148905309682692
2.05151515151515 0.157102230320128
2.12121212121212 0.164434032602695
2.19090909090909 0.171772564364842
2.26060606060606 0.177728378595582
2.33030303030303 0.185611204474888
2.4 0.190361884589125
2.46969696969697 0.196389015445571
2.53939393939394 0.201802356531871
2.60909090909091 0.206482458932418
2.67878787878788 0.210350258312454
2.74848484848485 0.214903531007063
2.81818181818182 0.218115565973537
2.88787878787879 0.220901533583174
2.95757575757576 0.224237706585729
3.02727272727273 0.22755948344954
3.0969696969697 0.2298546916046
3.16666666666667 0.232185179469742
3.23636363636364 0.233897679280283
3.30606060606061 0.235288421500705
3.37575757575758 0.236944878978717
3.44545454545455 0.23781721336601
3.51515151515152 0.239226240943676
3.58484848484849 0.239968661638548
3.65454545454545 0.240428828728249
3.72424242424242 0.240414696234299
3.79393939393939 0.241354159362839
3.86363636363636 0.241350197530114
3.93333333333333 0.241215080904389
4.0030303030303 0.24162190386157
4.07272727272727 0.241371565927244
4.14242424242424 0.240451859940766
4.21212121212121 0.239615355226735
4.28181818181818 0.239705369246772
4.35151515151515 0.238591067907944
4.42121212121212 0.237835582648259
4.49090909090909 0.237622991730608
4.56060606060606 0.236605021115559
4.63030303030303 0.235636960014677
4.7 0.234853164544297
4.76969696969697 0.233780999414915
4.83939393939394 0.233205921139861
4.90909090909091 0.232042379730847
4.97878787878788 0.23108591391895
5.04848484848485 0.229804957474764
5.11818181818182 0.228470528936326
5.18787878787879 0.227330483423516
5.25757575757576 0.226425349038162
5.32727272727273 0.224996881484958
5.3969696969697 0.223567408432635
5.46666666666667 0.221965468246867
5.53636363636364 0.22114364464363
5.60606060606061 0.219707662896088
5.67575757575758 0.217857016162519
5.74545454545455 0.217166783320059
5.81515151515152 0.21572992612078
5.88484848484848 0.215080615346396
5.95454545454546 0.213518896546105
6.02424242424242 0.211837531070466
6.09393939393939 0.210793140802592
6.16363636363636 0.209047491237714
6.23333333333333 0.208208558602887
6.3030303030303 0.206554444386008
6.37272727272727 0.205278064830626
6.44242424242424 0.203735340068829
6.51212121212121 0.202754158068567
6.58181818181818 0.201455355301695
6.65151515151515 0.200395801484275
6.72121212121212 0.199186145165498
6.79090909090909 0.197590106498962
6.86060606060606 0.1968030975867
6.93030303030303 0.195114738432135
7 0.193875899436818
};
\addlegendentry{$IR_{3}$}
\addplot [ultra thick, matlabblue]
table {%
0.1 0
0.16969696969697 0
0.239393939393939 7.7715611723761e-21
0.309090909090909 1.39366536533458e-12
0.378787878787879 4.11677969526014e-09
0.448484848484848 6.18368385380442e-07
0.518181818181818 1.2668203309525e-05
0.587878787878788 0.000103364387390743
0.657575757575758 0.00046766277702039
0.727272727272727 0.00136658226062453
0.796969696969697 0.00316634845870436
0.866666666666667 0.0059390615807998
0.936363636363636 0.00997688661016231
1.00606060606061 0.0152053150244957
1.07575757575758 0.0215598973856343
1.14545454545455 0.0287410888169499
1.21515151515152 0.0365139276018895
1.28484848484849 0.0450900008637991
1.35454545454545 0.0539666539258863
1.42424242424242 0.0628505390801254
1.49393939393939 0.0727226794311953
1.56363636363636 0.0818749544883301
1.63333333333333 0.0910618120701461
1.7030303030303 0.100006609896079
1.77272727272727 0.109635153458706
1.84242424242424 0.118589620958924
1.91212121212121 0.127051371660384
1.98181818181818 0.135782845275832
2.05151515151515 0.143364747800504
2.12121212121212 0.150965965652644
2.19090909090909 0.157876454658946
2.26060606060606 0.16516324852782
2.33030303030303 0.172843290527757
2.4 0.17970325296899
2.46969696969697 0.185343740494848
2.53939393939394 0.190348885962421
2.60909090909091 0.195558644199295
2.67878787878788 0.200206163430078
2.74848484848485 0.204787018688132
2.81818181818182 0.209398039416823
2.88787878787879 0.212747224749774
2.95757575757576 0.216941165913069
3.02727272727273 0.21980602170807
3.0969696969697 0.223702081569484
3.16666666666667 0.226528173922416
3.23636363636364 0.228401269781772
3.30606060606061 0.230381114892516
3.37575757575758 0.232862232331849
3.44545454545455 0.234295959533741
3.51515151515152 0.235640696794952
3.58484848484849 0.236798214643665
3.65454545454545 0.238129301781054
3.72424242424242 0.238668547743017
3.79393939393939 0.239780349034629
3.86363636363636 0.239602888815145
3.93333333333333 0.24083829811219
4.0030303030303 0.241197388586178
4.07272727272727 0.241304731194598
4.14242424242424 0.241581303655329
4.21212121212121 0.240370685432061
4.28181818181818 0.240799557643576
4.35151515151515 0.240471210023896
4.42121212121212 0.240805311321024
4.49090909090909 0.240426540023606
4.56060606060606 0.239084493113184
4.63030303030303 0.238645027073201
4.7 0.238220436942011
4.76969696969697 0.236744675267306
4.83939393939394 0.236496722685645
4.90909090909091 0.236277044062803
4.97878787878788 0.234827262140604
5.04848484848485 0.234258977125058
5.11818181818182 0.233221247531105
5.18787878787879 0.232421541671365
5.25757575757576 0.231236691954608
5.32727272727273 0.229950359007604
5.3969696969697 0.228807297034877
5.46666666666667 0.228045728270139
5.53636363636364 0.226540642843372
5.60606060606061 0.225590555721161
5.67575757575758 0.224827613955497
5.74545454545455 0.223326129303808
5.81515151515152 0.221731581225725
5.88484848484848 0.220724507413516
5.95454545454546 0.21986868891433
6.02424242424242 0.218641929120843
6.09393939393939 0.216701678597606
6.16363636363636 0.216140205817404
6.23333333333333 0.215228648706405
6.3030303030303 0.213648255596231
6.37272727272727 0.212172888880316
6.44242424242424 0.211528406622992
6.51212121212121 0.20981033256162
6.58181818181818 0.208601884651491
6.65151515151515 0.207057007808899
6.72121212121212 0.20605701409942
6.79090909090909 0.204947654568234
6.86060606060606 0.20356282199073
6.93030303030303 0.202712524379787
7 0.201266911152721
};
\addlegendentry{$IR_{4}$}
\addplot [ultra thick, matlabpurple]
table {%
0.1 0
0.16969696969697 0
0.239393939393939 0
0.309090909090909 1.26961832158656e-13
0.378787878787879 1.36388458171877e-09
0.448484848484848 2.66452512830752e-07
0.518181818181818 7.01497669835238e-06
0.587878787878788 6.62644906707557e-05
0.657575757575758 0.000316345755033901
0.727272727272727 0.000990382952923606
0.796969696969697 0.00241775968610459
0.866666666666667 0.00471733821951101
0.936363636363636 0.00810705954130584
1.00606060606061 0.0125983090061873
1.07575757575758 0.0182161400832714
1.14545454545455 0.0247495999919216
1.21515151515152 0.0321059284875178
1.28484848484849 0.0401528798409609
1.35454545454545 0.0483754848794102
1.42424242424242 0.0570940676084519
1.49393939393939 0.0656646382385267
1.56363636363636 0.0748940040213627
1.63333333333333 0.0845311328380147
1.7030303030303 0.0935476986344411
1.77272727272727 0.101949093434782
1.84242424242424 0.11154829953156
1.91212121212121 0.119658943898086
1.98181818181818 0.128023673380118
2.05151515151515 0.136299679701574
2.12121212121212 0.143630196369151
2.19090909090909 0.151456055176178
2.26060606060606 0.158937780418261
2.33030303030303 0.165104591453164
2.4 0.171158929706601
2.46969696969697 0.17867608131007
2.53939393939394 0.184136400418716
2.60909090909091 0.189138647223904
2.67878787878788 0.19437081428365
2.74848484848485 0.199128194741621
2.81818181818182 0.203359363047609
2.88787878787879 0.208250508114295
2.95757575757576 0.211625426263775
3.02727272727273 0.215542683380063
3.0969696969697 0.218887227008739
3.16666666666667 0.221575539732343
3.23636363636364 0.224154099133858
3.30606060606061 0.227281016282046
3.37575757575758 0.228164721162152
3.44545454545455 0.230697450262304
3.51515151515152 0.23314997580524
3.58484848484849 0.234498007312413
3.65454545454545 0.235745454098529
3.72424242424242 0.237416501617481
3.79393939393939 0.238199082774704
3.86363636363636 0.238647084143865
3.93333333333333 0.239407952954326
4.0030303030303 0.240722962504389
4.07272727272727 0.240886320086343
4.14242424242424 0.241018319510683
4.21212121212121 0.240741586424678
4.28181818181818 0.241404083035776
4.35151515151515 0.241211014972716
4.42121212121212 0.241234310146477
4.49090909090909 0.2407123869542
4.56060606060606 0.241087828270197
4.63030303030303 0.239690397959221
4.7 0.239596321262462
4.76969696969697 0.238823417712207
4.83939393939394 0.238129529849527
4.90909090909091 0.238209865758422
4.97878787878788 0.237018071703127
5.04848484848485 0.235880394919288
5.11818181818182 0.235496583594858
5.18787878787879 0.23403340862571
5.25757575757576 0.233679216455954
5.32727272727273 0.232743703351157
5.3969696969697 0.231505503636504
5.46666666666667 0.230667512066469
5.53636363636364 0.228977790945979
5.60606060606061 0.228078735349823
5.67575757575758 0.227443833688832
5.74545454545455 0.226557460852989
5.81515151515152 0.224971570417805
5.88484848484848 0.224311895006204
5.95454545454546 0.223021593945784
6.02424242424242 0.221678315694067
6.09393939393939 0.221167095979023
6.16363636363636 0.219150489750742
6.23333333333333 0.218127896958414
6.3030303030303 0.216817622647583
6.37272727272727 0.215776727665115
6.44242424242424 0.214827943982914
6.51212121212121 0.213644461214823
6.58181818181818 0.212247917647616
6.65151515151515 0.21102172833381
6.72121212121212 0.209927307433128
6.79090909090909 0.208440456623485
6.86060606060606 0.207458490973921
6.93030303030303 0.206629920260819
7 0.20524716166626
};
\addlegendentry{$IR_{5}$}
\end{axis}

\end{tikzpicture}

%% file: plots/optimal_C_random_matrix.tex
\begin{tikzpicture}

\definecolor{lightgray204}{RGB}{204,204,204}

\begin{axis}[
legend cell align={left},
legend style={legend columns=2,
  fill opacity=.8,
  draw opacity=1,
  text opacity=1,
  at={(0.95,0.05)},
  anchor=south east,
  draw=lightgray204
},
width=1\columnwidth,
height=0.5\columnwidth,
tick align=outside,
tick pos=left,
xlabel={\small Interaction (excluding the first intearction) },
xmin=-3.9, xmax=103.9,
ylabel={\small Average Leader Return},
ymin=0.648, ymax=0.701,
tick align=inside,
xmajorgrids,
ymajorgrids,
]  

\addplot [ultra thick, matlabred]
table {%
1 0.648749556461158
2 0.657172940891823
3 0.660036966952219
4 0.663343690071005
5 0.665475485089524
6 0.667070872190684
7 0.668496448071139
8 0.669931283590544
9 0.67104715647181
10 0.671977140660178
11 0.673004969410564
12 0.673737214050938
13 0.674445383203007
14 0.675225803250351
15 0.675781921150969
16 0.67620554142884
17 0.676894117579644
18 0.677624727444157
19 0.677854822573738
20 0.678460517325926
21 0.679152646653627
22 0.679414043028434
23 0.679645180081044
24 0.679553415385433
25 0.680528527932532
26 0.680700370088451
27 0.680639159783247
28 0.681112969791067
29 0.681445350958165
30 0.681660857604715
31 0.682032860885874
32 0.682158735147175
33 0.682828198677504
34 0.683133769766439
35 0.682936431583178
36 0.68326028508159
37 0.683728785231344
38 0.683657828336775
39 0.683756800413931
40 0.683958647302993
41 0.684418011906593
42 0.684326166489649
43 0.6847509791882
44 0.684469111290765
45 0.684949087967338
46 0.684912655465894
47 0.685350541754499
48 0.685432134310197
49 0.685517841944282
50 0.685852555462315
51 0.685654997227934
52 0.685880497668057
53 0.686203026258511
54 0.686152529756763
55 0.68636586424419
56 0.686346133197566
57 0.686595752006246
58 0.686424058954006
59 0.686835205732427
60 0.686694356282692
61 0.686941225120738
62 0.687045114721736
63 0.687182410387931
64 0.687471625998109
65 0.687573627456848
66 0.687661155545395
67 0.687716888397682
68 0.687673065541673
69 0.68788046605403
70 0.688029289893632
71 0.688120071772817
72 0.687993596191126
73 0.688355248999082
74 0.688293599195558
75 0.688314494588264
76 0.688341861911885
77 0.688550358101906
78 0.688588761136784
79 0.688620648287804
80 0.688975417602581
81 0.688925914154871
82 0.688877813770419
83 0.688842618726373
84 0.68910480902285
85 0.689169358282672
86 0.689220150513976
87 0.689097377496027
88 0.689389531216804
89 0.689269166241973
90 0.689580667932949
91 0.689534102771597
92 0.689368099487231
93 0.689657198188921
94 0.689630537793246
95 0.689717966951297
96 0.689927029772585
97 0.689639077506202
98 0.690106903107708
99 0.69014554034165
};
\addlegendentry{ \small $x^*(c = 0)$}
\addplot [ultra thick, matlabred, dashed, forget plot]
table {%
1 0.696987971599779
2 0.696987971599779
3 0.696987971599779
4 0.696987971599779
5 0.696987971599779
6 0.696987971599779
7 0.696987971599779
8 0.696987971599779
9 0.696987971599779
10 0.696987971599779
11 0.696987971599779
12 0.696987971599779
13 0.696987971599779
14 0.696987971599779
15 0.696987971599779
16 0.696987971599779
17 0.696987971599779
18 0.696987971599779
19 0.696987971599779
20 0.696987971599779
21 0.696987971599779
22 0.696987971599779
23 0.696987971599779
24 0.696987971599779
25 0.696987971599779
26 0.696987971599779
27 0.696987971599779
28 0.696987971599779
29 0.696987971599779
30 0.696987971599779
31 0.696987971599779
32 0.696987971599779
33 0.696987971599779
34 0.696987971599779
35 0.696987971599779
36 0.696987971599779
37 0.696987971599779
38 0.696987971599779
39 0.696987971599779
40 0.696987971599779
41 0.696987971599779
42 0.696987971599779
43 0.696987971599779
44 0.696987971599779
45 0.696987971599779
46 0.696987971599779
47 0.696987971599779
48 0.696987971599779
49 0.696987971599779
50 0.696987971599779
51 0.696987971599779
52 0.696987971599779
53 0.696987971599779
54 0.696987971599779
55 0.696987971599779
56 0.696987971599779
57 0.696987971599779
58 0.696987971599779
59 0.696987971599779
60 0.696987971599779
61 0.696987971599779
62 0.696987971599779
63 0.696987971599779
64 0.696987971599779
65 0.696987971599779
66 0.696987971599779
67 0.696987971599779
68 0.696987971599779
69 0.696987971599779
70 0.696987971599779
71 0.696987971599779
72 0.696987971599779
73 0.696987971599779
74 0.696987971599779
75 0.696987971599779
76 0.696987971599779
77 0.696987971599779
78 0.696987971599779
79 0.696987971599779
80 0.696987971599779
81 0.696987971599779
82 0.696987971599779
83 0.696987971599779
84 0.696987971599779
85 0.696987971599779
86 0.696987971599779
87 0.696987971599779
88 0.696987971599779
89 0.696987971599779
90 0.696987971599779
91 0.696987971599779
92 0.696987971599779
93 0.696987971599779
94 0.696987971599779
95 0.696987971599779
96 0.696987971599779
97 0.696987971599779
98 0.696987971599779
99 0.696987971599779
};
\addplot [ultra thick, matlabyellow]
table {%
1 0.65198441828947
2 0.659746144878763
3 0.661816891660999
4 0.664415918344018
5 0.665900375618443
6 0.667428144845486
7 0.668636845776046
8 0.669644827888729
9 0.670450092100809
10 0.671374099615117
11 0.672052379808353
12 0.672830437036739
13 0.673169289531483
14 0.673589385991984
15 0.674729019299055
16 0.674842545564207
17 0.675417563202568
18 0.675763427633503
19 0.675997350688036
20 0.676523914207424
21 0.676891137962453
22 0.677026990473331
23 0.677426455639715
24 0.677699400590042
25 0.67812293833831
26 0.678238099250847
27 0.678627069149658
28 0.678923709698079
29 0.679106671871193
30 0.679110873102095
31 0.67975841589819
32 0.679765509334775
33 0.679871103735253
34 0.680609978231009
35 0.680541883464429
36 0.680883907160991
37 0.680916050820044
38 0.680987231885418
39 0.680916253125692
40 0.68142308501845
41 0.681645780708352
42 0.681467934841037
43 0.681915274285356
44 0.682058633310687
45 0.681851358838516
46 0.682291718986974
47 0.682075220557427
48 0.682559056943892
49 0.682739672660922
50 0.682556979705778
51 0.68263680813337
52 0.683018781753281
53 0.683069793648506
54 0.683041316475929
55 0.6832132554175
56 0.683448206749071
57 0.683359022774448
58 0.683463696056173
59 0.683535941465005
60 0.683858820313002
61 0.683379236252284
62 0.684043098654391
63 0.68423537308827
64 0.683883548334846
65 0.684308281018477
66 0.684185516740213
67 0.684219963823971
68 0.684516689054894
69 0.684464507977128
70 0.684653119396062
71 0.684438271379818
72 0.684689091312559
73 0.684662623167643
74 0.684796990380848
75 0.684954444068464
76 0.684947388780264
77 0.685082652652108
78 0.68532845811813
79 0.685586925053874
80 0.68509133931528
81 0.68538052092628
82 0.685510697580336
83 0.685296090970749
84 0.685615569793739
85 0.685565410456559
86 0.685733139154771
87 0.68564536045066
88 0.6859666801767
89 0.685847670491746
90 0.685807100222689
91 0.686158342230737
92 0.685918870084882
93 0.686052320939383
94 0.68612696883711
95 0.686228524298539
96 0.686044935566488
97 0.686064958587794
98 0.686521672643734
99 0.686642932413523
};
\addlegendentry{ \small $x^*(c =  1)$}
\addplot [ultra thick, matlabyellow, dashed, forget plot]
table {%
1 0.692781526760978
2 0.692781526760978
3 0.692781526760978
4 0.692781526760978
5 0.692781526760978
6 0.692781526760978
7 0.692781526760978
8 0.692781526760978
9 0.692781526760978
10 0.692781526760978
11 0.692781526760978
12 0.692781526760978
13 0.692781526760978
14 0.692781526760978
15 0.692781526760978
16 0.692781526760978
17 0.692781526760978
18 0.692781526760978
19 0.692781526760978
20 0.692781526760978
21 0.692781526760978
22 0.692781526760978
23 0.692781526760978
24 0.692781526760978
25 0.692781526760978
26 0.692781526760978
27 0.692781526760978
28 0.692781526760978
29 0.692781526760978
30 0.692781526760978
31 0.692781526760978
32 0.692781526760978
33 0.692781526760978
34 0.692781526760978
35 0.692781526760978
36 0.692781526760978
37 0.692781526760978
38 0.692781526760978
39 0.692781526760978
40 0.692781526760978
41 0.692781526760978
42 0.692781526760978
43 0.692781526760978
44 0.692781526760978
45 0.692781526760978
46 0.692781526760978
47 0.692781526760978
48 0.692781526760978
49 0.692781526760978
50 0.692781526760978
51 0.692781526760978
52 0.692781526760978
53 0.692781526760978
54 0.692781526760978
55 0.692781526760978
56 0.692781526760978
57 0.692781526760978
58 0.692781526760978
59 0.692781526760978
60 0.692781526760978
61 0.692781526760978
62 0.692781526760978
63 0.692781526760978
64 0.692781526760978
65 0.692781526760978
66 0.692781526760978
67 0.692781526760978
68 0.692781526760978
69 0.692781526760978
70 0.692781526760978
71 0.692781526760978
72 0.692781526760978
73 0.692781526760978
74 0.692781526760978
75 0.692781526760978
76 0.692781526760978
77 0.692781526760978
78 0.692781526760978
79 0.692781526760978
80 0.692781526760978
81 0.692781526760978
82 0.692781526760978
83 0.692781526760978
84 0.692781526760978
85 0.692781526760978
86 0.692781526760978
87 0.692781526760978
88 0.692781526760978
89 0.692781526760978
90 0.692781526760978
91 0.692781526760978
92 0.692781526760978
93 0.692781526760978
94 0.692781526760978
95 0.692781526760978
96 0.692781526760978
97 0.692781526760978
98 0.692781526760978
99 0.692781526760978
};
\addplot [ultra thick, matlabgreen]
table {%
1 0.674026588485231
2 0.676669326745676
3 0.67662457701176
4 0.677799814096851
5 0.677911550649609
6 0.678385410947083
7 0.67866150558457
8 0.67885496396431
9 0.678776084211084
10 0.679039827201838
11 0.679456254163896
12 0.679324641411271
13 0.679188435501356
14 0.679549558811212
15 0.679378122470933
16 0.679449508395699
17 0.679447746858441
18 0.679651469310306
19 0.679704680173233
20 0.679778758728182
21 0.679729483976012
22 0.679827690775156
23 0.679817750139317
24 0.679872974114052
25 0.679884307624844
26 0.67991954541853
27 0.68009177171813
28 0.680251312628177
29 0.680025381942237
30 0.679979924372443
31 0.680152826690085
32 0.680042525642589
33 0.679910466635443
34 0.680133291754741
35 0.680108778115936
36 0.680056242997813
37 0.680253291294754
38 0.680276515177912
39 0.680342813910107
40 0.680459107486988
41 0.680360738447028
42 0.680364394438248
43 0.680348969236226
44 0.680381293246278
45 0.680318230910009
46 0.680346370693528
47 0.680394308551863
48 0.680320514662992
49 0.680231445160435
50 0.68042024924384
51 0.680372964514897
52 0.680494965924938
53 0.680446483921325
54 0.68033872160993
55 0.680420166666751
56 0.680630679117846
57 0.680497499541893
58 0.680654571860459
59 0.680529291266597
60 0.680599983936237
61 0.680728332067588
62 0.680678914267101
63 0.680235411185496
64 0.680576793998239
65 0.680694163201726
66 0.680781341664445
67 0.680635240124283
68 0.680569113091146
69 0.680692113076941
70 0.680698119406026
71 0.68063324159582
72 0.680720015026316
73 0.680527027070187
74 0.680548434038698
75 0.680610150851947
76 0.680701327751572
77 0.680610925937312
78 0.680693352621975
79 0.680775214349857
80 0.680679832655836
81 0.680725248403732
82 0.680683349795511
83 0.68057602137475
84 0.680689322072944
85 0.680714421230008
86 0.680586626486105
87 0.68073407207697
88 0.680696029790779
89 0.680688892431867
90 0.680770946271395
91 0.680841835760187
92 0.680716805608597
93 0.680846944260779
94 0.680715030369093
95 0.680873619684387
96 0.680750981110776
97 0.680835916910211
98 0.680825802341458
99 0.680765083991109
};
\addlegendentry{ \small $x^*(c = 10)$}
\addplot [ultra thick, matlabgreen, dashed, forget plot]
table {%
1 0.681547895753701
2 0.681547895753701
3 0.681547895753701
4 0.681547895753701
5 0.681547895753701
6 0.681547895753701
7 0.681547895753701
8 0.681547895753701
9 0.681547895753701
10 0.681547895753701
11 0.681547895753701
12 0.681547895753701
13 0.681547895753701
14 0.681547895753701
15 0.681547895753701
16 0.681547895753701
17 0.681547895753701
18 0.681547895753701
19 0.681547895753701
20 0.681547895753701
21 0.681547895753701
22 0.681547895753701
23 0.681547895753701
24 0.681547895753701
25 0.681547895753701
26 0.681547895753701
27 0.681547895753701
28 0.681547895753701
29 0.681547895753701
30 0.681547895753701
31 0.681547895753701
32 0.681547895753701
33 0.681547895753701
34 0.681547895753701
35 0.681547895753701
36 0.681547895753701
37 0.681547895753701
38 0.681547895753701
39 0.681547895753701
40 0.681547895753701
41 0.681547895753701
42 0.681547895753701
43 0.681547895753701
44 0.681547895753701
45 0.681547895753701
46 0.681547895753701
47 0.681547895753701
48 0.681547895753701
49 0.681547895753701
50 0.681547895753701
51 0.681547895753701
52 0.681547895753701
53 0.681547895753701
54 0.681547895753701
55 0.681547895753701
56 0.681547895753701
57 0.681547895753701
58 0.681547895753701
59 0.681547895753701
60 0.681547895753701
61 0.681547895753701
62 0.681547895753701
63 0.681547895753701
64 0.681547895753701
65 0.681547895753701
66 0.681547895753701
67 0.681547895753701
68 0.681547895753701
69 0.681547895753701
70 0.681547895753701
71 0.681547895753701
72 0.681547895753701
73 0.681547895753701
74 0.681547895753701
75 0.681547895753701
76 0.681547895753701
77 0.681547895753701
78 0.681547895753701
79 0.681547895753701
80 0.681547895753701
81 0.681547895753701
82 0.681547895753701
83 0.681547895753701
84 0.681547895753701
85 0.681547895753701
86 0.681547895753701
87 0.681547895753701
88 0.681547895753701
89 0.681547895753701
90 0.681547895753701
91 0.681547895753701
92 0.681547895753701
93 0.681547895753701
94 0.681547895753701
95 0.681547895753701
96 0.681547895753701
97 0.681547895753701
98 0.681547895753701
99 0.681547895753701
};
\addplot [ultra thick, matlabblue]
table {%
1 0.68490111992566
2 0.684924836399339
3 0.684853622117029
4 0.684893324264694
5 0.684924710015329
6 0.684847905784256
7 0.684863516692026
8 0.684943326278596
9 0.684854945907293
10 0.684825120675637
11 0.684787376101956
12 0.684837516444401
13 0.684890342872833
14 0.684948905492973
15 0.684843696848809
16 0.685010789781966
17 0.684777253356225
18 0.684887016820186
19 0.684921581289597
20 0.684842760908254
21 0.684839639648822
22 0.684806352139762
23 0.684875068439478
24 0.684887379240523
25 0.684871660282913
26 0.684844994368783
27 0.684761571734143
28 0.684806446961715
29 0.684745690415064
30 0.684814643335537
31 0.684842863662337
32 0.684870060611955
33 0.684841164746908
34 0.684980685567717
35 0.68481289616941
36 0.684949306798729
37 0.684956803528209
38 0.684902802284092
39 0.684811192214854
40 0.684956968739606
41 0.684850962375728
42 0.684898317207925
43 0.684880074256004
44 0.684905072987168
45 0.684918383829874
46 0.684884735957264
47 0.684854969828567
48 0.6848213902179
49 0.684903489364672
50 0.684946089785007
51 0.68488418577983
52 0.684870291819194
53 0.684836999106911
54 0.68484338193752
55 0.684875447243473
56 0.684878160775066
57 0.684846703019696
58 0.684924815375985
59 0.684883084861892
60 0.684849701602304
61 0.684817769441823
62 0.684826543029587
63 0.684821979028648
64 0.684911404152221
65 0.684728099555938
66 0.685069488702312
67 0.684819888228259
68 0.684916764739161
69 0.68482121469611
70 0.684972227379471
71 0.684793853093055
72 0.684781661547476
73 0.684864337176015
74 0.684879770878464
75 0.684983415540828
76 0.684774747033874
77 0.684964954581698
78 0.684784291949293
79 0.68486872372418
80 0.684824296132937
81 0.684779189089229
82 0.684972491384456
83 0.684798947809561
84 0.684848685878022
85 0.684728359519725
86 0.684757319054414
87 0.68483525541414
88 0.684991580964156
89 0.684817920303447
90 0.684927084130797
91 0.684850100217501
92 0.684926038142915
93 0.684902453911786
94 0.684916248482515
95 0.684800436766715
96 0.684858677224469
97 0.684945444012307
98 0.684962226219331
99 0.684937997575808
};
\addlegendentry{ \small $x^*(c =  100)$}
\addplot [ultra thick, matlabblue, dashed, forget plot]
table {%
1 0.684867750184784
2 0.684867750184784
3 0.684867750184784
4 0.684867750184784
5 0.684867750184784
6 0.684867750184784
7 0.684867750184784
8 0.684867750184784
9 0.684867750184784
10 0.684867750184784
11 0.684867750184784
12 0.684867750184784
13 0.684867750184784
14 0.684867750184784
15 0.684867750184784
16 0.684867750184784
17 0.684867750184784
18 0.684867750184784
19 0.684867750184784
20 0.684867750184784
21 0.684867750184784
22 0.684867750184784
23 0.684867750184784
24 0.684867750184784
25 0.684867750184784
26 0.684867750184784
27 0.684867750184784
28 0.684867750184784
29 0.684867750184784
30 0.684867750184784
31 0.684867750184784
32 0.684867750184784
33 0.684867750184784
34 0.684867750184784
35 0.684867750184784
36 0.684867750184784
37 0.684867750184784
38 0.684867750184784
39 0.684867750184784
40 0.684867750184784
41 0.684867750184784
42 0.684867750184784
43 0.684867750184784
44 0.684867750184784
45 0.684867750184784
46 0.684867750184784
47 0.684867750184784
48 0.684867750184784
49 0.684867750184784
50 0.684867750184784
51 0.684867750184784
52 0.684867750184784
53 0.684867750184784
54 0.684867750184784
55 0.684867750184784
56 0.684867750184784
57 0.684867750184784
58 0.684867750184784
59 0.684867750184784
60 0.684867750184784
61 0.684867750184784
62 0.684867750184784
63 0.684867750184784
64 0.684867750184784
65 0.684867750184784
66 0.684867750184784
67 0.684867750184784
68 0.684867750184784
69 0.684867750184784
70 0.684867750184784
71 0.684867750184784
72 0.684867750184784
73 0.684867750184784
74 0.684867750184784
75 0.684867750184784
76 0.684867750184784
77 0.684867750184784
78 0.684867750184784
79 0.684867750184784
80 0.684867750184784
81 0.684867750184784
82 0.684867750184784
83 0.684867750184784
84 0.684867750184784
85 0.684867750184784
86 0.684867750184784
87 0.684867750184784
88 0.684867750184784
89 0.684867750184784
90 0.684867750184784
91 0.684867750184784
92 0.684867750184784
93 0.684867750184784
94 0.684867750184784
95 0.684867750184784
96 0.684867750184784
97 0.684867750184784
98 0.684867750184784
99 0.684867750184784
};
\end{axis}

\end{tikzpicture}

%% file: human_subject_study.tex
\section{Human Subject Study for the Autonomous Car Example}

We assess the importance of inferability in repeated Stackelberg games through an in-person human subject study\footnote{This study was approved by UT IRB study \#7222.} in a simulated driving environment. The experiment closely follows the autonomous car and pedestrian interaction scenario described in Section \ref{section:carpedestrianexample}. Different from the autonomous car and pedestrian interaction, the human in the study is the driver who is controlling a car.

\subsection{Experiment Scenario}
The experiment represents the interaction at T-intersection in the CARLA simulator~\cite{dosovitskiy2017carla}. 
There are two agents in the environment, an autonomous car and a human-driven car. The human-driven car approaches the intersection on the terminating road, and the autonomous car approaches the intersection on the through road. 
A view from the environment is given in Fig. \ref{fig:turn}.
\begin{figure}[t]
    \centering
    \includegraphics[width=1\linewidth]{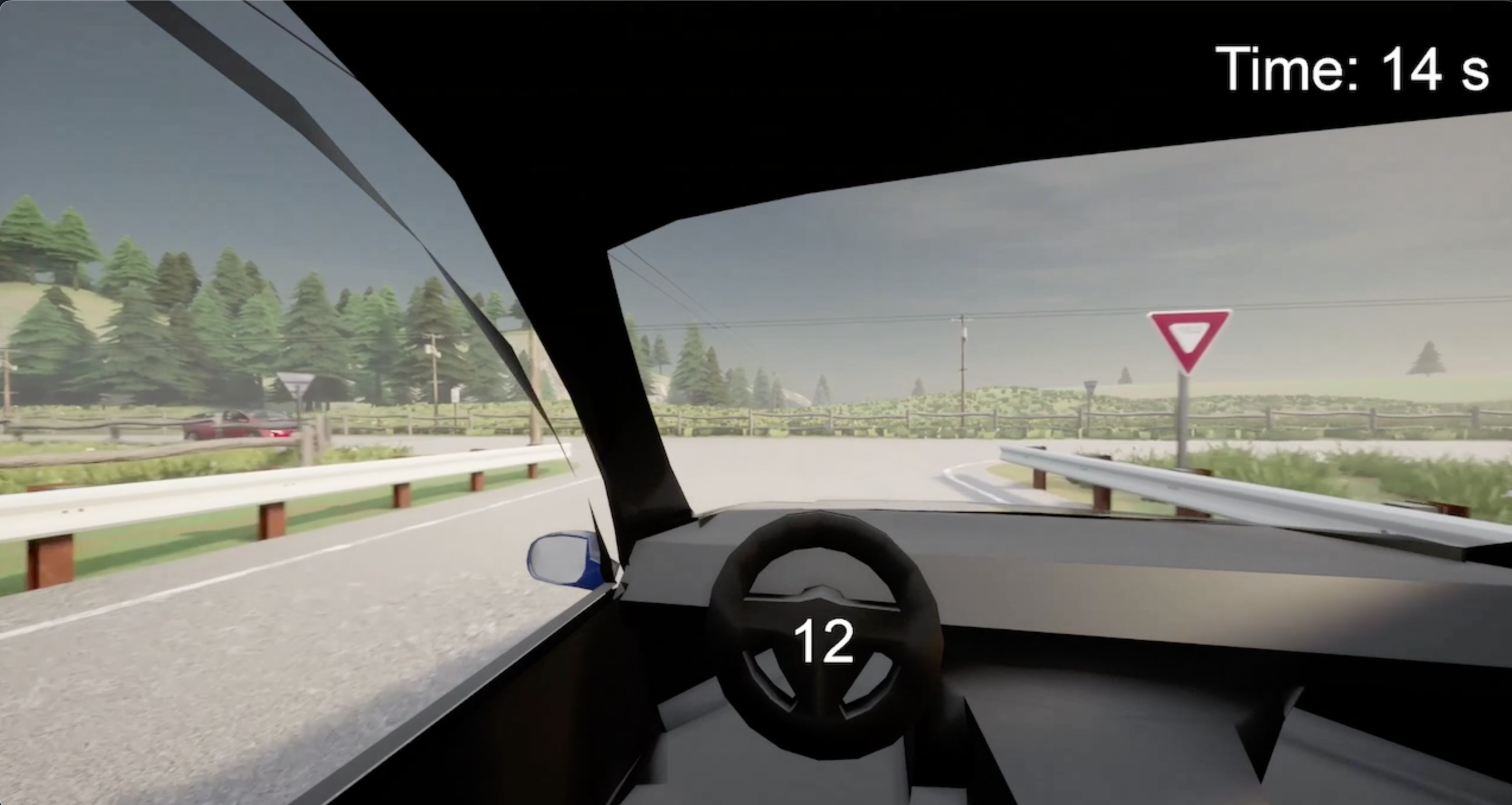}
    \caption{The simulated driving environment from the human driver's point of view as the driver approaches to the T-intersection. The autonomous car is the red car on the left.}
    \label{fig:turn}
\end{figure}

At the intersection, both the through road and the terminating road have yield signs. If the autonomous car decides to stop, it yields and gives the right of way to the human driver. If it decides not to stop, but the human also does not stop, the autonomous car makes an emergency stop to avoid a crash. 

The human-driven car is controlled using a commercially available video game driving wheel and pedals.

\subsection{Objectives and Rewards}
The objectives of the autonomous car are to pass through the intersection in minimal time and avoid an emergency stop. The instructed objectives of the human driver are completing a left turn at the intersection in minimal time and avoiding getting fined for not stopping. Based on these objectives, we give (virtual) rewards to the participants.

To observe the interaction between the players at the intersection, we do not give a reward to the participants if the turn is not completed within a certain time frame. Completing the turn in time and stopping for the autonomous car is not possible in the experiment if the autonomous car initially decides not to stop for the human driver. 

    \begin{table}[h]
    \caption{Utilities for the autonomous car and human driven car interaction}\label{tab:bimarixhss}
    \centering
\begin{tabular}{cccc}
          &                              & \multicolumn{2}{l}{Human's actions}                 \\ \cline{3-4} 
          & \multicolumn{1}{c|}{}        & \multicolumn{1}{c|}{Stop}  & \multicolumn{1}{c|}{Proceed}  \\ \cline{2-4} 
\multicolumn{1}{c|}{\multirow{2}{*}{\begin{tabular}[c]{@{}c@{}}Autonomous car's\\ actions\end{tabular}}} & \multicolumn{1}{c|}{Stop} & \multicolumn{1}{c|}{(0,2)} & \multicolumn{1}{c|}{(0,1)} \\ \cline{2-4} 
\multicolumn{1}{c|}{}                                                                         & \multicolumn{1}{c|}{Proceed}    & \multicolumn{1}{c|}{(2,0)}  & \multicolumn{1}{c|}{(-8,1)} \\ \cline{2-4} 
\end{tabular}
\end{table}

Table \ref{tab:bimarixhss} shows the (virtual) rewards given to the players when the human completes the turn in time. Similar to the example from Section \ref{section:carpedestrianexample}, the autonomous car collects a reward of $0$ if it stops for the human driver. It collects a reward of $2$ if it proceeds without stopping and the human driver stops. It collects a reward of $-8$ if it proceeds and the human decides to proceed as well because of an emergency stop. 

The participants collect a reward of $2$ if the turn is completed in time and they also stop for the autonomous car. They collect a reward of $1$ if the turn is completed in time, but they do not stop for the autonomous car. They collect a reward of $0$ if the turn is not completed in time.

As in the example from Section \ref{section:carpedestrianexample}, the optimal strategy for the autonomous car is to choose a stopping probability $p$ such that \(\genericprobability > 0.5\) and \(\genericprobability \approx 0.5\) in the full information setting. Under such a strategy, the human driver's optimal action is to stop, and the autonomous car's expected return is approximately $0.5$. 

\subsection{Experiment Setting and Independent Variables}
We recruited 24 participants with ages ranging from 19 to 31. After the participants completed a series of training turns to get familiar with the controllers, we asked each participant to interact with two types of autonomous cars that have different stopping probabilities. The types consisted of autonomous cars with stopping probabilities of 0.1, 0.25, 0.4, 0.6, 0.75, and 0.9. We randomly assigned the participants to the types. 

A participant interacted with a type 15 times, where each interaction lasted less than 45 seconds. At each interaction, the autonomous vehicle stopped with its type probability independent of the other interactions.

Before the interactions began, we instructed the participants to maximize their collected virtual rewards. After each interaction, we showed the participants the outcome of the scenario and the virtual reward.

The video recording of an interaction is available at \url{https://github.com/mustafakarabag/Inferability}.

\subsection{Dependent Variables}
Before each interaction (including the first interaction), participants filled out a survey regarding their estimations about the autonomous car. The questions and answer scales are (i) ``What is the percentage that the other car would do a courtesy stop in the next interaction?'' with allowed answers between 0 and 100 with increments of 1, and (ii) ``What is your confidence level in this estimation?'' with responses ``Not confident at all", ``Slightly confident", ``Somewhat confident", ``Fairly confident", and
                     ``Completely confident".

We collected the outcome of each interaction, including whether the autonomous car did a courtesy stop, the autonomous car did an emergency stop, the human driver stopped, and the human driver completed the turn in time. 

\subsection{Results and Discussion}

We give the results for the study in Figs. \ref{fig:estimationgraphshss}-\ref{fig:returngraphshss}. 
In Fig. \ref{fig:estimationgraphshss}, we plot the participant's estimate for the autonomous car's stopping probability after every interaction (solid lines) and the sample mean estimate for the autonomous car's stopping probability after every interaction (dashed lines). The pair of transparent lines for each color represents the data for a participant. The solid black line, A-E (Averaged estimate), represents the estimate of the autonomous car's stopping probability by humans, averaged over all participants. The dashed black line, A-E* (Averaged optimal estimates), represents the sample mean estimate of the autonomous car's stopping probability by humans, averaged over all participants.

The participants reported the following final confidence levels for the estimates after all interactions: \(4.12\text{ for }p=0.1\), \ \(  3.71 \text{ for }p=0.25\),\ \( 4.14 \text{ for } p=0.4\),\ \( 3.77 \text{ for } p=0.6\), \ \(4.22 \text{ for } p=0.75\), and \(4.28 \text{ for } p=0.9\) where  ``Not confident at all" is level $1$ and ``Completely confident" is level $5$.

\begin{figure}[!t]
  \centering
\input{plots/hss/estimations_stop_prob_0.1}
\vspace{2mm}
\input{plots/hss/estimations_stop_prob_0.25}
\vspace{2mm}
\input{plots/hss/estimations_stop_prob_0.4}
\vspace{2mm}
\input{plots/hss/estimations_stop_prob_0.6}
\vspace{2mm}
\input{plots/hss/estimations_stop_prob_0.75}
\vspace{2mm}
\input{plots/hss/estimations_stop_prob_0.9}
  \caption{Human and sample mean estimates after every interaction for different stopping probabilities, $p$. The number of participants is given with $\#$.}
  \label{fig:estimationgraphshss}
\end{figure}

\begin{figure}[!t]
  \centering
\input{plots/hss/rewards_stop_prob_0.1}
\vspace{2mm}
\input{plots/hss/rewards_stop_prob_0.25}
\vspace{2mm}
\input{plots/hss/rewards_stop_prob_0.4}
\vspace{2mm}
\input{plots/hss/rewards_stop_prob_0.6}
\vspace{2mm}
\input{plots/hss/rewards_stop_prob_0.75}
\vspace{2mm}
\input{plots/hss/rewards_stop_prob_0.9}
  \caption{Mean car return after every interaction for different stopping probabilities, $p$. The number of participants is given with $\#$. \vspace{0.5cm}}
  \label{fig:returngraphshss}
\end{figure}

We observe that the participants' estimates for the autonomous car's stopping probability $p$ converge to the actual value on average for most values of $p$. The estimates for the more deterministic values of $p$, such as 0.1 and 0.9, rapidly approached the actual value of $p$, while the estimates had higher variances for more stochastic values of $p$, indicating the significance of inferability. The averaged estimate for $p=0.1$ remains consistently higher than the averaged sample mean value due to the two outlier sets of interactions where the participants' estimates were significantly higher than 0.5, despite the sample mean estimate converging to 0.1. We also note that while the estimates approach the actual value, the initial estimates suggest an anchoring effect around 0.5 rather than immediately following the sample mean estimates. While the participants were told that the car would have an arbitrary type, the majority of the participants' initial estimate was $0.5$.

In Fig. \ref{fig:returngraphshss}, we plot the autonomous car's mean return from the first interaction to the last interaction (solid lines) and the autonomous car's hypothetical mean return if the human had acted optimally based on the previous samples after every interaction(dashed lines), and the Stackelberg return (dotted lines). The pair of transparent lines for each color represents the data for a participant. The solid black line, A-M-IR (averaged mean return under inference), represents the autonomous car's mean return, averaged over all participants. The dashed black line, A-M-IR* (averaged mean return under inference under optimal human action), represents the autonomous car's hypothetical mean return if the human had acted optimally based on the previous samples, averaged over all participants.

We observe that, as suggested by our theoretical results, the stochastic strategies that have higher returns in the full information setting suffer from an inferability gap. For example, the strategy with $p=0.6$ has a higher return in the full information setting. On the other hand, this strategy is outperformed by the strategies $p =0.75$ and $p=0.9$ in the inference setting due to its high variance, leading to estimation errors and changes in the human driver's decisions. We note that the strategy $p=0.6$ even has a negative empirical return under inference, which means that it would also be outperformed by the deterministic strategy $p=1$ that always stops.

While the strategies with $p <0.5$ have return $-8 (1-p) \in [-8, -4)$ in the full information setting, when interacting with humans, these strategies have higher empirical returns than their respective full information counterparts. This is due to two reasons. Firstly, the lack of inferability, in fact, helps with these strategies. Since the full information returns are already low, the followers' (humans') deviations from the optimal action under full information due to misinference improve the return under inference. Secondly, we observe that these strategies collect rewards (A-M-IR) that are even higher than the levels that would be achieved if the participant were to act optimally based on the previous samples (A-M-IR*). This difference is likely due to the inherent bias of the participants toward stopping, which makes the autonomous car collect a reward of 1 rather than -8 since it is more likely to proceed.

%% file: plots/hss/estimations_stop_prob_0.1.tex
\begin{tikzpicture}

\definecolor{darkgray176}{RGB}{176,176,176}
\definecolor{lightgray204}{RGB}{204,204,204}

\begin{axis}[
legend cell align={left},
legend style={
  fill opacity=0.8,
  draw opacity=1,
  text opacity=1,
  at={(0.03,0.97)},
  anchor=north west,
  draw=lightgray204
},
tick align=outside,
tick pos=left,
unbounded coords=jump,
x grid style={darkgray176},
xlabel={\small Interaction},
xmin=1, xmax=15,
xtick style={color=black},
y grid style={darkgray176},
ylabel={\small Probability},
ymin=-0.1, ymax=1.1,
ytick style={color=black},
width=1\columnwidth,
height=0.45\columnwidth,
    legend style={
      at={(1,1)},
      anchor=north east,
      legend columns=3
    }
]\addlegendimage{empty legend};
\addlegendentry{\small $p=$ 0.1}
\addlegendimage{empty legend};
\addlegendentry{\small $\#$: 8}
\addlegendimage{empty legend};
\addlegendentry{\small }
\addplot [ultra thick, white!80!blue, forget plot]
table {%
0 0
1 0.5
2 0
3 0
4 0
5 0
6 0
7 0
8 0
9 0
10 0
11 0
12 0.07
13 0.07
14 0.07
15 0.07
};
\addplot [ultra thick, white!80!blue, dash pattern=on 5.55pt off 2.4pt, forget plot]
table {%
0 0
1 nan
2 0
3 0
4 0
5 0
6 0
7 0
8 0
9 0
10 0
11 0
12 0.0909090909090909
13 0.0833333333333333
14 0.0769230769230769
15 0.0714285714285714
};
\addplot [ultra thick, white!80!olive, forget plot]
table {%
0 0
1 0.5
2 0.25
3 0
4 0.25
5 0.25
6 0
7 0
8 0
9 0
10 0
11 0
12 0
13 0.1
14 0.1
15 0.03
};
\addplot [ultra thick, white!80!olive, dash pattern=on 5.55pt off 2.4pt, forget plot]
table {%
0 0
1 nan
2 0
3 0
4 0
5 0
6 0
7 0
8 0
9 0
10 0
11 0
12 0
13 0.0833333333333333
14 0.0769230769230769
15 0.0714285714285714
};
\addplot [ultra thick, white!80!brown, forget plot]
table {%
0 0
1 0.74
2 0.32
3 0.21
4 0.21
5 0.21
6 0.21
7 0.21
8 0.21
9 0.5
10 0.66
11 0.66
12 0.66
13 0.66
14 0.66
15 0.66
};
\addplot [ultra thick, white!80!brown, dash pattern=on 5.55pt off 2.4pt, forget plot]
table {%
0 0
1 nan
2 0
3 0
4 0
5 0
6 0
7 0
8 0
9 0
10 0
11 0
12 0
13 0
14 0
15 0
};
\addplot [ultra thick, white!80!pink, forget plot]
table {%
0 0
1 0.5
2 0.5
3 0.5
4 0.5
5 0.76
6 0.79
7 0.79
8 0.79
9 0.82
10 0.82
11 0.82
12 0.82
13 0.82
14 0.82
15 0.82
};
\addplot [ultra thick, white!80!pink, dash pattern=on 5.55pt off 2.4pt, forget plot]
table {%
0 0
1 nan
2 0
3 0
4 0
5 0
6 0
7 0
8 0
9 0.125
10 0.111111111111111
11 0.1
12 0.0909090909090909
13 0.166666666666667
14 0.153846153846154
15 0.142857142857143
};
\addplot [ultra thick, white!80!purple, forget plot]
table {%
0 0
1 0.51
2 0.51
3 0.51
4 0.51
5 0.5
6 0.5
7 0.5
8 0.26
9 0.26
10 0.16
11 0.16
12 0.16
13 0.1
14 0.1
15 0.1
};
\addplot [ultra thick, white!80!purple, dash pattern=on 5.55pt off 2.4pt, forget plot]
table {%
0 0
1 nan
2 0
3 0
4 0
5 0
6 0
7 0
8 0
9 0
10 0
11 0
12 0
13 0
14 0
15 0
};
\addplot [ultra thick, white!80!green, forget plot]
table {%
0 0
1 0.5
2 0.5
3 0.5
4 0.34
5 0.34
6 0.29
7 0.29
8 0.1
9 0.04
10 0
11 0
12 0
13 0
14 0
15 0
};
\addplot [ultra thick, white!80!green, dash pattern=on 5.55pt off 2.4pt, forget plot]
table {%
0 0
1 nan
2 0
3 0
4 0
5 0
6 0
7 0
8 0
9 0
10 0
11 0
12 0
13 0
14 0
15 0
};
\addplot [ultra thick, white!80!red, forget plot]
table {%
0 0
1 0
2 0.06
3 0.25
4 0.25
5 0.15
6 0.15
7 0.1
8 0.1
9 0.1
10 0.1
11 0.15
12 0.09
13 0.06
14 0.06
15 0.06
};
\addplot [ultra thick, white!80!red, dash pattern=on 5.55pt off 2.4pt, forget plot]
table {%
0 0
1 nan
2 0
3 0.5
4 0.666666666666667
5 0.5
6 0.4
7 0.333333333333333
8 0.285714285714286
9 0.25
10 0.222222222222222
11 0.2
12 0.181818181818182
13 0.166666666666667
14 0.153846153846154
15 0.142857142857143
};
\addplot [ultra thick, white!80!orange, forget plot]
table {%
0 0
1 0.15
2 0.15
3 0.05
4 0.05
5 0.02
6 0.02
7 0.02
8 0.02
9 0.01
10 0.01
11 0
12 0
13 0
14 0
15 0
};
\addplot [ultra thick, white!80!orange, dash pattern=on 5.55pt off 2.4pt, forget plot]
table {%
0 0
1 nan
2 0
3 0
4 0
5 0
6 0
7 0
8 0
9 0
10 0
11 0
12 0
13 0
14 0
15 0
};
\addplot [ultra thick, black]
table {%
0 0
1 0.425
2 0.28625
3 0.2525
4 0.26375
5 0.27875
6 0.245
7 0.23875
8 0.185
9 0.21625
10 0.21875
11 0.22375
12 0.225
13 0.22625
14 0.22625
15 0.2175
};
\addlegendentry{\small A-E}
\addplot [ultra thick, black, dash pattern=on 7.4pt off 3.2pt]
table {%
0 0
1 nan
2 0
3 0.0625
4 0.0833333333333333
5 0.0625
6 0.05
7 0.0416666666666667
8 0.0357142857142857
9 0.046875
10 0.0416666666666667
11 0.0375
12 0.0454545454545455
13 0.0625
14 0.0576923076923077
15 0.0535714285714286
};
\addlegendentry{\small A-E*}
\end{axis}

\end{tikzpicture}

%% file: plots/hss/estimations_stop_prob_0.25.tex
\begin{tikzpicture}

\definecolor{darkgray176}{RGB}{176,176,176}
\definecolor{lightgray204}{RGB}{204,204,204}

\begin{axis}[
legend cell align={left},
legend style={fill opacity=0.8, draw opacity=1, text opacity=1, draw=lightgray204},
tick align=outside,
tick pos=left,
unbounded coords=jump,
x grid style={darkgray176},
xlabel={\small Interaction},
xmin=1, xmax=15,
xtick style={color=black},
y grid style={darkgray176},
ylabel={\small Probability},
ymin=-0.1, ymax=1.1,
ytick style={color=black},
width=1\columnwidth,
height=0.45\columnwidth,
    legend style={
      at={(1,1)},
      anchor=north east,
      legend columns=3
    }
]\addlegendimage{empty legend};
\addlegendentry{\small $p=$ 0.25}
\addlegendimage{empty legend};
\addlegendentry{\small $\#$: 7}
\addlegendimage{empty legend};
\addlegendentry{\small }
\addplot [ultra thick, white!80!blue, forget plot]
table {%
0 0
1 0
2 0.51
3 0.78
4 0.37
5 0.1
6 0.22
7 0.22
8 0.1
9 0.1
10 0.16
11 0.16
12 0.22
13 0.21
14 0.21
15 0.21
};
\addplot [ultra thick, white!80!blue, dash pattern=on 5.55pt off 2.4pt, forget plot]
table {%
0 0
1 nan
2 0
3 0
4 0
5 0
6 0.2
7 0.166666666666667
8 0.142857142857143
9 0.25
10 0.222222222222222
11 0.2
12 0.272727272727273
13 0.25
14 0.230769230769231
15 0.214285714285714
};
\addplot [ultra thick, white!80!purple, forget plot]
table {%
0 0
1 0.5
2 0.41
3 0.35
4 0.29
5 0.21
6 0.16
7 0.1
8 0.06
9 0.12
10 0.1
11 0.18
12 0.26
13 0.25
14 0.22
15 0.34
};
\addplot [ultra thick, white!80!purple, dash pattern=on 5.55pt off 2.4pt, forget plot]
table {%
0 0
1 nan
2 0
3 0
4 0
5 0
6 0
7 0
8 0
9 0.125
10 0.111111111111111
11 0.2
12 0.272727272727273
13 0.25
14 0.230769230769231
15 0.285714285714286
};
\addplot [ultra thick, white!80!brown, forget plot]
table {%
0 0
1 0.5
2 0.5
3 0
4 0
5 0
6 0
7 0
8 0.1
9 0.1
10 0.1
11 0.1
12 0.1
13 0.1
14 0.19
15 0.1
};
\addplot [ultra thick, white!80!brown, dash pattern=on 5.55pt off 2.4pt, forget plot]
table {%
0 0
1 nan
2 0
3 0
4 0
5 0
6 0
7 0
8 0.142857142857143
9 0.125
10 0.222222222222222
11 0.3
12 0.272727272727273
13 0.25
14 0.230769230769231
15 0.214285714285714
};
\addplot [ultra thick, white!80!green, forget plot]
table {%
0 0
1 0.5
2 0.5
3 0.5
4 0.5
5 0.5
6 0.4
7 0.4
8 0.32
9 0.32
10 0.32
11 0.32
12 0.32
13 0.32
14 0.32
15 0.32
};
\addplot [ultra thick, white!80!green, dash pattern=on 5.55pt off 2.4pt, forget plot]
table {%
0 0
1 nan
2 1
3 0.5
4 0.333333333333333
5 0.25
6 0.2
7 0.166666666666667
8 0.142857142857143
9 0.125
10 0.111111111111111
11 0.2
12 0.181818181818182
13 0.25
14 0.230769230769231
15 0.285714285714286
};
\addplot [ultra thick, white!80!orange, forget plot]
table {%
0 0
1 0.5
2 0.5
3 0.5
4 0.5
5 0.5
6 0.5
7 0.5
8 0.5
9 0.4
10 0.31
11 0.21
12 0.1
13 0.07
14 0.07
15 0.07
};
\addplot [ultra thick, white!80!orange, dash pattern=on 5.55pt off 2.4pt, forget plot]
table {%
0 0
1 nan
2 1
3 0.5
4 0.333333333333333
5 0.25
6 0.4
7 0.333333333333333
8 0.285714285714286
9 0.25
10 0.222222222222222
11 0.2
12 0.181818181818182
13 0.166666666666667
14 0.153846153846154
15 0.142857142857143
};
\addplot [ultra thick, white!80!red, forget plot]
table {%
0 0
1 0
2 0.5
3 0
4 0.21
5 0.4
6 0.5
7 0.5
8 0.5
9 0.5
10 0.5
11 0.5
12 0.5
13 0.5
14 0.5
15 0.5
};
\addplot [ultra thick, white!80!red, dash pattern=on 5.55pt off 2.4pt, forget plot]
table {%
0 0
1 nan
2 0
3 0
4 0.333333333333333
5 0.25
6 0.2
7 0.166666666666667
8 0.285714285714286
9 0.25
10 0.333333333333333
11 0.3
12 0.363636363636364
13 0.333333333333333
14 0.307692307692308
15 0.285714285714286
};
\addplot [ultra thick, white!80!pink, forget plot]
table {%
0 0
1 0.51
2 0.28
3 0.12
4 0.06
5 0
6 0
7 0.18
8 0.18
9 0.18
10 0.18
11 0.22
12 0.22
13 0.22
14 0.31
15 0.44
};
\addplot [ultra thick, white!80!pink, dash pattern=on 5.55pt off 2.4pt, forget plot]
table {%
0 0
1 nan
2 0
3 0
4 0
5 0
6 0
7 0.166666666666667
8 0.142857142857143
9 0.125
10 0.111111111111111
11 0.2
12 0.181818181818182
13 0.166666666666667
14 0.230769230769231
15 0.285714285714286
};
\addplot [ultra thick, black]
table {%
0 0
1 0.358571428571429
2 0.457142857142857
3 0.321428571428571
4 0.275714285714286
5 0.244285714285714
6 0.254285714285714
7 0.271428571428571
8 0.251428571428571
9 0.245714285714286
10 0.238571428571429
11 0.241428571428571
12 0.245714285714286
13 0.238571428571429
14 0.26
15 0.282857142857143
};
\addlegendentry{\small A-E}
\addplot [ultra thick, black, dash pattern=on 7.4pt off 3.2pt]
table {%
0 0
1 nan
2 0.285714285714286
3 0.142857142857143
4 0.142857142857143
5 0.107142857142857
6 0.142857142857143
7 0.142857142857143
8 0.163265306122449
9 0.178571428571429
10 0.19047619047619
11 0.228571428571429
12 0.246753246753247
13 0.238095238095238
14 0.230769230769231
15 0.244897959183673
};
\addlegendentry{\small A-E*}
\end{axis}

\end{tikzpicture}

%% file: plots/hss/estimations_stop_prob_0.4.tex
\begin{tikzpicture}

\definecolor{darkgray176}{RGB}{176,176,176}
\definecolor{lightgray204}{RGB}{204,204,204}

\begin{axis}[
legend cell align={left},
legend style={fill opacity=0.8, draw opacity=1, text opacity=1, draw=lightgray204},
tick align=outside,
tick pos=left,
unbounded coords=jump,
x grid style={darkgray176},
xlabel={\small Interaction},
xmin=1, xmax=15,
xtick style={color=black},
y grid style={darkgray176},
ylabel={\small Probability},
ymin=-0.1, ymax=1.1,
ytick style={color=black},
width=1\columnwidth,
height=0.45\columnwidth,
    legend style={
      at={(1,1)},
      anchor=north east,
      legend columns=3
    }
]\addlegendimage{empty legend};
\addlegendentry{\small $p=$ 0.4}
\addlegendimage{empty legend};
\addlegendentry{\small $\#$: 7}
\addlegendimage{empty legend};
\addlegendentry{\small }
\addplot [ultra thick, white!80!blue, forget plot]
table {%
0 0
1 0.5
2 0.25
3 0.21
4 0.21
5 0.31
6 0.43
7 0.75
8 0.5
9 0.66
10 0.51
11 0.31
12 0.31
13 0.31
14 0.31
15 0.12
};
\addplot [ultra thick, white!80!blue, dash pattern=on 5.55pt off 2.4pt, forget plot]
table {%
0 0
1 nan
2 0
3 0
4 0
5 0.25
6 0.4
7 0.5
8 0.428571428571429
9 0.5
10 0.444444444444444
11 0.4
12 0.363636363636364
13 0.333333333333333
14 0.384615384615385
15 0.357142857142857
};
\addplot [ultra thick, white!80!purple, forget plot]
table {%
0 0
1 0
2 0
3 0
4 0.1
5 0.1
6 0.07
7 0.07
8 0.07
9 0.07
10 0.07
11 0.07
12 0.07
13 0.07
14 0.07
15 0.03
};
\addplot [ultra thick, white!80!purple, dash pattern=on 5.55pt off 2.4pt, forget plot]
table {%
0 0
1 nan
2 0
3 0
4 0.333333333333333
5 0.5
6 0.4
7 0.5
8 0.571428571428571
9 0.625
10 0.555555555555556
11 0.5
12 0.545454545454545
13 0.5
14 0.461538461538462
15 0.5
};
\addplot [ultra thick, white!80!olive, forget plot]
table {%
0 0
1 0.54
2 0.38
3 0.19
4 0.19
5 0.19
6 0.19
7 0.1
8 0.1
9 0.1
10 0.1
11 0.1
12 0.1
13 0.25
14 0.25
15 0.25
};
\addplot [ultra thick, white!80!olive, dash pattern=on 5.55pt off 2.4pt, forget plot]
table {%
0 0
1 nan
2 0
3 0
4 0
5 0
6 0
7 0
8 0
9 0
10 0.111111111111111
11 0.1
12 0.181818181818182
13 0.25
14 0.307692307692308
15 0.285714285714286
};
\addplot [ultra thick, white!80!green, forget plot]
table {%
0 0
1 0.5
2 0.5
3 0.5
4 0.5
5 0.5
6 0.5
7 0.62
8 0.62
9 0.62
10 0.62
11 0.62
12 0.74
13 0.74
14 0.74
15 0.74
};
\addplot [ultra thick, white!80!green, dash pattern=on 5.55pt off 2.4pt, forget plot]
table {%
0 0
1 nan
2 1
3 0.5
4 0.333333333333333
5 0.5
6 0.6
7 0.666666666666667
8 0.571428571428571
9 0.625
10 0.555555555555556
11 0.6
12 0.636363636363636
13 0.666666666666667
14 0.615384615384615
15 0.571428571428571
};
\addplot [ultra thick, white!80!orange, forget plot]
table {%
0 0
1 0.5
2 0.28
3 0
4 0
5 0
6 0
7 0
8 0
9 0
10 0
11 0
12 0.24
13 0.1
14 0.1
15 0.1
};
\addplot [ultra thick, white!80!orange, dash pattern=on 5.55pt off 2.4pt, forget plot]
table {%
0 0
1 nan
2 0
3 0
4 0
5 0
6 0
7 0
8 0.142857142857143
9 0.125
10 0.111111111111111
11 0.2
12 0.272727272727273
13 0.25
14 0.307692307692308
15 0.285714285714286
};
\addplot [ultra thick, white!80!red, forget plot]
table {%
0 0
1 0.25
2 0.25
3 0.15
4 0.15
5 0.1
6 0.04
7 0.18
8 0.22
9 0.19
10 0.32
11 0.29
12 0.25
13 0.25
14 0.19
15 0.19
};
\addplot [ultra thick, white!80!red, dash pattern=on 5.55pt off 2.4pt, forget plot]
table {%
0 0
1 nan
2 0
3 0.5
4 0.333333333333333
5 0.25
6 0.2
7 0.333333333333333
8 0.428571428571429
9 0.375
10 0.444444444444444
11 0.4
12 0.363636363636364
13 0.333333333333333
14 0.307692307692308
15 0.285714285714286
};
\addplot [ultra thick, white!80!pink, forget plot]
table {%
0 0
1 0.51
2 0.51
3 0.5
4 0.66
5 0.57
6 0.5
7 0.5
8 0.34
9 0.25
10 0.25
11 0.15
12 0.1
13 0.1
14 0.1
15 0.1
};
\addplot [ultra thick, white!80!pink, dash pattern=on 5.55pt off 2.4pt, forget plot]
table {%
0 0
1 nan
2 0
3 0.5
4 0.666666666666667
5 0.5
6 0.4
7 0.333333333333333
8 0.285714285714286
9 0.25
10 0.222222222222222
11 0.2
12 0.181818181818182
13 0.166666666666667
14 0.230769230769231
15 0.285714285714286
};
\addplot [ultra thick, black]
table {%
0 0
1 0.4
2 0.31
3 0.221428571428571
4 0.258571428571429
5 0.252857142857143
6 0.247142857142857
7 0.317142857142857
8 0.264285714285714
9 0.27
10 0.267142857142857
11 0.22
12 0.258571428571429
13 0.26
14 0.251428571428571
15 0.218571428571429
};
\addlegendentry{\small A-E}
\addplot [ultra thick, black, dash pattern=on 7.4pt off 3.2pt]
table {%
0 0
1 nan
2 0.142857142857143
3 0.214285714285714
4 0.238095238095238
5 0.285714285714286
6 0.285714285714286
7 0.333333333333333
8 0.346938775510204
9 0.357142857142857
10 0.349206349206349
11 0.342857142857143
12 0.363636363636364
13 0.357142857142857
14 0.373626373626374
15 0.36734693877551
};
\addlegendentry{\small A-E*}
\end{axis}

\end{tikzpicture}

%% file: plots/hss/estimations_stop_prob_0.6.tex
\begin{tikzpicture}

\definecolor{darkgray176}{RGB}{176,176,176}
\definecolor{lightgray204}{RGB}{204,204,204}

\begin{axis}[
legend cell align={left},
legend style={
  fill opacity=0.8,
  draw opacity=1,
  text opacity=1,
  at={(1,1)},
  anchor=east,
  draw=lightgray204
},
tick align=outside,
tick pos=left,
unbounded coords=jump,
x grid style={darkgray176},
xlabel={\small Interaction},
xmin=1, xmax=15,
xtick style={color=black},
y grid style={darkgray176},
ylabel={\small Probability},
ymin=-0.1, ymax=1.1,
ytick style={color=black},
width=1\columnwidth,
height=0.45\columnwidth,
    legend style={
      at={(1,0.42)},
      anchor=north east,
      legend columns=3
    }
]\addlegendimage{empty legend};
\addlegendentry{\small $p=$ 0.6}
\addlegendimage{empty legend};
\addlegendentry{\small $\#$: 9}
\addlegendimage{empty legend};
\addlegendentry{\small }
\addplot [ultra thick, white!80!blue, forget plot]
table {%
0 0
1 0.5
2 1
3 0.5
4 0.32
5 0.49
6 0.41
7 0.5
8 0.65
9 0.69
10 0.76
11 0.79
12 0.81
13 0.75
14 0.71
15 0.72
};
\addplot [ultra thick, white!80!blue, dash pattern=on 5.55pt off 2.4pt, forget plot]
table {%
0 0
1 nan
2 1
3 0.5
4 0.333333333333333
5 0.5
6 0.4
7 0.333333333333333
8 0.428571428571429
9 0.5
10 0.555555555555556
11 0.6
12 0.636363636363636
13 0.583333333333333
14 0.538461538461538
15 0.571428571428571
};
\addplot [ultra thick, white!80!red, forget plot]
table {%
0 0
1 0.62
2 0.32
3 0.71
4 0.71
5 0.37
6 0.37
7 0.37
8 0.37
9 0.37
10 0.37
11 0.37
12 0.29
13 0.29
14 0.29
15 0.29
};
\addplot [ultra thick, white!80!red, dash pattern=on 5.55pt off 2.4pt, forget plot]
table {%
0 0
1 nan
2 1
3 1
4 0.666666666666667
5 0.5
6 0.6
7 0.5
8 0.571428571428571
9 0.625
10 0.555555555555556
11 0.5
12 0.454545454545455
13 0.416666666666667
14 0.384615384615385
15 0.428571428571429
};
\addplot [ultra thick, white!80!brown, forget plot]
table {%
0 0
1 0.5
2 0.5
3 0.5
4 0.5
5 0.5
6 0.5
7 0.5
8 0.5
9 0.5
10 0.5
11 0.78
12 0.78
13 0.78
14 0.74
15 0.74
};
\addplot [ultra thick, white!80!brown, dash pattern=on 5.55pt off 2.4pt, forget plot]
table {%
0 0
1 nan
2 1
3 1
4 1
5 1
6 0.8
7 0.833333333333333
8 0.857142857142857
9 0.875
10 0.888888888888889
11 0.8
12 0.818181818181818
13 0.75
14 0.692307692307692
15 0.642857142857143
};
\addplot [ultra thick, white!80!cyan, forget plot]
table {%
0 0
1 0.5
2 0.5
3 0.5
4 0.5
5 0.5
6 0.5
7 0.5
8 0.5
9 0.72
10 0.72
11 0.72
12 0.72
13 0.72
14 0.72
15 0.82
};
\addplot [ultra thick, white!80!cyan, dash pattern=on 5.55pt off 2.4pt, forget plot]
table {%
0 0
1 nan
2 0
3 0.5
4 0.666666666666667
5 0.75
6 0.6
7 0.666666666666667
8 0.714285714285714
9 0.75
10 0.777777777777778
11 0.8
12 0.727272727272727
13 0.75
14 0.769230769230769
15 0.785714285714286
};
\addplot [ultra thick, white!80!pink, forget plot]
table {%
0 0
1 0.5
2 0.25
3 0.25
4 0.15
5 0.15
6 0.25
7 0.4
8 0.5
9 0.6
10 0.75
11 0.75
12 0.75
13 0.75
14 0.75
15 0.75
};
\addplot [ultra thick, white!80!pink, dash pattern=on 5.55pt off 2.4pt, forget plot]
table {%
0 0
1 nan
2 0
3 0.5
4 0.333333333333333
5 0.5
6 0.6
7 0.666666666666667
8 0.714285714285714
9 0.75
10 0.777777777777778
11 0.7
12 0.727272727272727
13 0.75
14 0.692307692307692
15 0.714285714285714
};
\addplot [ultra thick, white!80!orange, forget plot]
table {%
0 0
1 0.71
2 0.71
3 0.71
4 0.71
5 0.71
6 0.71
7 0.71
8 0.71
9 0.71
10 0.71
11 0.66
12 0.66
13 0.66
14 0.66
15 0.66
};
\addplot [ultra thick, white!80!orange, dash pattern=on 5.55pt off 2.4pt, forget plot]
table {%
0 0
1 nan
2 0
3 0.5
4 0.666666666666667
5 0.75
6 0.6
7 0.5
8 0.428571428571429
9 0.375
10 0.333333333333333
11 0.3
12 0.363636363636364
13 0.333333333333333
14 0.307692307692308
15 0.357142857142857
};
\addplot [ultra thick, white!80!green, forget plot]
table {%
0 0
1 0.5
2 0.5
3 0.24
4 0.5
5 0.22
6 0.5
7 0.62
8 0.62
9 0.51
10 0.63
11 0.41
12 0.31
13 0.31
14 0.21
15 0.15
};
\addplot [ultra thick, white!80!green, dash pattern=on 5.55pt off 2.4pt, forget plot]
table {%
0 0
1 nan
2 0
3 0.5
4 0.666666666666667
5 0.75
6 0.8
7 0.833333333333333
8 0.714285714285714
9 0.625
10 0.666666666666667
11 0.6
12 0.545454545454545
13 0.5
14 0.461538461538462
15 0.428571428571429
};
\addplot [ultra thick, white!80!olive, forget plot]
table {%
0 0
1 0.5
2 0.5
3 0.66
4 0.66
5 0.75
6 0.69
7 0.74
8 0.81
9 0.75
10 0.71
11 0.6
12 0.6
13 0.56
14 0.56
15 0.68
};
\addplot [ultra thick, white!80!olive, dash pattern=on 5.55pt off 2.4pt, forget plot]
table {%
0 0
1 nan
2 1
3 1
4 0.666666666666667
5 0.75
6 0.6
7 0.666666666666667
8 0.714285714285714
9 0.625
10 0.555555555555556
11 0.5
12 0.454545454545455
13 0.416666666666667
14 0.384615384615385
15 0.428571428571429
};
\addplot [ultra thick, white!80!purple, forget plot]
table {%
0 0
1 0.15
2 0.6
3 0.7
4 0.75
5 0.65
6 0.65
7 0.5
8 0.4
9 0.45
10 0.45
11 0.45
12 0.4
13 0.4
14 0.4
15 0.4
};
\addplot [ultra thick, white!80!purple, dash pattern=on 5.55pt off 2.4pt, forget plot]
table {%
0 0
1 nan
2 1
3 1
4 1
5 0.75
6 0.8
7 0.666666666666667
8 0.571428571428571
9 0.625
10 0.555555555555556
11 0.6
12 0.636363636363636
13 0.666666666666667
14 0.615384615384615
15 0.642857142857143
};
\addplot [ultra thick, black]
table {%
0 0
1 0.497777777777778
2 0.542222222222222
3 0.53
4 0.533333333333333
5 0.482222222222222
6 0.508888888888889
7 0.537777777777778
8 0.562222222222222
9 0.588888888888889
10 0.622222222222222
11 0.614444444444444
12 0.591111111111111
13 0.58
14 0.56
15 0.578888888888889
};
\addlegendentry{\small A-E}
\addplot [ultra thick, black, dash pattern=on 7.4pt off 3.2pt]
table {%
0 0
1 nan
2 0.555555555555556
3 0.722222222222222
4 0.666666666666667
5 0.694444444444444
6 0.644444444444444
7 0.62962962962963
8 0.634920634920635
9 0.638888888888889
10 0.62962962962963
11 0.6
12 0.595959595959596
13 0.574074074074074
14 0.538461538461539
15 0.555555555555556
};
\addlegendentry{\small A-E*}
\end{axis}

\end{tikzpicture}

%% file: plots/hss/estimations_stop_prob_0.75.tex
\begin{tikzpicture}

\definecolor{darkgray176}{RGB}{176,176,176}
\definecolor{lightgray204}{RGB}{204,204,204}

\begin{axis}[
legend cell align={left},
legend style={
  fill opacity=0.8,
  draw opacity=1,
  text opacity=1,
  at={(0.97,0.03)},
  anchor=south east,
  draw=lightgray204
},
tick align=outside,
tick pos=left,
unbounded coords=jump,
x grid style={darkgray176},
xlabel={\small  \small Interaction},
xmin=1, xmax=15,
xtick style={color=black},
y grid style={darkgray176},
ylabel={\small Probability},
ymin=-0.1, ymax=1.1,
ytick style={color=black},
width=1\columnwidth,
height=0.45\columnwidth,
    legend style={
      at={(1,0.42)},
      anchor=north east,
      legend columns=3
    }
]\addlegendimage{empty legend};
\addlegendentry{\small $p=$ 0.75}
\addlegendimage{empty legend};
\addlegendentry{\small $\#$: 9}
\addlegendimage{empty legend};
\addlegendentry{\small }
\addplot [ultra thick, white!80!blue, forget plot]
table {%
0 0
1 0.5
2 1
3 0.5
4 0.66
5 0.75
6 0.81
7 0.66
8 0.72
9 0.63
10 0.56
11 0.6
12 0.56
13 0.6
14 0.63
15 0.6
};
\addplot [ultra thick, white!80!blue, dash pattern=on 5.55pt off 2.4pt, forget plot]
table {%
0 0
1 nan
2 1
3 0.5
4 0.666666666666667
5 0.75
6 0.8
7 0.666666666666667
8 0.714285714285714
9 0.625
10 0.555555555555556
11 0.6
12 0.545454545454545
13 0.583333333333333
14 0.615384615384615
15 0.571428571428571
};
\addplot [ultra thick, white!80!red, forget plot]
table {%
0 0
1 0.5
2 0.63
3 1
4 0.74
5 0.74
6 0.74
7 0.91
8 0.91
9 0.91
10 0.91
11 0.81
12 0.81
13 0.81
14 0.81
15 0.88
};
\addplot [ultra thick, white!80!red, dash pattern=on 5.55pt off 2.4pt, forget plot]
table {%
0 0
1 nan
2 1
3 1
4 0.666666666666667
5 0.75
6 0.8
7 0.833333333333333
8 0.857142857142857
9 0.875
10 0.888888888888889
11 0.8
12 0.818181818181818
13 0.833333333333333
14 0.846153846153846
15 0.857142857142857
};
\addplot [ultra thick, white!80!brown, forget plot]
table {%
0 0
1 0.5
2 0.5
3 0.6
4 0.7
5 0.7
6 0.75
7 0.8
8 0.8
9 0.8
10 0.8
11 0.76
12 0.76
13 0.73
14 0.73
15 0.68
};
\addplot [ultra thick, white!80!brown, dash pattern=on 5.55pt off 2.4pt, forget plot]
table {%
0 0
1 nan
2 1
3 1
4 1
5 1
6 1
7 1
8 0.857142857142857
9 0.875
10 0.888888888888889
11 0.8
12 0.818181818181818
13 0.75
14 0.692307692307692
15 0.642857142857143
};
\addplot [ultra thick, white!80!cyan, forget plot]
table {%
0 0
1 0.5
2 0.5
3 0.75
4 0.5
5 0.5
6 0.5
7 0.66
8 0.66
9 0.66
10 0.66
11 0.75
12 0.75
13 0.75
14 0.75
15 0.75
};
\addplot [ultra thick, white!80!cyan, dash pattern=on 5.55pt off 2.4pt, forget plot]
table {%
0 0
1 nan
2 1
3 1
4 0.666666666666667
5 0.5
6 0.6
7 0.666666666666667
8 0.571428571428571
9 0.625
10 0.666666666666667
11 0.7
12 0.727272727272727
13 0.75
14 0.769230769230769
15 0.714285714285714
};
\addplot [ultra thick, white!80!pink, forget plot]
table {%
0 0
1 0.5
2 0.5
3 0.75
4 0.79
5 1
6 1
7 1
8 1
9 1
10 1
11 1
12 1
13 1
14 1
15 1
};
\addplot [ultra thick, white!80!pink, dash pattern=on 5.55pt off 2.4pt, forget plot]
table {%
0 0
1 nan
2 1
3 1
4 1
5 1
6 1
7 1
8 1
9 1
10 1
11 1
12 1
13 1
14 1
15 1
};
\addplot [ultra thick, white!80!orange, forget plot]
table {%
0 0
1 0.5
2 0.5
3 0.75
4 0.56
5 0.6
6 0.6
7 0.66
8 0.75
9 0.75
10 0.8
11 0.85
12 0.7
13 0.7
14 0.7
15 0.7
};
\addplot [ultra thick, white!80!orange, dash pattern=on 5.55pt off 2.4pt, forget plot]
table {%
0 0
1 nan
2 1
3 1
4 0.666666666666667
5 0.75
6 0.8
7 0.833333333333333
8 0.857142857142857
9 0.875
10 0.888888888888889
11 0.9
12 0.818181818181818
13 0.833333333333333
14 0.846153846153846
15 0.857142857142857
};
\addplot [ultra thick, white!80!green, forget plot]
table {%
0 0
1 0.5
2 0.5
3 0.54
4 0.6
5 0.6
6 0.71
7 0.79
8 0.75
9 0.75
10 0.69
11 0.6
12 0.6
13 0.6
14 0.6
15 0.66
};
\addplot [ultra thick, white!80!green, dash pattern=on 5.55pt off 2.4pt, forget plot]
table {%
0 0
1 nan
2 1
3 1
4 1
5 1
6 1
7 1
8 0.857142857142857
9 0.875
10 0.777777777777778
11 0.7
12 0.727272727272727
13 0.75
14 0.692307692307692
15 0.714285714285714
};
\addplot [ultra thick, white!80!olive, forget plot]
table {%
0 0
1 0.65
2 0.65
3 0.72
4 0.75
5 0.75
6 0.78
7 0.81
8 0.84
9 0.84
10 0.91
11 0.91
12 0.91
13 0.91
14 0.91
15 0.94
};
\addplot [ultra thick, white!80!olive, dash pattern=on 5.55pt off 2.4pt, forget plot]
table {%
0 0
1 nan
2 1
3 1
4 1
5 1
6 1
7 1
8 1
9 1
10 1
11 1
12 1
13 1
14 1
15 1
};
\addplot [ultra thick, white!80!purple, forget plot]
table {%
0 0
1 0.5
2 0.19
3 0.49
4 0.63
5 0.76
6 0.88
7 1
8 1
9 1
10 0.9
11 0.9
12 0.9
13 0.83
14 0.85
15 0.85
};
\addplot [ultra thick, white!80!purple, dash pattern=on 5.55pt off 2.4pt, forget plot]
table {%
0 0
1 nan
2 0
3 0.5
4 0.666666666666667
5 0.75
6 0.8
7 0.833333333333333
8 0.857142857142857
9 0.875
10 0.777777777777778
11 0.8
12 0.818181818181818
13 0.75
14 0.769230769230769
15 0.785714285714286
};
\addplot [ultra thick, black]
table {%
0 0
1 0.516666666666667
2 0.552222222222222
3 0.677777777777778
4 0.658888888888889
5 0.711111111111111
6 0.752222222222222
7 0.81
8 0.825555555555555
9 0.815555555555556
10 0.803333333333333
11 0.797777777777778
12 0.776666666666667
13 0.77
14 0.775555555555556
15 0.784444444444444
};
\addlegendentry{\small A-E}
\addplot [ultra thick, black, dash pattern=on 7.4pt off 3.2pt]
table {%
0 0
1 nan
2 0.888888888888889
3 0.888888888888889
4 0.814814814814815
5 0.833333333333333
6 0.866666666666667
7 0.87037037037037
8 0.841269841269841
9 0.847222222222222
10 0.82716049382716
11 0.811111111111111
12 0.808080808080808
13 0.805555555555556
14 0.803418803418803
15 0.793650793650794
};
\addlegendentry{\small A-E*}
\end{axis}

\end{tikzpicture}

%% file: plots/hss/estimations_stop_prob_0.9.tex
\begin{tikzpicture}

\definecolor{darkgray176}{RGB}{176,176,176}
\definecolor{lightgray204}{RGB}{204,204,204}

\begin{axis}[
legend cell align={left},
legend style={
  fill opacity=0.8,
  draw opacity=1,
  text opacity=1,
  at={(0.97,0.03)},
  anchor=south east,
  draw=lightgray204
},
tick align=outside,
tick pos=left,
unbounded coords=jump,
x grid style={darkgray176},
xlabel={\small Interaction},
xmin=1, xmax=15,
xtick style={color=black},
y grid style={darkgray176},
ylabel={\small Probability},
ymin=-0.1, ymax=1.1,
ytick style={color=black},
width=1\columnwidth,
height=0.45\columnwidth,
    legend style={
      at={(1,0.42)},
      anchor=north east,
      legend columns=3
    }
]\addlegendimage{empty legend};
\addlegendentry{\small $p=$ 0.9}
\addlegendimage{empty legend};
\addlegendentry{\small $\#$: 7}
\addlegendimage{empty legend};
\addlegendentry{\small }
\addplot [ultra thick, white!80!blue, forget plot]
table {%
0 0
1 0
2 0.49
3 0.72
4 0.82
5 0.9
6 0.9
7 0.9
8 0.9
9 0.9
10 0.9
11 0.9
12 0.93
13 0.93
14 0.93
15 0.93
};
\addplot [ultra thick, white!80!blue, dash pattern=on 5.55pt off 2.4pt, forget plot]
table {%
0 0
1 nan
2 1
3 1
4 1
5 1
6 1
7 1
8 1
9 1
10 1
11 1
12 1
13 1
14 1
15 1
};
\addplot [ultra thick, white!80!purple, forget plot]
table {%
0 0
1 0.5
2 0.5
3 0.62
4 0.71
5 0.63
6 0.59
7 0.59
8 0.59
9 0.59
10 0.59
11 0.59
12 0.78
13 0.84
14 0.84
15 0.84
};
\addplot [ultra thick, white!80!purple, dash pattern=on 5.55pt off 2.4pt, forget plot]
table {%
0 0
1 nan
2 1
3 1
4 1
5 0.75
6 0.8
7 0.833333333333333
8 0.857142857142857
9 0.875
10 0.888888888888889
11 0.9
12 0.909090909090909
13 0.916666666666667
14 0.923076923076923
15 0.928571428571429
};
\addplot [ultra thick, white!80!brown, forget plot]
table {%
0 0
1 0.5
2 0.5
3 0.65
4 0.7
5 0.75
6 0.82
7 0.82
8 0.86
9 0.9
10 0.9
11 1
12 0.95
13 0.95
14 0.95
15 0.95
};
\addplot [ultra thick, white!80!brown, dash pattern=on 5.55pt off 2.4pt, forget plot]
table {%
0 0
1 nan
2 1
3 1
4 1
5 1
6 1
7 1
8 1
9 1
10 1
11 1
12 0.909090909090909
13 0.916666666666667
14 0.923076923076923
15 0.928571428571429
};
\addplot [ultra thick, white!80!red, forget plot]
table {%
0 0
1 0.6
2 0.6
3 0.6
4 0.6
5 0.6
6 0.6
7 0.6
8 0.6
9 0.81
10 0.81
11 0.81
12 0.81
13 0.81
14 0.81
15 0.81
};
\addplot [ultra thick, white!80!red, dash pattern=on 5.55pt off 2.4pt, forget plot]
table {%
0 0
1 nan
2 1
3 1
4 1
5 1
6 1
7 1
8 1
9 1
10 1
11 1
12 1
13 1
14 1
15 0.928571428571429
};
\addplot [ultra thick, white!80!green, forget plot]
table {%
0 0
1 0.5
2 0.5
3 0.5
4 0.4
5 0.5
6 0.6
7 0.65
8 0.81
9 1
10 1
11 1
12 1
13 1
14 0.85
15 0.85
};
\addplot [ultra thick, white!80!green, dash pattern=on 5.55pt off 2.4pt, forget plot]
table {%
0 0
1 nan
2 1
3 0.5
4 0.333333333333333
5 0.5
6 0.6
7 0.666666666666667
8 0.714285714285714
9 0.75
10 0.777777777777778
11 0.8
12 0.818181818181818
13 0.833333333333333
14 0.769230769230769
15 0.785714285714286
};
\addplot [ultra thick, white!80!pink, forget plot]
table {%
0 0
1 0.5
2 0.75
3 0.6
4 0.6
5 1
6 1
7 1
8 1
9 1
10 1
11 1
12 1
13 1
14 1
15 1
};
\addplot [ultra thick, white!80!pink, dash pattern=on 5.55pt off 2.4pt, forget plot]
table {%
0 0
1 nan
2 1
3 0.5
4 0.666666666666667
5 0.75
6 0.6
7 0.666666666666667
8 0.714285714285714
9 0.75
10 0.777777777777778
11 0.8
12 0.727272727272727
13 0.75
14 0.769230769230769
15 0.785714285714286
};
\addplot [ultra thick, white!80!olive, forget plot]
table {%
0 0
1 0.51
2 0.51
3 0.51
4 0.51
5 0.75
6 0.75
7 0.75
8 0.75
9 0.75
10 0.75
11 0.75
12 0.75
13 0.75
14 1
15 1
};
\addplot [ultra thick, white!80!olive, dash pattern=on 5.55pt off 2.4pt, forget plot]
table {%
0 0
1 nan
2 1
3 1
4 0.666666666666667
5 0.75
6 0.8
7 0.833333333333333
8 0.714285714285714
9 0.75
10 0.777777777777778
11 0.7
12 0.727272727272727
13 0.75
14 0.692307692307692
15 0.714285714285714
};
\addplot [ultra thick, black]
table {%
0 0
1 0.444285714285714
2 0.55
3 0.6
4 0.62
5 0.732857142857143
6 0.751428571428571
7 0.758571428571429
8 0.787142857142857
9 0.85
10 0.85
11 0.864285714285714
12 0.888571428571429
13 0.897142857142857
14 0.911428571428571
15 0.911428571428571
};
\addlegendentry{\small A-E}
\addplot [ultra thick, black, dash pattern=on 7.4pt off 3.2pt]
table {%
0 0
1 nan
2 1
3 0.857142857142857
4 0.80952380952381
5 0.821428571428571
6 0.828571428571428
7 0.857142857142857
8 0.857142857142857
9 0.875
10 0.888888888888889
11 0.885714285714286
12 0.87012987012987
13 0.880952380952381
14 0.868131868131868
15 0.86734693877551
};
\addlegendentry{\small A-E*}
\end{axis}

\end{tikzpicture}

%% file: plots/hss/rewards_stop_prob_0.1.tex
\begin{tikzpicture}

\definecolor{darkgray176}{RGB}{176,176,176}
\definecolor{lightgray204}{RGB}{204,204,204}

\begin{axis}[
legend cell align={left},
legend style={fill opacity=0.8, draw opacity=1, text opacity=1, draw=lightgray204},
tick align=outside,
tick pos=left,
x grid style={darkgray176},
xlabel={\small Interaction},
xmin=1, xmax=15,
xtick style={color=black},
y grid style={darkgray176},
ylabel={\small \small Car reward},
ymin=-10, ymax=2,
ytick style={color=black},
width=1\columnwidth,
height=0.45\columnwidth,
    legend style={
      at={(1,1)},
      anchor=north east,
      legend columns=3
    }
]
\addlegendimage{empty legend};
\addlegendentry{\small $p=$ 0.1}
\addlegendimage{empty legend};
\addlegendentry{\small $\#$: 8}
\addlegendimage{empty legend};
\addlegendentry{\small }
\addplot [ultra thick, white!80!blue, forget plot]
table {%
0 0
1 -8
2 -8
3 -8
4 -8
5 -8
6 -8
7 -8
8 -8
9 -8
10 -8
11 -7.27272727272727
12 -7.33333333333333
13 -7.38461538461539
14 -7.42857142857143
15 -7.46666666666667
};
\addplot [ultra thick, white!80!blue, dash pattern=on 5.55pt off 2.4pt, forget plot]
table {%
0 0
1 -8
2 -8
3 -8
4 -8
5 -8
6 -8
7 -8
8 -8
9 -8
10 -8
11 -7.27272727272727
12 -7.33333333333333
13 -7.38461538461539
14 -7.42857142857143
15 -7.46666666666667
};
\addplot [ultra thick, white!80!orange, forget plot]
table {%
0 0
1 1
2 1
3 -2
4 -3.5
5 -4.4
6 -5
7 -4.14285714285714
8 -3.5
9 -3
10 -2.6
11 -3.09090909090909
12 -2.83333333333333
13 -2.53846153846154
14 -2.28571428571429
15 -2.06666666666667
};
\addplot [ultra thick, white!80!orange, dash pattern=on 5.55pt off 2.4pt, forget plot]
table {%
0 0
1 1
2 -3.5
3 -5
4 -5.75
5 -6.2
6 -6.5
7 -6.71428571428571
8 -6.875
9 -7
10 -7.1
11 -7.18181818181818
12 -6.58333333333333
13 -6.69230769230769
14 -6.78571428571429
15 -6.86666666666667
};
\addplot [ultra thick, white!80!cyan, forget plot]
table {%
0 0
1 -8
2 -3.5
3 -2
4 -1.25
5 -2.6
6 -3.5
7 -4.14285714285714
8 -4.625
9 -5
10 -5.3
11 -5.54545454545455
12 -5.75
13 -5.92307692307692
14 -6.07142857142857
15 -6.2
};
\addplot [ultra thick, white!80!cyan, dash pattern=on 5.55pt off 2.4pt, forget plot]
table {%
0 0
1 -8
2 -8
3 -8
4 -8
5 -8
6 -8
7 -8
8 -8
9 -8
10 -8
11 -8
12 -8
13 -8
14 -8
15 -8
};
\addplot [ultra thick, white!80!lime, forget plot]
table {%
0 0
1 1
2 -3.5
3 -5
4 -5.75
5 -4.4
6 -5
7 -5.42857142857143
8 -4.75
9 -4.11111111111111
10 -4.5
11 -4.81818181818182
12 -4.41666666666667
13 -4.69230769230769
14 -4.92857142857143
15 -4.6
};
\addplot [ultra thick, white!80!lime, dash pattern=on 5.55pt off 2.4pt, forget plot]
table {%
0 0
1 1
2 -3.5
3 -5
4 -5.75
5 -6.2
6 -6.5
7 -6.71428571428571
8 -5.875
9 -6.11111111111111
10 -6.3
11 -6.45454545454545
12 -5.91666666666667
13 -6.07692307692308
14 -6.21428571428571
15 -5.8
};
\addplot [ultra thick, white!80!green, forget plot]
table {%
0 0
1 1
2 1
3 -2
4 -3.5
5 -2.6
6 -2
7 -2.85714285714286
8 -3.5
9 -4
10 -3.5
11 -3.90909090909091
12 -4.25
13 -4.53846153846154
14 -4.78571428571429
15 -5
};
\addplot [ultra thick, white!80!green, dash pattern=on 5.55pt off 2.4pt, forget plot]
table {%
0 0
1 1
2 -3.5
3 -5
4 -5.75
5 -6.2
6 -6.5
7 -6.71428571428571
8 -6.875
9 -7
10 -7.1
11 -7.18181818181818
12 -7.25
13 -7.30769230769231
14 -7.35714285714286
15 -7.4
};
\addplot [ultra thick, white!80!red, forget plot]
table {%
0 0
1 1
2 1
3 1
4 1
5 1
6 1
7 1
8 1
9 1
10 0.1
11 -0.636363636363636
12 -1.25
13 -1.07692307692308
14 -0.928571428571429
15 -0.8
};
\addplot [ultra thick, white!80!red, dash pattern=on 5.55pt off 2.4pt, forget plot]
table {%
0 0
1 1
2 -3.5
3 -5
4 -5.75
5 -6.2
6 -6.5
7 -6.71428571428571
8 -6.875
9 -7
10 -7.1
11 -7.18181818181818
12 -7.25
13 -7.30769230769231
14 -7.35714285714286
15 -7.4
};
\addplot [ultra thick, white!80!pink, forget plot]
table {%
0 0
1 1
2 0.5
3 0.333333333333333
4 -1.75
5 -3
6 -2.33333333333333
7 -3.14285714285714
8 -3.75
9 -4.22222222222222
10 -4.6
11 -4.90909090909091
12 -5.16666666666667
13 -5.38461538461539
14 -5.57142857142857
15 -5.73333333333333
};
\addplot [ultra thick, white!80!pink, dash pattern=on 5.55pt off 2.4pt, forget plot]
table {%
0 0
1 1
2 0.5
3 0.333333333333333
4 0.5
5 0.6
6 -0.833333333333333
7 -1.85714285714286
8 -2.625
9 -3.22222222222222
10 -3.7
11 -4.09090909090909
12 -4.41666666666667
13 -4.69230769230769
14 -4.92857142857143
15 -5.13333333333333
};
\addplot [ultra thick, white!80!olive, forget plot]
table {%
0 0
1 -8
2 -8
3 -5
4 -5.75
5 -6.2
6 -5
7 -4.14285714285714
8 -4.625
9 -5
10 -5.3
11 -5.54545454545455
12 -5.75
13 -5.23076923076923
14 -5.42857142857143
15 -5.6
};
\addplot [ultra thick, white!80!olive, dash pattern=on 5.55pt off 2.4pt, forget plot]
table {%
0 0
1 -8
2 -8
3 -8
4 -8
5 -8
6 -8
7 -8
8 -8
9 -8
10 -8
11 -8
12 -8
13 -8
14 -8
15 -8
};
\addplot [ultra thick, black]
table {%
0 0
1 -2.375
2 -2.4375
3 -2.83333333333333
4 -3.5625
5 -3.775
6 -3.72916666666667
7 -3.85714285714286
8 -3.96875
9 -4.04166666666667
10 -4.2125
11 -4.46590909090909
12 -4.59375
13 -4.59615384615385
14 -4.67857142857143
15 -4.68333333333333
};
\addlegendentry{\small A-M-IR}
\addplot [ultra thick, black, dash pattern=on 7.4pt off 3.2pt]
table {%
0 0
1 -2.375
2 -4.6875
3 -5.45833333333333
4 -5.8125
5 -6.025
6 -6.35416666666667
7 -6.58928571428571
8 -6.640625
9 -6.79166666666667
10 -6.9125
11 -6.92045454545454
12 -6.84375
13 -6.93269230769231
14 -7.00892857142857
15 -7.00833333333333
};
\addlegendentry{\small A-M-IR*}
\addplot [ultra thick, black, dash pattern=on 1.5pt off 2.475pt]
table {%
1 -7.2
15 -7.2
};
\addlegendentry{\small SR}
\end{axis}
\end{tikzpicture}

%% file: plots/hss/rewards_stop_prob_0.25.tex
\begin{tikzpicture}

\definecolor{darkgray176}{RGB}{176,176,176}
\definecolor{lightgray204}{RGB}{204,204,204}

\begin{axis}[
legend cell align={left},
legend style={
  fill opacity=0.8,
  draw opacity=1,
  text opacity=1,
  at={(0.97,0.03)},
  anchor=south east,
  draw=lightgray204
},
tick align=outside,
tick pos=left,
x grid style={darkgray176},
xlabel={\small Interaction},
xmin=1, xmax=15,
xtick style={color=black},
y grid style={darkgray176},
ylabel={\small Car reward},
ymin=-10, ymax=2,
ytick style={color=black},
width=1\columnwidth,
height=0.45\columnwidth,
    legend style={
      at={(1,1)},
      anchor=north east,
      legend columns=3
    }
]
\addlegendimage{empty legend};
\addlegendentry{\small $p=$ 0.25}
\addlegendimage{empty legend};
\addlegendentry{\small $\#$: 7}
\addlegendimage{empty legend};
\addlegendentry{\small }
\addplot [ultra thick, white!80!blue, forget plot]
table {%
0 0
1 1
2 -3.5
3 -5
4 -5.75
5 -4.6
6 -5.16666666666667
7 -5.57142857142857
8 -4.875
9 -5.22222222222222
10 -5.5
11 -5
12 -5.25
13 -5.46153846153846
14 -5.64285714285714
15 -5.8
};
\addplot [ultra thick, white!80!blue, dash pattern=on 5.55pt off 2.4pt, forget plot]
table {%
0 0
1 1
2 -3.5
3 -5
4 -5.75
5 -4.6
6 -5.16666666666667
7 -5.57142857142857
8 -4.875
9 -5.22222222222222
10 -5.5
11 -5
12 -5.25
13 -5.46153846153846
14 -5.64285714285714
15 -5.8
};
\addplot [ultra thick, white!80!red, forget plot]
table {%
0 0
1 -8
2 -3.5
3 -5
4 -3.5
5 -2.6
6 -2
7 -1.57142857142857
8 -1.375
9 -1.11111111111111
10 -1
11 -0.909090909090909
12 -0.75
13 -1.30769230769231
14 -1.21428571428571
15 -1.13333333333333
};
\addplot [ultra thick, white!80!red, dash pattern=on 5.55pt off 2.4pt, forget plot]
table {%
0 0
1 -8
2 -8
3 -8
4 -8
5 -8
6 -8
7 -8
8 -7
9 -7.11111111111111
10 -6.4
11 -5.81818181818182
12 -6
13 -6.15384615384615
14 -5.71428571428571
15 -5.33333333333333
};
\addplot [ultra thick, white!80!olive, forget plot]
table {%
0 0
1 1
2 -3.5
3 -5
4 -5.75
5 -6.2
6 -6.5
7 -5.57142857142857
8 -5.875
9 -5.22222222222222
10 -4.7
11 -5
12 -5.25
13 -5.46153846153846
14 -5.64285714285714
15 -5.8
};
\addplot [ultra thick, white!80!olive, dash pattern=on 5.55pt off 2.4pt, forget plot]
table {%
0 0
1 1
2 -3.5
3 -5
4 -5.75
5 -6.2
6 -6.5
7 -5.57142857142857
8 -5.875
9 -5.22222222222222
10 -4.7
11 -5
12 -5.25
13 -5.46153846153846
14 -5.64285714285714
15 -5.8
};
\addplot [ultra thick, white!80!orange, forget plot]
table {%
0 0
1 0
2 0.5
3 0.666666666666667
4 0.75
5 0.8
6 0.833333333333333
7 -0.428571428571429
8 -0.25
9 -1.11111111111111
10 -1
11 -1.63636363636364
12 -1.5
13 -1.30769230769231
14 -1.21428571428571
15 -1.06666666666667
};
\addplot [ultra thick, white!80!orange, dash pattern=on 5.55pt off 2.4pt, forget plot]
table {%
0 0
1 0
2 0.5
3 0.666666666666667
4 -1.5
5 -2.8
6 -3.66666666666667
7 -4.28571428571429
8 -4.75
9 -5.11111111111111
10 -4.6
11 -4.90909090909091
12 -4.5
13 -4.76923076923077
14 -4.42857142857143
15 -4.66666666666667
};
\addplot [ultra thick, white!80!pink, forget plot]
table {%
0 0
1 0
2 -4
3 -5.33333333333333
4 -3.75
5 -3
6 -3.83333333333333
7 -4.42857142857143
8 -4.875
9 -5.22222222222222
10 -5.5
11 -5.72727272727273
12 -5.91666666666667
13 -6.07692307692308
14 -6.21428571428571
15 -5.8
};
\addplot [ultra thick, white!80!pink, dash pattern=on 5.55pt off 2.4pt, forget plot]
table {%
0 0
1 0
2 0.5
3 0.666666666666667
4 -1.5
5 -1.2
6 -2.33333333333333
7 -3.14285714285714
8 -3.75
9 -4.22222222222222
10 -4.6
11 -4.90909090909091
12 -5.16666666666667
13 -5.38461538461539
14 -5.57142857142857
15 -5.2
};
\addplot [ultra thick, white!80!purple, forget plot]
table {%
0 0
1 -8
2 -8
3 -5.33333333333333
4 -6
5 -6.4
6 -6.66666666666667
7 -5.71428571428571
8 -6
9 -5.33333333333333
10 -5.6
11 -5.09090909090909
12 -4.58333333333333
13 -4.84615384615385
14 -5.07142857142857
15 -5.26666666666667
};
\addplot [ultra thick, white!80!purple, dash pattern=on 5.55pt off 2.4pt, forget plot]
table {%
0 0
1 -8
2 -8
3 -5.33333333333333
4 -6
5 -6.4
6 -6.66666666666667
7 -5.71428571428571
8 -6
9 -5.33333333333333
10 -5.6
11 -5.09090909090909
12 -5.33333333333333
13 -5.53846153846154
14 -5.71428571428571
15 -5.86666666666667
};
\addplot [ultra thick, white!80!green, forget plot]
table {%
0 0
1 1
2 1
3 1
4 1
5 1
6 0.833333333333333
7 0.857142857142857
8 0.875
9 0.888888888888889
10 0.8
11 0.818181818181818
12 0.833333333333333
13 0.769230769230769
14 0.714285714285714
15 0.733333333333333
};
\addplot [ultra thick, white!80!green, dash pattern=on 5.55pt off 2.4pt, forget plot]
table {%
0 0
1 1
2 -3.5
3 -5
4 -5.75
5 -6.2
6 -5.16666666666667
7 -5.57142857142857
8 -5.875
9 -6.11111111111111
10 -5.5
11 -5.72727272727273
12 -5.91666666666667
13 -5.46153846153846
14 -5.07142857142857
15 -5.26666666666667
};
\addplot [ultra thick, black]
table {%
0 0
1 -1.85714285714286
2 -3
3 -3.42857142857143
4 -3.28571428571429
5 -3
6 -3.21428571428571
7 -3.20408163265306
8 -3.19642857142857
9 -3.19047619047619
10 -3.21428571428571
11 -3.22077922077922
12 -3.20238095238095
13 -3.38461538461538
14 -3.46938775510204
15 -3.44761904761905
};
\addlegendentry{\small A-M-IR}
\addplot [ultra thick, black, dash pattern=on 7.4pt off 3.2pt]
table {%
0 0
1 -1.85714285714286
2 -3.64285714285714
3 -3.85714285714286
4 -4.89285714285714
5 -5.05714285714286
6 -5.35714285714286
7 -5.40816326530612
8 -5.44642857142857
9 -5.47619047619048
10 -5.27142857142857
11 -5.20779220779221
12 -5.34523809523809
13 -5.46153846153846
14 -5.39795918367347
15 -5.41904761904762
};
\addlegendentry{\small A-M-IR*}
\addplot [ultra thick, black, dash pattern=on 1.5pt off 2.475pt]
table {%
1 -6
15 -6
};
\addlegendentry{\small SR}
\end{axis}

\end{tikzpicture}

%% file: plots/hss/rewards_stop_prob_0.4.tex
\begin{tikzpicture}

\definecolor{darkgray176}{RGB}{176,176,176}
\definecolor{lightgray204}{RGB}{204,204,204}

\begin{axis}[
legend cell align={left},
legend style={
  fill opacity=0.8,
  draw opacity=1,
  text opacity=1,
  at={(0.97,0.03)},
  anchor=south east,
  draw=lightgray204
},
tick align=outside,
tick pos=left,
x grid style={darkgray176},
xlabel={\small Interaction},
xmin=1, xmax=15,
xtick style={color=black},
y grid style={darkgray176},
ylabel={\small Car reward},
ymin=-10, ymax=2,
ytick style={color=black},
width=1\columnwidth,
height=0.45\columnwidth,
    legend style={
      at={(0.622,0.42)},
      anchor=north,
      legend columns=3
    }
]
\addlegendimage{empty legend};
\addlegendentry{\small $p=$ 0.4}
\addlegendimage{empty legend};
\addlegendentry{\small $\#$: 7}
\addlegendimage{empty legend};
\addlegendentry{\small }
\addplot [ultra thick, white!80!blue, forget plot]
table {%
0 0
1 1
2 1
3 1
4 0.75
5 0.6
6 0.5
7 0.571428571428571
8 0.5
9 0.555555555555556
10 0.6
11 -0.181818181818182
12 -0.0833333333333333
13 -0.0769230769230769
14 -0.642857142857143
15 -1.13333333333333
};
\addplot [ultra thick, white!80!blue, dash pattern=on 5.55pt off 2.4pt, forget plot]
table {%
0 0
1 1
2 -3.5
3 -5
4 -3.75
5 -3
6 -2.5
7 -2
8 -1.75
9 -1.44444444444444
10 -2.1
11 -2.63636363636364
12 -3.08333333333333
13 -2.84615384615385
14 -3.21428571428571
15 -3.53333333333333
};
\addplot [ultra thick, white!80!orange, forget plot]
table {%
0 0
1 1
2 -3.5
3 -2.33333333333333
4 -1.75
5 -3
6 -2.5
7 -2.14285714285714
8 -1.875
9 -1.55555555555556
10 -1.3
11 -1.18181818181818
12 -1
13 -0.846153846153846
14 -0.785714285714286
15 -0.666666666666667
};
\addplot [ultra thick, white!80!orange, dash pattern=on 5.55pt off 2.4pt, forget plot]
table {%
0 0
1 -8
2 -8
3 -5.33333333333333
4 -4
5 -3
6 -2.5
7 -2.14285714285714
8 -1.875
9 -1.55555555555556
10 -1.3
11 -1.18181818181818
12 -1
13 -0.846153846153846
14 -0.785714285714286
15 -0.666666666666667
};
\addplot [ultra thick, white!80!olive, forget plot]
table {%
0 0
1 -8
2 -8
3 -5
4 -5.75
5 -6.2
6 -6.5
7 -6.71428571428571
8 -6.875
9 -6.11111111111111
10 -6.3
11 -5.72727272727273
12 -5.25
13 -4.84615384615385
14 -4.42857142857143
15 -4.13333333333333
};
\addplot [ultra thick, white!80!olive, dash pattern=on 5.55pt off 2.4pt, forget plot]
table {%
0 0
1 -8
2 -8
3 -8
4 -8
5 -8
6 -8
7 -8
8 -8
9 -7.11111111111111
10 -7.2
11 -6.54545454545455
12 -6
13 -5.53846153846154
14 -5.71428571428571
15 -5.33333333333333
};
\addplot [ultra thick, white!80!brown, forget plot]
table {%
0 0
1 0
2 0.5
3 0.666666666666667
4 0.5
5 0.4
6 0.333333333333333
7 0.428571428571429
8 0.375
9 0.444444444444444
10 0.4
11 0.363636363636364
12 0.333333333333333
13 0.384615384615385
14 0.428571428571429
15 0.4
};
\addplot [ultra thick, white!80!brown, dash pattern=on 5.55pt off 2.4pt, forget plot]
table {%
0 0
1 0
2 0.5
3 0.666666666666667
4 0.5
5 0.4
6 0.333333333333333
7 0.428571428571429
8 0.375
9 0.444444444444444
10 0.4
11 0.363636363636364
12 0.333333333333333
13 0.384615384615385
14 0.428571428571429
15 0.4
};
\addplot [ultra thick, white!80!red, forget plot]
table {%
0 0
1 -8
2 -3.5
3 -2
4 -3.5
5 -4.4
6 -5
7 -4.28571428571429
8 -3.625
9 -3.11111111111111
10 -2.8
11 -2.54545454545455
12 -3
13 -2.76923076923077
14 -3.14285714285714
15 -2.93333333333333
};
\addplot [ultra thick, white!80!red, dash pattern=on 5.55pt off 2.4pt, forget plot]
table {%
0 0
1 -8
2 -8
3 -8
4 -8
5 -8
6 -8
7 -6.85714285714286
8 -7
9 -7.11111111111111
10 -6.4
11 -5.81818181818182
12 -6
13 -5.53846153846154
14 -5.71428571428571
15 -5.33333333333333
};
\addplot [ultra thick, white!80!purple, forget plot]
table {%
0 0
1 1
2 0.5
3 -2.33333333333333
4 -3.75
5 -4.6
6 -3.83333333333333
7 -3.28571428571429
8 -2.75
9 -2.44444444444444
10 -2.1
11 -1.81818181818182
12 -2.33333333333333
13 -2.76923076923077
14 -3.14285714285714
15 -3.46666666666667
};
\addplot [ultra thick, white!80!purple, dash pattern=on 5.55pt off 2.4pt, forget plot]
table {%
0 0
1 1
2 0.5
3 0.666666666666667
4 -1.5
5 -2.8
6 -2.33333333333333
7 -2
8 -2.75
9 -2.44444444444444
10 -3
11 -3.45454545454545
12 -3.83333333333333
13 -4.15384615384615
14 -4.42857142857143
15 -4.66666666666667
};
\addplot [ultra thick, white!80!green, forget plot]
table {%
0 0
1 1
2 0.5
3 0.333333333333333
4 0.5
5 -1.2
6 -2.33333333333333
7 -1.85714285714286
8 -1.5
9 -1.22222222222222
10 -1.9
11 -2.45454545454545
12 -2.16666666666667
13 -2
14 -1.85714285714286
15 -2.26666666666667
};
\addplot [ultra thick, white!80!green, dash pattern=on 5.55pt off 2.4pt, forget plot]
table {%
0 0
1 1
2 0.5
3 0.333333333333333
4 0.5
5 0.6
6 -0.833333333333333
7 -1.85714285714286
8 -2.625
9 -3.22222222222222
10 -3.7
11 -4.09090909090909
12 -4.41666666666667
13 -4.07692307692308
14 -3.78571428571429
15 -4.06666666666667
};
\addplot [ultra thick, black]
table {%
0 0
1 -1.71428571428571
2 -1.78571428571429
3 -1.38095238095238
4 -1.85714285714286
5 -2.62857142857143
6 -2.76190476190476
7 -2.46938775510204
8 -2.25
9 -1.92063492063492
10 -1.91428571428571
11 -1.93506493506494
12 -1.92857142857143
13 -1.84615384615385
14 -1.93877551020408
15 -2.02857142857143
};
\addlegendentry{\small A-M-IR}
\addplot [ultra thick, black, dash pattern=on 7.4pt off 3.2pt]
table {%
0 0
1 -3
2 -3.71428571428571
3 -3.52380952380952
4 -3.46428571428571
5 -3.4
6 -3.4047619047619
7 -3.20408163265306
8 -3.375
9 -3.20634920634921
10 -3.32857142857143
11 -3.33766233766234
12 -3.42857142857143
13 -3.23076923076923
14 -3.31632653061224
15 -3.31428571428571
};
\addlegendentry{\small A-M-IR*}
\addplot [ultra thick, black, dash pattern=on 1.5pt off 2.475pt]
table {%
1 -4.8
15 -4.8
};
\addlegendentry{\small SR}
\end{axis}

\end{tikzpicture}

%% file: plots/hss/rewards_stop_prob_0.6.tex
\begin{tikzpicture}

\definecolor{darkgray176}{RGB}{176,176,176}
\definecolor{lightgray204}{RGB}{204,204,204}

\begin{axis}[
legend cell align={left},
legend style={
  fill opacity=0.8,
  draw opacity=1,
  text opacity=1,
  at={(0.97,0.03)},
  anchor=south east,
  draw=lightgray204
},
tick align=outside,
tick pos=left,
x grid style={darkgray176},
xlabel={\small Interaction},
xmin=1, xmax=15,
xtick style={color=black},
y grid style={darkgray176},
ylabel={\small Car reward},
ymin=-10, ymax=2,
ytick style={color=black},
width=1\columnwidth,
height=0.45\columnwidth,
    legend style={
      at={(0.622,0.42)},
      anchor=north,
      legend columns=3
    }
]
\addlegendimage{empty legend};
\addlegendentry{\small $p=$ 0.6}
\addlegendimage{empty legend};
\addlegendentry{\small $\#$: 9}
\addlegendimage{empty legend};
\addlegendentry{\small }
\addplot [ultra thick, white!80!blue, forget plot]
table {%
0 0
1 0
2 0.5
3 0.666666666666667
4 0.5
5 0.6
6 -0.833333333333333
7 -0.714285714285714
8 -0.625
9 -0.555555555555556
10 -0.5
11 -0.454545454545455
12 -0.333333333333333
13 -0.230769230769231
14 -0.214285714285714
15 -0.2
};
\addplot [ultra thick, white!80!blue, dash pattern=on 5.55pt off 2.4pt, forget plot]
table {%
0 0
1 0
2 0.5
3 0.666666666666667
4 0.5
5 0.6
6 -0.833333333333333
7 -0.714285714285714
8 -0.625
9 -0.555555555555556
10 -0.5
11 -0.454545454545455
12 -0.333333333333333
13 -0.230769230769231
14 -0.214285714285714
15 -0.2
};
\addplot [ultra thick, white!80!brown, forget plot]
table {%
0 0
1 0
2 0
3 0.333333333333333
4 0.5
5 0.4
6 0.5
7 0.428571428571429
8 0.375
9 -0.555555555555556
10 -1.3
11 -1.90909090909091
12 -2.41666666666667
13 -2.84615384615385
14 -2.64285714285714
15 -2.46666666666667
};
\addplot [ultra thick, white!80!brown, dash pattern=on 5.55pt off 2.4pt, forget plot]
table {%
0 0
1 0
2 0
3 0.333333333333333
4 0.5
5 0.4
6 0.5
7 0.428571428571429
8 0.375
9 0.444444444444444
10 0.5
11 0.545454545454545
12 -0.166666666666667
13 -0.769230769230769
14 -0.714285714285714
15 -0.666666666666667
};
\addplot [ultra thick, white!80!brown, forget plot]
table {%
0 0
1 0
2 0
3 0
4 0
5 0.2
6 0.166666666666667
7 0.142857142857143
8 0.125
9 0.111111111111111
10 0.2
11 0.181818181818182
12 -0.5
13 -1.07692307692308
14 -1.57142857142857
15 -1.46666666666667
};
\addplot [ultra thick, white!80!brown, dash pattern=on 5.55pt off 2.4pt, forget plot]
table {%
0 0
1 0
2 0
3 0
4 0
5 0.2
6 0.166666666666667
7 0.142857142857143
8 0.125
9 0.111111111111111
10 0.2
11 0.181818181818182
12 0.25
13 0.307692307692308
14 0.357142857142857
15 0.333333333333333
};
\addplot [ultra thick, white!80!cyan, forget plot]
table {%
0 0
1 1
2 0.5
3 0.333333333333333
4 0.25
5 0.4
6 0.333333333333333
7 0.285714285714286
8 0.25
9 0.222222222222222
10 0.2
11 0.272727272727273
12 0.25
13 0.230769230769231
14 0.214285714285714
15 -0.333333333333333
};
\addplot [ultra thick, white!80!cyan, dash pattern=on 5.55pt off 2.4pt, forget plot]
table {%
0 0
1 1
2 0.5
3 0.333333333333333
4 0.25
5 0.4
6 0.333333333333333
7 0.285714285714286
8 0.25
9 0.222222222222222
10 0.2
11 0.272727272727273
12 0.25
13 0.230769230769231
14 0.214285714285714
15 0.266666666666667
};
\addplot [ultra thick, white!80!pink, forget plot]
table {%
0 0
1 -8
2 -4
3 -2.33333333333333
4 -1.75
5 -1.4
6 -1.16666666666667
7 -1
8 -0.875
9 -0.777777777777778
10 -0.6
11 -0.545454545454545
12 -0.5
13 -1.07692307692308
14 -1
15 -0.933333333333333
};
\addplot [ultra thick, white!80!pink, dash pattern=on 5.55pt off 2.4pt, forget plot]
table {%
0 0
1 -8
2 -4
3 -2.33333333333333
4 -1.75
5 -1.4
6 -1.16666666666667
7 -1
8 -0.875
9 -0.777777777777778
10 -0.6
11 -0.545454545454545
12 -0.5
13 -0.384615384615385
14 -0.357142857142857
15 -0.333333333333333
};
\addplot [ultra thick, white!80!green, forget plot]
table {%
0 0
1 1
2 0.5
3 0.333333333333333
4 0.25
5 -1.4
6 -1
7 -0.714285714285714
8 -1.625
9 -1.33333333333333
10 -2
11 -1.81818181818182
12 -1.58333333333333
13 -2.07692307692308
14 -1.92857142857143
15 -1.8
};
\addplot [ultra thick, white!80!green, dash pattern=on 5.55pt off 2.4pt, forget plot]
table {%
0 0
1 1
2 0.5
3 0.333333333333333
4 0.25
5 0.4
6 0.5
7 0.571428571428571
8 -0.5
9 -1.33333333333333
10 -2
11 -1.81818181818182
12 -2.33333333333333
13 -2.76923076923077
14 -2.57142857142857
15 -2.4
};
\addplot [ultra thick, white!80!red, forget plot]
table {%
0 0
1 -8
2 -4
3 -2.66666666666667
4 -2
5 -1.6
6 -1.33333333333333
7 -1
8 -0.75
9 -0.666666666666667
10 -0.5
11 -0.363636363636364
12 -1
13 -1.53846153846154
14 -2
15 -2.4
};
\addplot [ultra thick, white!80!red, dash pattern=on 5.55pt off 2.4pt, forget plot]
table {%
0 0
1 -8
2 -4
3 -2.66666666666667
4 -2
5 -1.6
6 -1.33333333333333
7 -1
8 -0.75
9 -0.666666666666667
10 -0.5
11 -0.363636363636364
12 -0.25
13 -0.153846153846154
14 -0.714285714285714
15 -1.2
};
\addplot [ultra thick, white!80!purple, forget plot]
table {%
0 0
1 0
2 0
3 -2.66666666666667
4 -2
5 -1.4
6 -1.16666666666667
7 -1
8 -0.75
9 -0.555555555555556
10 -0.4
11 -1.09090909090909
12 -1.66666666666667
13 -1.46153846153846
14 -1.35714285714286
15 -1.26666666666667
};
\addplot [ultra thick, white!80!purple, dash pattern=on 5.55pt off 2.4pt, forget plot]
table {%
0 0
1 0
2 0
3 0.333333333333333
4 0.25
5 0.4
6 0.333333333333333
7 0.285714285714286
8 0.375
9 0.444444444444444
10 0.5
11 0.545454545454545
12 -0.166666666666667
13 -0.769230769230769
14 -0.714285714285714
15 -0.666666666666667
};
\addplot [ultra thick, white!80!olive, forget plot]
table {%
0 0
1 0
2 0
3 0
4 -2
5 -1.6
6 -2.66666666666667
7 -3.42857142857143
8 -3
9 -3.55555555555556
10 -3.2
11 -2.90909090909091
12 -2.66666666666667
13 -3.07692307692308
14 -2.85714285714286
15 -2.66666666666667
};
\addplot [ultra thick, white!80!olive, dash pattern=on 5.55pt off 2.4pt, forget plot]
table {%
0 0
1 0
2 0
3 0
4 0.25
5 0.2
6 0.333333333333333
7 0.428571428571429
8 0.375
9 0.444444444444444
10 0.4
11 0.363636363636364
12 0.333333333333333
13 0.384615384615385
14 0.357142857142857
15 0.333333333333333
};
\addplot [ultra thick, black]
table {%
0 0
1 -1.55555555555556
2 -0.722222222222222
3 -0.666666666666667
4 -0.694444444444444
5 -0.644444444444444
6 -0.796296296296296
7 -0.777777777777778
8 -0.763888888888889
9 -0.851851851851852
10 -0.9
11 -0.95959595959596
12 -1.15740740740741
13 -1.46153846153846
14 -1.48412698412698
15 -1.5037037037037
};
\addlegendentry{\small A-M-IR}
\addplot [ultra thick, black, dash pattern=on 7.4pt off 3.2pt]
table {%
0 0
1 -1.55555555555556
2 -0.722222222222222
3 -0.333333333333333
4 -0.194444444444444
5 -0.0444444444444444
6 -0.12962962962963
7 -0.0634920634920635
8 -0.138888888888889
9 -0.185185185185185
10 -0.2
11 -0.141414141414141
12 -0.324074074074074
13 -0.461538461538461
14 -0.484126984126984
15 -0.503703703703704
};
\addlegendentry{\small A-M-IR*}
\addplot [ultra thick, black, dash pattern=on 1.5pt off 2.475pt]
table {%
1 0.399999999999998
15 0.399999999999998
};
\addlegendentry{\small SR}
\end{axis}

\end{tikzpicture}

%% file: plots/hss/rewards_stop_prob_0.75.tex
\begin{tikzpicture}

\definecolor{darkgray176}{RGB}{176,176,176}
\definecolor{lightgray204}{RGB}{204,204,204}

\begin{axis}[
legend cell align={left},
legend style={
  fill opacity=0.8,
  draw opacity=1,
  text opacity=1,
  at={(0.03,0.03)},
  anchor=south west,
  draw=lightgray204
},
tick align=outside,
tick pos=left,
x grid style={darkgray176},
xlabel={\small Interaction},
xmin=1, xmax=15,
xtick style={color=black},
y grid style={darkgray176},
ylabel={\small Car reward},
ymin=-10, ymax=2,
ytick style={color=black},
width=1\columnwidth,
height=0.45\columnwidth,
    legend style={
      at={(0.61,0.42)},
      anchor=north,
      legend columns=3
    }
]
\addlegendimage{empty legend};
\addlegendentry{\small $p=$ 0.75}
\addlegendimage{empty legend};
\addlegendentry{\small $\#$: 9}
\addlegendimage{empty legend};
\addlegendentry{\small }
\addplot [ultra thick, white!80!blue, forget plot]
table {%
0 0
1 0
2 -4
3 -2.66666666666667
4 -2
5 -1.6
6 -1.16666666666667
7 -1
8 -0.75
9 -0.555555555555556
10 -0.5
11 -0.363636363636364
12 -0.333333333333333
13 -0.307692307692308
14 -0.214285714285714
15 -0.2
};
\addplot [ultra thick, white!80!blue, dash pattern=on 5.55pt off 2.4pt, forget plot]
table {%
0 0
1 0
2 0.5
3 0.333333333333333
4 0.25
5 0.2
6 0.333333333333333
7 0.285714285714286
8 0.375
9 0.444444444444444
10 0.4
11 0.454545454545455
12 0.416666666666667
13 0.384615384615385
14 0.428571428571429
15 0.4
};
\addplot [ultra thick, white!80!green, forget plot]
table {%
0 0
1 0
2 0
3 0.333333333333333
4 0.25
5 0.2
6 0.166666666666667
7 0.142857142857143
8 0.125
9 0.111111111111111
10 0.2
11 0.181818181818182
12 0.166666666666667
13 0.153846153846154
14 0.142857142857143
15 0.133333333333333
};
\addplot [ultra thick, white!80!green, dash pattern=on 5.55pt off 2.4pt, forget plot]
table {%
0 0
1 0
2 0
3 0.333333333333333
4 0.25
5 0.2
6 0.166666666666667
7 0.142857142857143
8 0.125
9 0.111111111111111
10 0.2
11 0.181818181818182
12 0.166666666666667
13 0.153846153846154
14 0.142857142857143
15 0.133333333333333
};
\addplot [ultra thick, white!80!olive, forget plot]
table {%
0 0
1 0
2 0
3 0
4 0
5 0
6 0
7 0.142857142857143
8 0.125
9 0.111111111111111
10 0.2
11 0.181818181818182
12 -0.5
13 -0.384615384615385
14 -0.285714285714286
15 -0.266666666666667
};
\addplot [ultra thick, white!80!olive, dash pattern=on 5.55pt off 2.4pt, forget plot]
table {%
0 0
1 0
2 0
3 0
4 0
5 0
6 0
7 0.142857142857143
8 0.125
9 0.111111111111111
10 0.2
11 0.181818181818182
12 0.25
13 0.307692307692308
14 0.357142857142857
15 0.333333333333333
};
\addplot [ultra thick, white!80!brown, forget plot]
table {%
0 0
1 0
2 0
3 0.333333333333333
4 0.5
5 0.4
6 0.333333333333333
7 0.428571428571429
8 0.375
9 0.333333333333333
10 0.3
11 0.272727272727273
12 0.25
13 0.230769230769231
14 -0.357142857142857
15 -0.266666666666667
};
\addplot [ultra thick, white!80!brown, dash pattern=on 5.55pt off 2.4pt, forget plot]
table {%
0 0
1 0
2 0
3 0.333333333333333
4 0.5
5 0.4
6 0.333333333333333
7 0.428571428571429
8 0.375
9 0.333333333333333
10 0.3
11 0.272727272727273
12 0.25
13 0.230769230769231
14 0.285714285714286
15 0.333333333333333
};
\addplot [ultra thick, white!80!pink, forget plot]
table {%
0 0
1 0
2 0
3 0
4 0
5 0
6 0
7 0
8 0
9 0
10 0
11 0
12 0
13 0
14 0
15 0
};
\addplot [ultra thick, white!80!pink, dash pattern=on 5.55pt off 2.4pt, forget plot]
table {%
0 0
1 0
2 0
3 0
4 0
5 0
6 0
7 0
8 0
9 0
10 0
11 0
12 0
13 0
14 0
15 0
};
\addplot [ultra thick, white!80!orange, forget plot]
table {%
0 0
1 0
2 0
3 -2.66666666666667
4 -2
5 -1.6
6 -1.33333333333333
7 -1.14285714285714
8 -1
9 -0.888888888888889
10 -0.8
11 -0.636363636363636
12 -0.583333333333333
13 -0.538461538461538
14 -0.5
15 -0.466666666666667
};
\addplot [ultra thick, white!80!orange, dash pattern=on 5.55pt off 2.4pt, forget plot]
table {%
0 0
1 0
2 0
3 0.333333333333333
4 0.25
5 0.2
6 0.166666666666667
7 0.142857142857143
8 0.125
9 0.111111111111111
10 0.1
11 0.181818181818182
12 0.166666666666667
13 0.153846153846154
14 0.142857142857143
15 0.133333333333333
};
\addplot [ultra thick, white!80!red, forget plot]
table {%
0 0
1 0
2 0
3 0
4 0
5 0
6 0
7 -1.14285714285714
8 -1
9 -0.777777777777778
10 -0.6
11 -0.545454545454545
12 -0.5
13 -0.384615384615385
14 -0.357142857142857
15 -0.333333333333333
};
\addplot [ultra thick, white!80!red, dash pattern=on 5.55pt off 2.4pt, forget plot]
table {%
0 0
1 0
2 0
3 0
4 0
5 0
6 0
7 0.142857142857143
8 0.125
9 0.222222222222222
10 0.3
11 0.272727272727273
12 0.25
13 0.307692307692308
14 0.285714285714286
15 0.266666666666667
};
\addplot [ultra thick, white!80!purple, forget plot]
table {%
0 0
1 0
2 0
3 0
4 0
5 0
6 0
7 0
8 0
9 0
10 0
11 0
12 0
13 0
14 0
15 0
};
\addplot [ultra thick, white!80!purple, dash pattern=on 5.55pt off 2.4pt, forget plot]
table {%
0 0
1 0
2 0
3 0
4 0
5 0
6 0
7 0
8 0
9 0
10 0
11 0
12 0
13 0
14 0
15 0
};
\addplot [ultra thick, white!80!cyan, forget plot]
table {%
0 0
1 1
2 0.5
3 0.333333333333333
4 0.25
5 0.2
6 0.166666666666667
7 0.142857142857143
8 0.125
9 0.222222222222222
10 0.2
11 0.181818181818182
12 0.25
13 0.230769230769231
14 0.214285714285714
15 0.2
};
\addplot [ultra thick, white!80!cyan, dash pattern=on 5.55pt off 2.4pt, forget plot]
table {%
0 0
1 1
2 0.5
3 0.333333333333333
4 0.25
5 0.2
6 0.166666666666667
7 0.142857142857143
8 0.125
9 0.222222222222222
10 0.2
11 0.181818181818182
12 0.25
13 0.230769230769231
14 0.214285714285714
15 0.2
};
\addplot [ultra thick, black]
table {%
0 0
1 0.111111111111111
2 -0.388888888888889
3 -0.481481481481481
4 -0.333333333333333
5 -0.266666666666667
6 -0.203703703703704
7 -0.26984126984127
8 -0.222222222222222
9 -0.160493827160494
10 -0.111111111111111
11 -0.0808080808080808
12 -0.138888888888889
13 -0.111111111111111
14 -0.150793650793651
15 -0.133333333333333
};
\addlegendentry{\small AMRUI}
\addplot [ultra thick, black, dash pattern=on 7.4pt off 3.2pt]
table {%
0 0
1 0.111111111111111
2 0.111111111111111
3 0.185185185185185
4 0.166666666666667
5 0.133333333333333
6 0.12962962962963
7 0.158730158730159
8 0.152777777777778
9 0.17283950617284
10 0.188888888888889
11 0.191919191919192
12 0.194444444444444
13 0.196581196581197
14 0.206349206349206
15 0.2
};
\addlegendentry{\small AMRUI*}
\addplot [ultra thick, black, dash pattern=on 1.5pt off 2.475pt]
table {%
1 0.249999999999998
15 0.249999999999998
};
\addlegendentry{\small SR}
\end{axis}

\end{tikzpicture}

%% file: plots/hss/rewards_stop_prob_0.9.tex
\begin{tikzpicture}

\definecolor{darkgray176}{RGB}{176,176,176}
\definecolor{lightgray204}{RGB}{204,204,204}

\begin{axis}[
legend cell align={left},
legend style={
  fill opacity=0.8,
  draw opacity=1,
  text opacity=1,
  at={(0.03,0.03)},
  anchor=south west,
  draw=lightgray204
},
tick align=outside,
tick pos=left,
x grid style={darkgray176},
xlabel={\small Interaction},
xmin=1, xmax=15,
xtick style={color=black},
y grid style={darkgray176},
ylabel={\small Car reward},
ymin=-10, ymax=2,
ytick style={color=black},
width=1\columnwidth,
height=0.45\columnwidth,
    legend style={
      at={(0.622,0.42)},
      anchor=north,
      legend columns=3
    }
]
\addlegendimage{empty legend};
\addlegendentry{\small $p=$ 0.9}
\addlegendimage{empty legend};
\addlegendentry{\small $\#$: 7}
\addlegendimage{empty legend};
\addlegendentry{\small }
\addplot [ultra thick, white!80!blue, forget plot]
table {%
0 0
1 0
2 0
3 0
4 0
5 0
6 0
7 0
8 0
9 0
10 0
11 0
12 0
13 0
14 0
15 0
};
\addplot [ultra thick, white!80!blue, dash pattern=on 5.55pt off 2.4pt, forget plot]
table {%
0 0
1 0
2 0
3 0
4 0
5 0
6 0
7 0
8 0
9 0
10 0
11 0
12 0
13 0
14 0
15 0
};
\addplot [ultra thick, white!80!green, forget plot]
table {%
0 0
1 0
2 0
3 0
4 -2
5 -1.6
6 -1.33333333333333
7 -1.14285714285714
8 -1
9 -0.888888888888889
10 -0.8
11 -0.727272727272727
12 -0.666666666666667
13 -0.615384615384615
14 -0.571428571428571
15 -0.533333333333333
};
\addplot [ultra thick, white!80!green, dash pattern=on 5.55pt off 2.4pt, forget plot]
table {%
0 0
1 0
2 0
3 0
4 0.25
5 0.2
6 0.166666666666667
7 0.142857142857143
8 0.125
9 0.111111111111111
10 0.1
11 0.0909090909090909
12 0.0833333333333333
13 0.0769230769230769
14 0.0714285714285714
15 0.0666666666666667
};
\addplot [ultra thick, white!80!pink, forget plot]
table {%
0 0
1 0
2 0
3 0
4 0
5 0
6 0
7 0
8 0
9 0
10 0
11 0.0909090909090909
12 0.0833333333333333
13 0.0769230769230769
14 0.0714285714285714
15 0.0666666666666667
};
\addplot [ultra thick, white!80!pink, dash pattern=on 5.55pt off 2.4pt, forget plot]
table {%
0 0
1 0
2 0
3 0
4 0
5 0
6 0
7 0
8 0
9 0
10 0
11 0.0909090909090909
12 0.0833333333333333
13 0.0769230769230769
14 0.0714285714285714
15 0.0666666666666667
};
\addplot [ultra thick, white!80!cyan, forget plot]
table {%
0 0
1 0
2 0
3 0
4 0
5 0
6 0
7 0
8 0
9 0
10 0
11 0
12 0
13 0
14 -0.571428571428571
15 -0.533333333333333
};
\addplot [ultra thick, white!80!cyan, dash pattern=on 5.55pt off 2.4pt, forget plot]
table {%
0 0
1 0
2 0
3 0
4 0
5 0
6 0
7 0
8 0
9 0
10 0
11 0
12 0
13 0
14 0.0714285714285714
15 0.0666666666666667
};
\addplot [ultra thick, white!80!purple, forget plot]
table {%
0 0
1 0
2 -4
3 -5.33333333333333
4 -4
5 -3.2
6 -2.66666666666667
7 -2.28571428571429
8 -2
9 -1.77777777777778
10 -1.6
11 -1.45454545454545
12 -1.33333333333333
13 -1.84615384615385
14 -1.71428571428571
15 -1.6
};
\addplot [ultra thick, white!80!purple, dash pattern=on 5.55pt off 2.4pt, forget plot]
table {%
0 0
1 0
2 0.5
3 0.666666666666667
4 0.5
5 0.4
6 0.333333333333333
7 0.285714285714286
8 0.25
9 0.222222222222222
10 0.2
11 0.181818181818182
12 0.166666666666667
13 0.230769230769231
14 0.214285714285714
15 0.2
};
\addplot [ultra thick, white!80!red, forget plot]
table {%
0 0
1 0
2 0.5
3 0.333333333333333
4 0.25
5 -1.4
6 -1.16666666666667
7 -1
8 -0.875
9 -0.777777777777778
10 -0.7
11 -1.36363636363636
12 -1.25
13 -1.15384615384615
14 -1.07142857142857
15 -1
};
\addplot [ultra thick, white!80!red, dash pattern=on 5.55pt off 2.4pt, forget plot]
table {%
0 0
1 0
2 0.5
3 0.333333333333333
4 0.25
5 0.4
6 0.333333333333333
7 0.285714285714286
8 0.25
9 0.222222222222222
10 0.2
11 0.272727272727273
12 0.25
13 0.230769230769231
14 0.214285714285714
15 0.2
};
\addplot [ultra thick, white!80!orange, forget plot]
table {%
0 0
1 0
2 0
3 0.333333333333333
4 0.25
5 0.2
6 0.166666666666667
7 0.285714285714286
8 0.25
9 0.222222222222222
10 0.3
11 0.272727272727273
12 0.25
13 0.307692307692308
14 0.285714285714286
15 0.266666666666667
};
\addplot [ultra thick, white!80!orange, dash pattern=on 5.55pt off 2.4pt, forget plot]
table {%
0 0
1 0
2 0
3 0.333333333333333
4 0.25
5 0.2
6 0.166666666666667
7 0.285714285714286
8 0.25
9 0.222222222222222
10 0.3
11 0.272727272727273
12 0.25
13 0.307692307692308
14 0.285714285714286
15 0.266666666666667
};
\addplot [ultra thick, black]
table {%
0 0
1 0
2 -0.5
3 -0.666666666666667
4 -0.785714285714286
5 -0.857142857142857
6 -0.714285714285714
7 -0.591836734693878
8 -0.517857142857143
9 -0.46031746031746
10 -0.4
11 -0.454545454545455
12 -0.416666666666667
13 -0.461538461538462
14 -0.510204081632653
15 -0.476190476190476
};
\addlegendentry{\small A-M-IR}
\addplot [ultra thick, black, dash pattern=on 7.4pt off 3.2pt]
table {%
0 0
1 0
2 0.142857142857143
3 0.19047619047619
4 0.178571428571429
5 0.171428571428571
6 0.142857142857143
7 0.142857142857143
8 0.125
9 0.111111111111111
10 0.114285714285714
11 0.12987012987013
12 0.119047619047619
13 0.131868131868132
14 0.13265306122449
15 0.123809523809524
};
\addlegendentry{\small A-M-IR*}
\addplot [ultra thick, black, dash pattern=on 1.5pt off 2.475pt]
table {%
1 0.0999999999999982
15 0.0999999999999982
};
\addlegendentry{\small SR}
\end{axis}

\end{tikzpicture}

%% file: appendix.tex
\begin{proof} [Proof of Lemma \ref{lemma:varianceofemprical}]
    Due to the linearity of expectation, we have 
    \(\mathbb{E} [ \| \emprical{\randomversion{\leaderdistribution}}_{\genericinteraction} - \leaderdistribution \|^{2} ] = \sum_{i=1}^{\leadernumofactions} \mathbb{E} [  ((\emprical{\randomversion{\leaderdistribution}}_{\genericinteraction})_{i} - (\leaderdistribution)_{i} )^2 ]. \) Note that \(\randomversion{\leadergenericaction}_{1},\ldots,\randomversion{\leadergenericaction}_{k-1}\) is sampled independently from the distribution \(\leaderdistribution\), and \(\Pr(\randomversion{\leadergenericaction}_{t} = a| \leaderdistribution)  = (\leaderdistribution)_{i}\) for \(t=1, \ldots, k-1\). Consequently, \((k-1)(\emprical{\randomversion{\leaderdistribution}}_{\genericinteraction})_{i}\) is a binomial distribution with parameters \(k-1\) and \((\leaderdistribution)_{i}\), and satisfies \[\mathbb{E} [  ((k-1)(\emprical{\randomversion{\leaderdistribution}}_{\genericinteraction})_{i} - (k-1)(\leaderdistribution)_{i} )^2 ] = (k-1) (\leaderdistribution)_{i} (1-(\leaderdistribution)_{i}).\] Due to this fact, we have \[
    \pushQED{\qed} 
    \mathbb{E} [ \| \emprical{\randomversion{\leaderdistribution}}_{\genericinteraction} - \leaderdistribution \|^{2} ] =  \frac{\nu(\leaderdistribution)^2}{\genericinteraction-1}
    .
    \qedhere
    \popQED\]
    \renewcommand{\qedsymbol}{}
\end{proof}

\begin{proof}[Proof of Lemma \ref{lemma:ycloseifxclose}]

    We first define 
\begin{align*}
    q^{\leader} &= \max_{i,j} \leaderutilitymatrix_{ij} + \min_{i,j} \leaderutilitymatrix_{ij}, \quad     q^{\follower} &= \max_{i,j} \followerutilitymatrix_{ij} + \min_{i,j} \followerutilitymatrix_{ij}.
\end{align*}

Let \(J\) be a matrix of ones. We note that
\begin{align*}
\softmax_{\rationalityconstant}(\followerutilitymatrix^{\top} \emprical{\leaderdistribution}_{k}) &= \softmax_{\rationalityconstant} \left(\left(\followerutilitymatrix - \nicefrac{\onesmatrix q^{\follower}}{2} \right) ^{\top} \emprical{\leaderdistribution}_{\genericinteraction}  \right)\\
    \softmax_{\rationalityconstant}(\followerutilitymatrix^{\top} \leaderdistribution) &= \softmax_{\rationalityconstant}\left(\left(\followerutilitymatrix - \nicefrac{\onesmatrix q^{\follower}}{2}\right)^{\top} \leaderdistribution \right)
\end{align*}
since subtracting the same constant from all elements does not change the result of the softmax function. Let \(z_{i}\) be the \(i^{\text{th}}\) column of \(\followerutilitymatrix - \nicefrac{\onesmatrix q^{\follower}}{2}\). We have
\begin{subequations}
    \begin{align}
    &\| \softmax_{\rationalityconstant}(\followerutilitymatrix^{\top} \leaderdistribution) -\softmax_{\rationalityconstant}(\followerutilitymatrix^{\top} \emprical{\leaderdistribution}_{k}) \| \nonumber 
    \\
    &= \left\| \softmax_{\rationalityconstant} \left(\left(\followerutilitymatrix - \frac{\onesmatrix q^{\follower}}{2} \right)^{\top} \leaderdistribution   \right) -\softmax_{\rationalityconstant}\left(\left(\followerutilitymatrix - \frac{\onesmatrix q^{\follower}}{2}\right)^{\top} \emprical{\leaderdistribution}_{\genericinteraction}  \right)\right\| \nonumber
    \\
    &\leq \rationalityconstant \left\| \left(\followerutilitymatrix - \frac{\onesmatrix q^{\follower}}{2} \right)^{\top} \leaderdistribution   -     \left(\followerutilitymatrix - \frac{\onesmatrix q^{\follower}}{2}\right)^{\top}  \emprical{\leaderdistribution}_{\genericinteraction}  \right\| \label{eqn:softmaxlipchitz} 
    \\
    &= \rationalityconstant \sqrt{\sum_{i= 1}^{\followernumofactions} \langle \leaderdistribution - \emprical{\leaderdistribution}_{\genericinteraction} , z_{i}\rangle^2 }
    \\
    & \leq \rationalityconstant \sqrt{\sum_{i= 1}^{\followernumofactions} \| \leaderdistribution - \emprical{\leaderdistribution}_{\genericinteraction} \|^2 \| z_{i} \|^2} \label{eqn:cauchyschwartz}
     \\ 
    &\leq  \| \leaderdistribution - \emprical{\leaderdistribution}_{\genericinteraction} \| \,  \,  \frac{ \sqrt{\leadernumofactions \followernumofactions }}{2} \label{eqn:upperboundingzij}
\end{align}
\end{subequations}
 where \eqref{eqn:softmaxlipchitz} is due to the \(\rationalityconstant\)-Lipschitzness of \(\softmax_{x}\), \eqref{eqn:cauchyschwartz} is due to Cauchy-Schwartz ineq., \eqref{eqn:upperboundingzij} is due to \(\max_{i,j} |z_{i,j}| = \nicefrac{1}{2}\).
\end{proof} 

\begin{proof} [Proof of Lemma \ref{lemma:returncloseifxclose}]
Let \(J\) be a matrix of ones, and \(z_{i}\) be the \(i^{\text{th}}\) column of \(\leaderutilitymatrix - \nicefrac{\onesmatrix q^{\follower}}{2}\). We have
\begin{subequations}
        \begin{align}
    &| \leaderdistribution^{\top} \leaderutilitymatrix \followerdistribution_{\genericinteraction} -  \leaderdistribution^{\top} \leaderutilitymatrix \followerdistribution | \nonumber
    \\
    &= \left\lvert \leaderdistribution^{\top} \leaderutilitymatrix \followerdistribution_{\genericinteraction} -  \leaderdistribution^{\top} \leaderutilitymatrix \followerdistribution -  \leaderdistribution^{\top}\frac{q^\leader}{2} ( \onesmatrix\followerdistribution_{\genericinteraction} -  \onesmatrix\followerdistribution)  \right\rvert \label{eqn:subtractJykaddJy} 
    \\
    &\leq \left\| \leaderdistribution^{\top} \left(\leaderutilitymatrix - \frac{q^\leader}{2}\onesmatrix \right) \right\|  \| \followerdistribution_{\genericinteraction} - \followerdistribution \| \label{eqn:cauchschwartz2}
    \\
     &\leq \max_{\leaderdistribution' \in \simplex^{\leadernumofactions}} \left( \sum_{i}^{\followernumofactions} (\leaderdistribution'^{\top} z_i)^{2} \right)^{1/2} \| \followerdistribution_{ \genericinteraction} - \followerdistribution\| \nonumber
    \\  
    &\leq  \left( \sum_{i}^{\followernumofactions} \max_{\leaderdistribution' \in \simplex^{\leadernumofactions}} (\leaderdistribution'^{\top} z_i)^{2} \right)^{1/2} \| \followerdistribution_{ \genericinteraction} - \followerdistribution\| \nonumber
    \\
    & \leq \left( \frac{\followernumofactions ^{2}}{4} \right) ^{1/2} \| \followerdistribution_{\genericinteraction} - \followerdistribution \| = \frac{\sqrt{\followernumofactions}}{2} \| \followerdistribution_{ \genericinteraction} - \followerdistribution\|\label{eqn:upperboundingz}
\end{align}
\end{subequations}
where \(\eqref{eqn:subtractJykaddJy}\) is due to \(\onesmatrix\followerdistribution = \onesmatrix\followerdistribution_{k}\), \eqref{eqn:cauchschwartz2} is due Cauchy-Schwartz ineq., and \eqref{eqn:upperboundingz} is due to \(\max_{j} |z_{ij}| \leq \nicefrac{1}{2}\).
\end{proof}

\begin{proof} [Proof sketch for Lemma \ref{lemma:simulation}]
    The proof follows from considering the dynamic game as a discounted MDP (since the leader's strategy is fixed and we have Assumption \ref{assumption:contraction}), where the follower's strategies in different settings in the dynamic game are different policies for the MDP. 
\end{proof}

\begin{proof} [Proof sketch for Lemma \ref{lemma:enoughestimation}]
Lemma 3 \cite{karabag2023sample} uses Theorem 2.1 of \cite{weissman2003inequalities}, the upper bound 2 on \(\varphi(p_{s})\), and the union bound to derive a similar looser bound that does not depend on the concentration of \(\leaderdistribution(s)\). The proof of Lemma \ref{lemma:enoughestimation} follows from not invoking the upper bound on \(\varphi(p_{s})\) and following the same steps as in Lemma 3 \cite{karabag2023sample}.
\end{proof}

\begin{proof}[Proof sketch for Lemma \ref{lemma:enoughsamples}]
The proof directly follows from combining Lemmas 4 and 5 from \cite{karabag2023sample}.    
\end{proof}